\documentclass[%
 prb,
 preprintnumbers,
 twocolumn,
 superscriptaddress,
 amsmath, amssymb,
 aps,
 longbibliography
]{revtex4-2}
\usepackage{amsmath}
\usepackage{amssymb}
\usepackage{amsthm}

\usepackage{graphicx}
\usepackage{dcolumn}
\usepackage{booktabs,caption,subcaption,amsfonts,dcolumn}
\usepackage{multirow,array,ragged2e}
\usepackage{diagbox}
\usepackage{graphicx}
\usepackage[english]{babel}
\usepackage{color}   
\usepackage{bbm}
\usepackage{tabularx}
\usepackage{algorithm}
\usepackage{algpseudocode}
\usepackage{bm}
\usepackage{tikz}
\usepackage{tkz-euclide}
\usepackage{physics}
\usetikzlibrary{positioning}
\usetikzlibrary{decorations.markings}
\usepackage[pdfpagelabels,plainpages=false,bookmarks=true,colorlinks,linkcolor=red,urlcolor=blue,citecolor=blue]{hyperref}
\usepackage{braket}

\newcolumntype{C}{>{\centering\arraybackslash}X}
\newcolumntype{R}{>{\raggedleft\arraybackslash}X}

\newcommand{\iMML}{\texttt{iMM-Local}}
\newcommand{\iMMC}{\texttt{iMM-CG}}
\newcommand{\iMMP}{\texttt{iMM-Polar}}
\newcommand{\iMMM}{\texttt{iMM-MPO}}

\renewcommand{\Re}{\text{Re}}
\renewcommand{\Im}{\text{Im}}
\newcommand{\tTr}{\text{tTr}}
\renewcommand{\max}{\text{max}}
\renewcommand{\epsilon}{\varepsilon}

\newcommand{\mzhide}[1]{}

\DeclareMathOperator*{\argmax}{argmax} 

\tikzset{midarrow/.style={
    decoration={markings,
        mark= at position 0.625 with {\arrow{#1}} ,
    },
    postaction={decorate}
  }
}

\tikzset{beginarrow/.style={
    decoration={markings,
        mark= at position 0.375 with {\arrow{#1}} ,
    },
    postaction={decorate}
  }
}

\tikzset{endarrow/.style={
    decoration={markings,
        mark= at position 0.875 with {\arrow{#1}} ,
    },
    postaction={decorate}
  }
}

\newtheorem{thm}{Theorem}

\begin{document}

\title{Two Dimensional Isometric Tensor Networks on an Infinite Strip}
\author{Yantao Wu}
\email{yantaow@berkeley.edu}
\affiliation{%
RIKEN iTHEMS, Wako, Saitama 351-0198, Japan
}
\affiliation{%
Department of Physics, University of California, Berkeley, CA 94720, USA
}%

\author{Sajant Anand}
\affiliation{%
Department of Physics, University of California, Berkeley, CA 94720, USA
}%

\author{Sheng-Hsuan Lin}
\affiliation{%
Department of Physics, TFK, Technische Universit{\"a}t M{\"u}nchen,
James-Franck-Stra{\ss}e 1, 85748 Garching, Germany
}%

\author{Frank Pollmann}
\affiliation{%
Department of Physics, TFK, Technische Universit{\"a}t M{\"u}nchen,
James-Franck-Stra{\ss}e 1, 85748 Garching, Germany
}%

\author{Michael P. Zaletel}
\affiliation{%
Department of Physics, University of California, Berkeley, CA 94720, USA
}%
\affiliation{%
Material Science Division, Lawrence Berkeley National Laboratory, Berkeley, CA 94720, USA
}%

\date{\today}

\begin{abstract}
The exact contraction of a generic two-dimensional (2D) tensor network state (TNS) is known to be exponentially hard, making simulation of 2D systems difficult.
The recently introduced class of isometric TNS (isoTNS) represents a subset of TNS that allows for efficient simulation of such systems on finite square lattices . 
The isoTNS ansatz requires the identification of an ``orthogonality column" of tensors, within which one-dimensional matrix product state (MPS) methods can be used for calculation of observables and optimization of tensors.
Here we extend isoTNS to infinitely long strip geometries and introduce an infinite version of the Moses Move algorithm for moving the orthogonality column around the network.
Using this algorithm, we iteratively transform an infinite MPS representation of a 2D quantum state into a strip isoTNS and investigate the entanglement properties of the resulting state.
In addition, we demonstrate that the local observables can be evaluated efficiently.
Finally, we introduce an infinite time-evolving block decimation algorithm (iTEBD\textsuperscript{2}) and use it to approximate the ground state of the 2D transverse field Ising model on lattices of infinite strip geometry.
\end{abstract}

\preprint{RIKEN-iTHEMS-Report-22}
\maketitle

\tableofcontents


\section{Introduction}
\label{sec:intro}

The simulation of strongly interacting systems is a fundamental problem in quantum-many body physics.
Tensor network state (TNS) methods~\cite{cirac2021review} provide a controllable and unbiased approach for approximating the exponential complexity of the wavefunction in a manner suitable for computation.
Matrix product states (MPS) are one-dimensional TNS and have proven very successful in numerically and analytically investigating quantum many-body systems in one-dimension~\cite{fannes1992abundance,schollwock2011density,bridgeman2017hand}.
They have been used to find ground states and excited states, study quantum dynamics, and calculate the response to external probes~\cite{schollwock2011density,paeckel2019time,vanderstraeten2019tangent,cirac2021review}.
The stability and efficiency of many MPS algorithms, such as the density matrix renormalization group (DMRG)~\cite{white1992density,white1993density}, time-evolution block-decimation (TEBD)~\cite{vidal2003efficient,vidal2004efficient}, and time-dependent variational principle (TDVP)~\cite{haegeman2011tdvp,haegeman2016tdvp}, rely on the fact that any MPS has an exact isometric form~\footnote{The isometric form is the first of two conditions in the better known canonical form. An injective MPS is in canonical form if (i, isometry) $\sum_{i,\sigma} A^{\sigma}{i,j} \Bar{A}^{\sigma}_{i,k} = \delta_{j,k}$ and (ii, diagonal density matrix) $\sum_{j,j',\sigma} A^{\sigma}{i,j} \Bar{A}^{\sigma}_{i',j'} \Lambda_{j,j'} = \Lambda'_{i,i'}$, where $\Lambda, \Lambda'$ are both positive, diagonal matrices.}\cite{perez2006matrix}.
Additionally this isometric form enables efficient evaluation of expectation values, regardless of the size of the system.

Tensor network states~\cite{richter1995construction,niggemann1997quantum,sierra1998density,nishino1998density}, a.k.a. projected entangled pair states (PEPS)~\cite{verstraete2004renormalization}, are a natural generalization of MPS to lattices in two-dimensions and have been studied extensively in the last decade~\cite{cirac2021review}.
While algorithms have been introduced to determine ground states and calculate expectation values for both finite~\cite{verstraete2004renormalization,pivzorn2010fermionic,lubasch2014algorithms,liu2017gradient,liu2021accurate,vieijra2021direct} and infinite~\cite{phien2015infinite,vanderstraeten2016iPEPS,corboz2016iPEPS,liao2019differentiable} geometries, such methods inherently rely on costly approximations of the exact contraction~\cite{schuch2007computational,scarpa2020limitations,haferkamp2020hard,lubasch2014algorithms,orus2009simulation,fishman2017faster,evenbly2015TNR}. 
Motivated by the benefits of the isometric form in MPS algorithms, isometric tensor network states (isoTNS) were introduced in higher dimensions~\cite{zaletel2020isometric,haghshenas2019conversion,hyatt2019dmrg,tepaske20213D}.
These networks impose isometric constraints on the tensors, generalizing the orthogonality center of an MPS to an orthogonality hypersurface.
Such isometric TNS trade representational power for more efficient and principled optimization.

As it is not known whether or how well a generic TNS can be approximated by an isoTNS, one needs to both study the classes of physical states that can be efficiently represented by isoTNS and develop efficient numerical algorithms for isoTNS.
It was shown that 2D isoTNS on a honeycomb lattice can exactly represent the fixed points of string-net liquids and states reached by finite depth circuit perturbations~\cite{soejima2020isometric}. 
Manipulating such networks on a square lattice relies on the \textit{Moses Move} (MM)~\footnote{The name comes from the biblical story of Moses splitting the Red Sea.} algorithm for approximately splitting a two-sided MPS $\ket{\Psi}$ into the combination of an isometric tensor network operator (isoTNO) $A$ and a normalized two-sided MPS $\ket{\Phi}$, i.e. $\ket{\Psi} \approx A \ket{\Phi}$. 
This is a 2D generalization of the QR-algorithm~\footnote{Strictly speaking, it is not a QR decomposition as we do not require $\ket{\Phi}$ to be upper-triangular.} and was introduced in~\cite{zaletel2020isometric}.
The MM algorithm allows for efficient evaluation of observables without the need for boundary contraction approximations, as evaluating expectation values of operators contained within the orthogonality hypersurface reduces to an MPS problem.
Meanwhile, optimization of finite isoTNS for the ground states of 2D models has been demonstrated both with 2D generalizations of imaginary TEBD and DMRG, denoted TEBD\textsuperscript{2} and DMRG\textsuperscript{2}, respectively~\cite{zaletel2020isometric,lin2022efficient}.
These methods benefit from the isometric conditions placed on the network; the orthogonality center is moved to the tensor being updated, thus yielding well-conditioned optimization in an orthonormal basis.

MPS are often used to study quasi-2D systems with infinite lengths and finite widths, such as an infinite strip (cylinder) when open (periodic) boundary conditions are used for the finite dimension.
Since an MPS is inherently a 1D ansatz targeting 1D area-law states, the complexity of such quasi-2D MPS simulation grows exponentially with the width~\cite{iregui2017cylinders, motruk2016xk}, necessitating the extrapolation of results from strips of modest width.
To avoid this exponential scaling yet retain the benefits of the isometric form to study quasi-2D systems, we are motivated to generalize the finite isoTNS and the MM algorithm on a square lattice to infinite strip geometries.

The paper is structured as follows.
We consider the generalization of the finite isoTNS on a square lattice to infinite strip geometries in Sec.~\ref{sec:isoTNS} and extend the MM algorithm to the thermodynamic limit in Sec.~\ref{sec:iMM}.
This infinite Moses Move (iMM) algorithm splits a two-sided infinite MPS (iMPS) into an infinite isoTNO and a normalized iMPS, which are all translationally invariant.
In Sec.~\ref{sec:iMML}-Sec.~\ref{sec:iMMC}, we propose and compare four different methods for solving the splitting problem:  (i) repeated application of the finite, local MM algorithm; (ii-iii) two iterative update methods optimizing over different objective functions; and (iv) a conjugate gradient optimization maximizing the overlap.
In Sec.~\ref{sec:error}, we analyze the errors of this splitting procedure and introduce a structure theorem which clarifies the assignment of bond dimensions to the isoTNO and iMPS produced by iMM.
We then benchmark the various iMM algorithms in Sec.~\ref{sec:benchmarks}. 
As the first application of iMM algorithm, in Sec.~\ref{sec:iDMRG_to_isoTNS}, we iteratively transform an iMPS found by iDMRG into a strip isoTNS by repeated applications of the iMM algorithm to peel off columns. 
In Sec.~\ref{sec:evaluation_energy}, we then show that one can evaluate the expectation value of physical observables, e.g., energy, efficiently by utilizing iMM algorithm.
We compare the iMM approach to different methods for evaluating expectation value, including boundary MPO contractions methods. 
As a highlight of applications using iMM, we develop a TEBD\textsuperscript{2} algorithm and investigate its performance via imaginary time evolution for finding the ground state (GS) of the 2D transverse field Ising model in Sec.~\ref{sec:iTEBD2}.
We end with the discussion and outlook in Sec.~\ref{sec:conclusion}.

\section{Isometric TNS and Moses Move}
In this section we briefly review the finite isoTNS before discussing the generalization of the isoTNS to infinite strip geometries. A more detailed review of isometric tensor networks in both one- and two-dimensions as well as technical details of the finite MM algorithm can be found in~\cite{lin2022efficient}.

\label{sec:isoTNS}

\subsection{Finite isoTNS}
\label{sec:isoTNS_finite}

On a 2D square lattice, an isoTNS can be represented diagrammatically as  
\begin{equation}
    \label{eq:iso_finite}
    \ket{\Psi} = \begin{tikzpicture}[baseline = (X.base),every node/.style={scale=1.0},scale=1.0]
\newcommand{\LL}{3}      
\newcommand{\ic}{1}
\newcommand{\jc}{1}
\renewcommand{\d}{1.0}   
\renewcommand{\r}{0.1} 
\renewcommand{\a}{0.8}   
\newcommand{\x}{0}
\newcommand{\y}{0}
\fill [red!20] (0.75, -0.5) rectangle ++(0.5,4);
\fill [red!20] (-0.5, 0.75) rectangle ++(4,0.5);
\draw (\LL*\d/2,\LL*\d/2) node (X) {};
\foreach \i in {0,...,\LL}
{
  \foreach \j in {0,...,\LL}
  {
  \renewcommand{\x}{\i*\d}
  \renewcommand{\y}{\j*\d}
  \pgfmathsetmacro{\color}{ifthenelse(\i==\ic && \j==\jc,"red", ifthenelse(\i==\ic || \j==\jc, "blue", "black"))}
  \pgfmathsetmacro{\Hdir}{ifthenelse(\i <\ic,"latex", "latex reversed")}
  \pgfmathsetmacro{\Vdir}{ifthenelse(\j <\jc, "latex", "latex reversed")}
  \draw [fill=\color] (\x,\y) circle (\r);
  \ifthenelse{\i<\LL}{\draw [midarrow={\Hdir}](\x+\r,\y) -- (\x+\r+\a,\y);}{}
  \ifthenelse{\j<\LL}{\draw [midarrow={\Vdir}](\x,\y+\r) -- (\x,\y+\r+\a);}{} 
  
  \draw [midarrow={latex reversed}] (\x+\r*0.707,\y+\r*0.707) -- (\x+\d/3,\y+\d/3);
  }
}
\end{tikzpicture},
\end{equation}
where each tensor has five indices, one physical index and four virtual indices. 
Each tensor is an isometric matrix when its legs with incoming and outgoing arrows are respectively grouped as the row and column indices of the matrix. 
Hence when we contract a matrix with its complex conjugate over the incoming legs, the result is an identity operator. 
Note that the physical indices always carry incoming arrows.
The dimension of the physical leg coming out of the page is the local Hilbert space dimension $d$.
The bond dimension $\chi$ of the virtual indices between sites in the lattice controls the variational power and computational cost of the networks and algorithms, respectively.

The tensor with only incoming indices, e.g. the red tensor in Eq.~\eqref{eq:iso_finite}, is the orthogonality center (OC) of the isoTNS.
The horizontal and vertical columns of tensors with only incoming indices from other parts of the isoTNS, i.e., the red and the blue tensors in Eq.~\eqref{eq:iso_finite} with light red background, is the orthogonality hypersurface.
For an 1D MPS in isometric form, the OC is a 0D wave function representing the system expressed in an orthonormal basis, and the entire MPS is the orthogonality hypersurface. 
In a 2D isoTNS, the orthogonality hypersurface is a 1D wavefunction representing the full 2D state also in an orthonormal basis.
We will only work with the isoTNS with one OC~\footnote{It is possible to have more than one OCs in the TNS. In that case, the reduced density matrices of the OCs will be separable}.
Given this isometric form, expectation values of operators contained entirely within the orthogonality hypersurface are easy to evaluate, as this reduces to the standard MPS problem.
For more complicated operators, we must contract all tensors that can be reached by outgoing arrows, and thus this is in principle as difficult as generic TNS contractions.

For an MPS, we can exactly move the OC to any position by the QR algorithm, enabling efficient evaluation of local operators.
For an isoTNS, moving the orthogonality hypersurface to a neighboring column while ensuring consistent isometric arrows throughout the network in general cannot be performed exactly. 
The finite MM algorithm approximately accomplishes this task by repeatedly applying the following local tripartite decomposition at each site of the column:
\begin{equation}
    \label{eq:tripartite}
    \newcommand{\x}{0}
\newcommand{\y}{0.1}
\renewcommand{\d}{1.0}   
\renewcommand{\r}{0.1}   
\renewcommand{\a}{0.8}   
\begin{tikzpicture}[baseline = (X.base),every node/.style={scale=1.0},scale=1.0]
\draw (0,0) node (X) {};
  \draw [fill=red] (\x,\y) circle (\r);
  \draw [midarrow={latex reversed}](\x+\r,\y) -- (\x+\r+\a,\y); 
  \draw [midarrow={latex}](\x-\r-\a,\y) -- (\x-\r,\y); 
  \draw [midarrow={latex reversed}](\x,\y+\r) -- (\x,\y+\r+\a);
  \draw [midarrow={latex}](\x-0.5*\a,\y-0.866*\a) to[out=90,in=210] (\x-0.866*\r,\y-0.5*\r);
  \draw [midarrow={latex}](\x+0.5*\a,\y-0.866*\a) to[out=90,in=330] (\x+0.866*\r,\y-0.5*\r);
\end{tikzpicture}
\approx
\begin{tikzpicture}[baseline = (X.base),every node/.style={scale=1.0},scale=1.0]
  \newcommand{\dH}{0.8}
  \renewcommand{\a}{0.6}
  \draw (0,-0.25+0.1) node{$j$};
  \draw (-\dH/4*1.7-0.1,\dH/2*1.1+0.1) node{$k$};
  \draw (\dH/4*1.7+0.1,\dH/2*1.1+0.1) node{$i$};
  \renewcommand{\x}{-\dH/2}
  \renewcommand{\y}{0.1}
  \draw [fill=black] (\x,\y) circle (\r);
  \draw [midarrow={latex reversed}] (\x-\r,\y) -- (\x-\r-\a,\y);
  \draw [midarrow={latex}] (\x+\r,\y) -- (\x+\r+\a,\y);
  \draw [midarrow={latex}] (\x+\r*0.5,\y+\r*0.866) -- (0-\r*0.5,\dH*0.866-\r*0.866+\y);
  \draw [midarrow={latex reversed}] (\x,\y-\r) -- (\x,-\dH*0.866+\y);
  \renewcommand{\x}{\dH/2}
  \renewcommand{\y}{0.1}
  \draw [fill=blue] (\x,\y) circle (\r);
  \draw [midarrow={latex reversed}] (\x+\r,\y) -- (\x+\r+\a,\y);
  \draw [midarrow={latex}] (\x-\r*0.5,\y+\r*0.866) -- (0+\r*0.5,\dH*0.866-\r*0.866+\y);
  \draw [midarrow={latex reversed}] (\x,\y-\r) -- (\x,-\dH*0.866+\y);
  \renewcommand{\x}{0}
  \renewcommand{\y}{\dH*0.866+0.1}
  \draw [fill=red] (\x,\y) circle (\r);
  \draw [midarrow={latex reversed}] (\x,\y+\r) -- (\x,\y+\r+\a);
\end{tikzpicture}. 
\end{equation}
Due to the isometric form, this local decomposition is done in an orthonormal basis and is in general a constrained optimization problem.
For details of the tri-splitting algorithm in Eq.~\eqref{eq:tripartite}, see \cite{zaletel2020isometric,lin2022efficient}.

The MM algorithm performs the following decomposition of the orthogonality column as a two-sided MPS:
\begin{equation}
    \label{eq:finite_MM}
    \newcommand{\LL}{3}      
\renewcommand{\d}{1.0}   
\renewcommand{\r}{0.1}   
\renewcommand{\a}{0.8}   
\newcommand{\dH}{0.5}   
\newcommand{\aH}{0.3}   
\begin{tikzpicture}[baseline = (X.base),every node/.style={scale=1.0},scale=1.0]
\newcommand{\x}{0}
\newcommand{\y}{0}
\draw (\LL*\d/2,\LL*\d/2) node (X) {};
\foreach \j in {0,...,\LL}
{
  \renewcommand{\x}{0}
  \renewcommand{\y}{\j*\d}
  \pgfmathsetmacro{\color}{ifthenelse(\j==0, "red", "blue")}
  \draw [fill=\color] (\x,\y) circle (\r);
  \draw [midarrow={latex reversed}](\x+\r,\y) -- (\x+\r+\aH,\y); 
  \draw [midarrow={latex}](\x-\r-\aH,\y) -- (\x-\r,\y); 
  \ifthenelse{\j<\LL}{
  \draw [midarrow={latex reversed}](\x,\y+\r) -- (\x,\y+\r+\a);
  }{}
}
\draw (0,-0.5) node {(i)};
\end{tikzpicture}
\hspace{-12mm}
\approx
\hspace{-3mm}
\begin{tikzpicture}[baseline = (X.base),every node/.style={scale=1.0},scale=1.0]
\newcommand{\x}{0}
\newcommand{\y}{0}
\draw (\LL*\d/2,\LL*\d/2) node (X) {};
\foreach \j in {1,...,\LL}
{
  \renewcommand{\x}{0}
  \renewcommand{\y}{\j*\d}
  \pgfmathsetmacro{\color}{ifthenelse(\j==0, "red", "blue")}
  \draw [fill=\color] (\x,\y) circle (\r);
  \draw [midarrow={latex reversed}](\x+\r,\y) -- (\x+\r+\aH,\y); 
  \draw [midarrow={latex}](\x-\r-\aH,\y) -- (\x-\r,\y); 
  \ifthenelse{\j<\LL}{
  \draw [midarrow={latex reversed}](\x,\y+\r) -- (\x,\y+\r+\a);
  }{}<++>
}
\renewcommand{\x}{-\dH/2}
\renewcommand{\y}{0}
\draw [midarrow={latex reversed}](\x-\r,\y) -- (\x-\r-\aH,\y); 
\draw [fill=black] (\x,\y) circle (\r);
\draw [midarrow={latex}](\x+\r,\y) -- (\x+\r+\aH,\y); 
\renewcommand{\x}{\dH/2}
\draw [fill=blue] (\x,\y) circle (\r);
\draw [midarrow={latex reversed}](\x+\r,\y) -- (\x+\r+\aH,\y); 

\renewcommand{\x}{0}
\renewcommand{\y}{\d*0.5}
\draw [fill=red] (\x,\y) circle (\r);
\draw [midarrow={latex reversed}](\x,\y+\r) -- (\x,\y-\r+\d*0.5); 
\renewcommand{\x}{-\dH/2}
\renewcommand{\y}{0}
\draw [midarrow={latex}](\x+\r*0.5,\y+\r*0.866) -- (0-\r*0.5,\y+\d*0.5-\r*0.866); 
\renewcommand{\x}{\dH/2}
\draw [midarrow={latex}](\x-\r*0.5,\y+\r*0.866) -- (0+\r*0.5,\y+\d*0.5-\r*0.866); 
\draw (0,-0.5) node {(ii)};
\end{tikzpicture}
\hspace{-11mm}
=
\hspace{-1mm}
\begin{tikzpicture}[baseline = (X.base),every node/.style={scale=1.0},scale=1.0]
\newcommand{\x}{0}
\newcommand{\y}{0}
\draw (\LL*\d/2,\LL*\d/2) node (X) {};
\foreach \j in {1,...,\LL}
{
  \renewcommand{\x}{0}
  \renewcommand{\y}{\j*\d}
  \pgfmathsetmacro{\color}{ifthenelse(\j==1, "red", "blue")}
  \draw [fill=\color] (\x,\y) circle (\r);
  \draw [midarrow={latex reversed}](\x+\r,\y) -- (\x+\r+\aH,\y); 
  \draw [midarrow={latex}](\x-\r-\aH,\y) -- (\x-\r,\y); 
  \ifthenelse{\j<\LL}{
  \draw [midarrow={latex reversed}](\x,\y+\r) -- (\x,\y+\r+\a);
  }{}<++>
}
\renewcommand{\x}{-\dH/2}
\renewcommand{\y}{0}
\draw [midarrow={latex reversed}](\x-\r,\y) -- (\x-\r-\aH,\y); 
\draw [fill=black] (\x,\y) circle (\r);
\draw [midarrow={latex}](\x+\r,\y) -- (\x+\r+\aH,\y); 
\renewcommand{\x}{\dH/2}
\draw [fill=blue] (\x,\y) circle (\r);
\draw [midarrow={latex reversed}](\x+\r,\y) -- (\x+\r+\aH,\y); 

\renewcommand{\x}{0}
\renewcommand{\y}{\d*0.5}
\renewcommand{\x}{-\dH/2}
\renewcommand{\y}{0}
\draw [midarrow={latex}](\x+\r*0.24,\y+\r*0.866) -- (0-\r*0.24,\y+\d-\r*0.866); 
\renewcommand{\x}{\dH/2}
\draw [midarrow={latex}](\x-\r*0.24,\y+\r*0.866) -- (0+\r*0.24,\y+\d-\r*0.866); 
  \draw (0,-0.5) node {(iii)};
\end{tikzpicture}
\hspace{-11mm}
\approx
\hspace{-1mm}
\begin{tikzpicture}[baseline = (X.base),every node/.style={scale=1.0},scale=1.0]
\newcommand{\x}{0}
\newcommand{\y}{0}
\draw (\LL*\d/2,\LL*\d/2) node (X) {};
\foreach \j in {2,...,\LL}
{
  \renewcommand{\x}{0}
  \renewcommand{\y}{\j*\d}
  \pgfmathsetmacro{\color}{ifthenelse(\j==2, "red", "blue")}
  \draw [fill=\color] (\x,\y) circle (\r);
  \draw [midarrow={latex reversed}](\x+\r,\y) -- (\x+\r+\aH,\y); 
  \draw [midarrow={latex}](\x-\r-\aH,\y) -- (\x-\r,\y); 
  \ifthenelse{\j<\LL}{
  \draw [midarrow={latex reversed}](\x,\y+\r) -- (\x,\y+\r+\a);
  }{}
}
\foreach \j in {0,...,1}
{
  \renewcommand{\x}{-\dH/2}
  \renewcommand{\y}{\j*\d}
  \draw [midarrow={latex reversed}](\x-\r,\y) -- (\x-\r-\aH,\y); 
  \draw [fill=black] (\x,\y) circle (\r);
  \draw [midarrow={latex}](\x+\r,\y) -- (\x+\r+\aH,\y); 
  \renewcommand{\x}{\dH/2}
  \draw [fill=blue] (\x,\y) circle (\r);
  \draw [midarrow={latex reversed}](\x+\r,\y) -- (\x+\r+\aH,\y); 
  \ifthenelse{\j<1}{
  \foreach \i in {0,...,1}
  {
  \renewcommand{\x}{-\dH/2+\i*\dH}
  \draw [midarrow={latex}](\x,\y+\r) -- (\x,\y+\r+\a);
  }
  }{}
}
\renewcommand{\x}{-\dH/2}
\renewcommand{\y}{\d}
\draw [midarrow={latex}](\x+\r*0.24,\y+\r*0.866) -- (0-\r*0.24,\y+\d-\r*0.866); 
\renewcommand{\x}{\dH/2}
\draw [midarrow={latex}](\x-\r*0.24,\y+\r*0.866) -- (0+\r*0.24,\y+\d-\r*0.866); 
  \draw (0,-0.5) node {(iv)};
\end{tikzpicture}
\hspace{-13mm}
\approx
...
\approx
\hspace{-3mm}
\begin{tikzpicture}[baseline = (X.base),every node/.style={scale=1.0},scale=1.0]
\draw (\LL*\d/2,\LL*\d/2) node (X) {};
\newcommand{\x}{0}
\newcommand{\y}{0}
\foreach \i in {0,...,1}
{
  \foreach \j in {0,...,\LL}
  {
    \renewcommand{\x}{\i*\dH}
    \renewcommand{\y}{\j*\d}
    \pgfmathsetmacro{\Hdir}{ifthenelse(\i == 0, "latex", "latex reversed")}
    \pgfmathsetmacro{\Vdir}{ifthenelse(\j <\LL,"latex", "latex reversed")}
    \pgfmathsetmacro{\color}{ifthenelse(\i==0, ifthenelse(\j==\LL, "blue", "black"), ifthenelse(\j<\LL, "blue", "red"))}
    \draw [fill=\color] (\x,\y) circle (\r);
    \draw [midarrow={\Hdir}](\x+\r,\y) -- (\x+\r+\aH,\y); 
    \ifthenelse{\j<\LL}{
    \draw [midarrow={\Vdir}](\x,\y+\r) -- (\x,\y+\r+\a);
    }{}
  }
}
\foreach \j in {0,...,\LL}
{
  \renewcommand{\x}{0*\d}
  \renewcommand{\y}{\j*\d}
  \draw [midarrow={latex reversed}](\x-\r,\y) -- (\x-\r-\aH,\y); 
}
  \draw (\dH/2,-0.5) node {(v)};
\end{tikzpicture},
    \hspace{-1cm}
\end{equation}
where in the 2D isoTNS, the physical sites are grouped either with the left or the right virtual index.
The MM algorithm effectively ``unzips" the original MPS into an isoTNO and a new MPS.
The starting point of the MM algorithm is a orthogonality column with all vertical arrows pointing down, all horizontal arrows pointing in, and the bottom-most site being the OC.

The MM algorithm has a complexity of $\mathcal{O}(\chi^7)$, when all virtual bond dimensions in the isoTNS are $\chi$ and physical dimension $d = \mathcal{O}(1)$ is ignored.
To optimize finite isoTNS, originally a two-dimensional imaginary time evolution algorithm TEBD\textsuperscript{2} was proposed~\cite{zaletel2020isometric}. 
TEBD\textsuperscript{2} on isoTNS is equivalent to full update in unconstrained TNS, as for isoTNS, the environment of the orthogonality center (i.e., norm matrix)~\cite{lubasch2014algorithms} is always an identity and doesn't require any approximated boundary contraction.
Thus the cost of the full-update in TEBD\textsuperscript{2} calculations is reduced from $\mathcal{O}(\chi^{10})$ for unconstrained TNS to $\mathcal{O}(\chi^{7})$ for isoTNS, at the expense of errors due to the MM.

Recently, DMRG\textsuperscript{2}, a 2D generalization of DMRG, was introduced in which effective Hamiltonian environments for each tensor are formed from local terms in the Hamiltonian~\cite{lin2022efficient}.
DMRG\textsuperscript{2} for either isoTNS or unconstrained TNS have the similar complexity~\footnote{The complexity is $\mathcal{O}(\chi^{12})$ for DMRG\textsuperscript{2} with explicit $H_\text{eff}$ construction and $\mathcal{O}(\chi^{10})$ utilizing the sparse structure.}, but the optimization problem for local tensors reduces from the generalized eigenvalue problem for unconstrained TNS to the regular eigenvalue problem for isoTNS.
This is because optimization is done in an orthonormal basis, and thus the norm matrix is an identity operator~\cite{schollwock2011density}.
This method has recently been successfully applied to the Kitaev model on honeycomb lattice~\cite{lin2022efficient}.

\subsection{Infinite strip geometry}
\label{sec:isoTNS_infinite}

An isoTNS on an infinite strip that is uniform along the infinite direction, for a $4\times \infty$ example, can be represented diagrammatically as
\begin{equation}
    \label{eq:iso_strip}
    \ket{\Psi}\: =\:\begin{tikzpicture}[baseline = (X.base),every node/.style={scale=1.0},scale=1.0]
\newcommand{\LL}{3}      
\newcommand{\ic}{1}
\newcommand{\jc}{10}
\renewcommand{\d}{1.0}   
\renewcommand{\r}{0.1} 
\renewcommand{\a}{0.8}   
\newcommand{\x}{0}
\newcommand{\y}{0}
\fill [red!20] (0.75, -1.5) rectangle ++(0.5,6);
\draw (\LL*\d/2,\LL*\d/2) node (X) {};
\foreach \i in {0,...,\LL}
{
  \foreach \j in {0,...,\LL}
  {
  \renewcommand{\x}{\i*\d}
  \renewcommand{\y}{\j*\d}
  \pgfmathsetmacro{\color}{ifthenelse(\i==\ic && \j==\jc,"red", ifthenelse(\i==\ic || \j==\jc, "blue", "black"))}
  \pgfmathsetmacro{\Hdir}{ifthenelse(\i <\ic,"latex", "latex reversed")}
  \pgfmathsetmacro{\Vdir}{ifthenelse(\j <\jc, "latex", "latex reversed")}
  \draw [fill=\color] (\x,\y) circle (\r);
  \draw [midarrow={\Hdir}](\x+\r,\y) -- (\x+\r+\a,\y); 
  \draw [midarrow={\Vdir}](\x,\y+\r) -- (\x,\y+\r+\a);
  \draw [midarrow={latex reversed}] (\x+\r*0.707,\y+\r*0.707) -- (\x+\d/3,\y+\d/3);
  }
}
\foreach \i in {0,...,\LL}
{
  \renewcommand{\x}{\i*\d}
  \renewcommand{\y}{0*\d}
  \draw [midarrow={latex}](\x,\y-\r-\a) -- (\x,\y-\r); 
  \draw (\x,\y-\r-\a*1.25) node{$\vdots$}; 
  \draw (\x,\y+\d*\LL+\r+\a*1.5) node{$\vdots$}; 
}
\foreach \j in {0,...,\LL}
{
  \renewcommand{\x}{0*\d}
  \renewcommand{\y}{\j*\d}
  \draw [midarrow={latex}](\x-\r-\a,\y) -- (\x-\r,\y); 
}
\end{tikzpicture},
\end{equation}
where the tensors within a column are the same but may be different for different columns. 
Within the 1D orthogonality hypersurface (colored light red), the OC can be placed anywhere within the column using the standard uniform MPS methods~\cite{vanderstraeten2019tangent}.
In the networks we consider, we place the OC at either $y = \pm \infty$ (given the upward pointing arrows, it is placed at positive $\infty$ in Eq.~\eqref{eq:iso_strip}), so that the network has uniform vertical isometry arrows for all columns.
Different to the finite case, here we work with an orthogonality ``column'', instead of the $+$-shape orthogonality hypersurface. 
We note that while we only consider a single-site unit-cell in this work, all algorithms presented can be extended to multi-site unit cells, allowing for periodic inhomogeneity along the infinite direction.

While boundary indices for finite isoTNS are typically trivial, as in Eq.~\eqref{eq:iso_finite}, in general the indices on the boundary can be non-trivial with dimension $D_b$, as shown in the infinite strip network Eq.~\eqref{eq:iso_strip}.
The boundary legs always carry incoming arrows as the OC is contained within the network.
When $D_b=1$, $\ket{\Psi}$ is a pure quantum state. 
When $D_b>1$, indicating non-trivial boundary indices, $\ket{\Psi}$ can be viewed as a purification of the density matrix of the physical sites.
Typically for a bond-dimension $\chi$ strip isoTNS, we take $D_b = \chi$.
As we will show later, having non-trivial boundaries makes optimization easier and ensures that the orthogonality column, viewed as a two-sided MPS, is injective~\cite{cirac2021review}.

\section{Infinite Moses Move}
\label{sec:iMM}
Similar to the case of finite isoTNS, we desire the ability to change the isometric structure and move the orthogonality column of strip isoTNS network.
This would allow for evaluation of correlation functions within the orthogonality column by 1D iMPS methods and also ensures that optimization algorithms are done in an orthonormal basis.
The iMM algorithm solves the problem of splitting a normalized infinite two-sided MPS $\ket{\Psi}$ into an infinite isoTNO $A$ and a normalized infinite two-sided MPS $\ket{\Phi}$: %
\begin{equation}
    \label{eq:iQR}
    \newcommand{\LL}{3}      
\renewcommand{\d}{1.1}   
\renewcommand{\r}{0.25}   
\renewcommand{\a}{0.6}   
\newcommand{\dH}{1.0}   
\newcommand{\aH}{0.5}   
\begin{tikzpicture}[baseline = (X.base),every node/.style={scale=0.8},scale=0.8]
\newcommand{\x}{0}
\newcommand{\y}{0}
\draw (\LL*\d/2,\LL*\d/2) node (X) {};
\foreach \j in {0,...,\LL}
{
  \renewcommand{\x}{0}
  \renewcommand{\y}{\j*\d}
  \draw [color=blue, very thick](\x,\y) circle (\r);
  \draw (\x,\y) node {$\psi$};
  \draw [midarrow={latex reversed}](\x+\r,\y) -- (\x+\r+\aH,\y); 
  \draw [midarrow={latex}](\x-\r-\aH,\y) -- (\x-\r,\y); 
  \draw [midarrow={latex reversed}](\x,\y+\r) -- (\x,\y+\r+\a);
}
\renewcommand{\y}{0}
\draw [midarrow={latex}](\x,\y-\r) -- (\x,\y-\r-\a);
\draw (\x,\LL*\d+\r+\a*1.7) node{$\vdots$}; 
\draw (\x,0-\r-\a*1.25) node{$\vdots$}; 
\draw (\x, -1.7) node{$\ket{\Psi}$};

\draw (1.5, 1.5) node{$\approx$};

\renewcommand{\x}{3}
\renewcommand{\y}{0}
\foreach \i in {0,...,1}
{
  \renewcommand{\x}{3+\i*\dH}
  \foreach \j in {0,...,\LL}
  {
    \renewcommand{\y}{\j*\d}
    \pgfmathsetmacro{\Hdir}{ifthenelse(\i == 0, "latex", "latex reversed")}
    \ifthenelse{\i<1}{\draw [color=black, very thick] (\x,\y) circle (\r);\draw (\x,\y) node{$a$};}
                {\draw [color=blue, very thick] (\x,\y) circle (\r);\draw (\x,\y) node{$\phi$};}
    \ifthenelse{\i<1}{
        \draw [midarrow={\Hdir}](\x+\r,\y) -- (\x+\r+\aH,\y);
    }{
        \draw [midarrow={\Hdir}](\x+\r,\y) -- (\x+\r+\aH,\y);
    }
     
    \draw [midarrow={latex}](\x,\y+\r) -- (\x,\y+\r+\a);
  }
  \draw (\x,\y+\LL*\d+\r+\a*1.7) node{$\vdots$}; 
  \draw (\x,\y-\r-\a*1.25) node{$\vdots$}; 
  \draw [midarrow={latex}](\x,\y-\r-\a) -- (\x,\y-\r);
}
\foreach \j in {0,...,\LL}
{
  \renewcommand{\x}{3+0*\d}
  \renewcommand{\y}{\j*\d}
  \draw [midarrow={latex reversed}](\x-\r,\y) -- (\x-\r-\aH,\y); 
}
\draw (3+0*\dH, -1.7) node{$A$};
\draw (3+1*\dH, -1.7) node{$\ket{\Phi}$};

\end{tikzpicture}.
\end{equation}
Note that we group the physical leg with the right or left legs of the $\psi$ tensor depending on whether we want the physical leg to be on the $a$ tensor or the $\phi$ tensor after splitting.
In cases where $\ket{\psi}$ represents the horizontal contraction of two columns, we distribute one physical leg to each of the left and right legs of $\psi$ so that both $a$ and $\phi$ will have a physical leg.
For all iMM algorithms we present, we require the initial iMPS to have all vertical arrows pointing down, all horizontal arrows incoming, and the OC to be at $-\infty$.
With this, the iMM can again be viewed as unzipping the two-column iMPS into an isoTNO $A$ and a new iMPS $\ket{\Phi}$.

Formulating this as an optimization problem, one finds the optimal isometric tensor $a'$ and normalized MPS tensor $\phi'$ by maximizing the overlap density between $\ket{\Psi}$ and $A\ket{\Phi}$:    
\begin{equation}
    \label{eq:lambda1}
    a', \phi' = \argmax_{a,\phi\in\text{isometry}} \Re \, \lambda_1( T_{\psi: a\phi})
\end{equation} 
where $\Re\ \lambda_1$ denotes the real part of the largest eigenvalue of the mixed transfer matrix $T_{\psi: a\phi}$ between $\ket{\Psi}$ and $A\ket{\Phi}$: 
\begin{equation}
    \label{eq:T}
    T \equiv T_{\psi:a\phi} = \begin{tikzpicture}[baseline = (X.base),every node/.style={scale=0.6},scale=.4]
\draw (0.0,0.0) node (X) {};
\draw (0.0,0.0) circle (0.5);
\draw (0.0,0.0) node {$\overline\psi$};
\draw (-1.5,-0.4) -- (-0.5,-0.15);
\draw (0.5,0.15) -- (1.5,0.4); 
\draw (0.0,-0.5) -- (0.0,-1.5);
\draw (0.0,0.5) -- (0.0,1.5);
\newcommand{\x}{3}
\draw (0.0+\x,0.0) circle (0.5);
\draw (0.0+\x,0.0) node {$a$};
\draw (-1.5+\x,-0.4) -- (-0.5+\x,-0.15);
\draw (0.5+\x,0.15) -- (1.5+\x,0.4); 
\draw (0.0+\x,-.5) -- (0.0+\x,-1.5);
\draw (0.0+\x,0.5) -- (0.0+\x,1.5);
\newcommand{\y}{\x+3}
\newcommand{\z}{0.8}
\draw (0.0+\y,0.0+\z) circle (0.5);
\draw (0.0+\y,0.0+\z) node {$\phi$};
\draw (-1.5+\y,-0.4+\z) -- (-0.5+\y,-0.15+\z);
\draw (0.5+\y,0.15+\z) -- (1.5+\y,0.4+\z); 
\draw (0.0+\y,-.5+\z) -- (0.0+\y,-1.5+\z);
\draw (0.0+\y,0.5+\z) -- (0.0+\y,1.5+\z);
\draw (-1.5,-0.4) edge[in=195,out=195] (-1.5+\x,-0.4); 
\draw (1.5, 0.4) -- (1.5+2,0.4+0.5);
\draw (1.5+2, 0.4+0.5) edge[in=15,out=15] (1.5+\y,0.4+\z);
\end{tikzpicture}.
\end{equation}
As any iMPS can be written exactly in isometric form with the same bond dimension, we directly search for isometric $\phi$ tensors.
In the following, when no confusion arises, we simply use $T$ to denote this mixed transfer matrix. 
Note that maximizing $\Re \, \lambda_1(T)$ or $\abs{\lambda_1(T)}$ is equivalent for our purpose because of the unitary freedom of the isometries.

We now propose and evaluate four algorithms for solving the splitting problem posed in Eq.~\eqref{eq:iQR}:  \iMML, \iMMP, \iMMM, and \iMMC.
We first give a brief overview of these methods. 
\iMML\ is the direct generalization of the finite MM, in which the local tripartite decomposition in Eq.~\eqref{eq:tripartite} is iterated until convergence.
The method is not guaranteed to converge, but we find it quickly provides an approximate  solution.
Thus, we use it as an initialization for the following optimization methods.
The \iMMP\ method minimizes the error of the fixed-point tripartite decomposition using repeated polar decompositions.
It is slower than \iMML\, and does not directly maximize the overlap density $\Re\ \lambda_1$.
However, we find it to be nearly optimal in practice.
Finally, \iMMM\ and \iMMC\ maximize the overlap density $\Re\ \lambda_1$ through two different methods.
The \iMMM\ method finds the $\phi$ tensor by the variational MPO-MPS compression algorithm and the $a$ tensor by the polar decomposition over the linearized overlap.
The method is slightly slower than \iMMP\ and yield comparable results.
The \iMMC\ method is based on the conjugate gradient ascent of the overlap density $\Re\ \lambda_1$.
This method has difficulty reaching a satisfying (local) minimum on its own and is best used to improve the results of other methods. 
It is found that the most efficient strategy is to use \iMML\ to obtain a stable initial guess and then use \iMMP\ to finish the infinite splitting optimization. 
Using \iMMM\ or \iMMC\ in the end is optional, as the improvements they provide after \iMMP\, are very little.
We include details of the four methods for completeness but do not use \iMMM\ and \iMMC\ beyond benchmarking.

\subsection{\iMML}
\label{sec:iMML}

Given the input iMPS $\ket{\Psi}$ made of tensors $\psi$ as shown on the left in Eq.~\eqref{eq:iQR}, one performs the finite MM in Eq.~\eqref{eq:finite_MM} assuming that all the blue tensors above the OC equal $\psi$.
The first, normalized $s_1$ is randomly generated such that the dimensions of the lower two legs match that of the desired $a$ and $\phi$ tensors; $s_1$ plays the role of the initial zero-site wavefunction.
One then uses the tripartite decomposition in Eq.~\eqref{eq:tripartite} to solve the following problem for site $n$ iteratively along the infinite direction: 
\begin{equation}
    \label{eq:renyi_iteration}
    \newcommand{\LL}{1}      
\renewcommand{\d}{1.0}   
\renewcommand{\r}{0.25}   
\renewcommand{\a}{0.5}   
\newcommand{\dH}{\d}   
\newcommand{\aH}{\a}   
\begin{tikzpicture}[baseline = (X.base),every node/.style={scale=0.8},scale=0.8]
\newcommand{\x}{0}
\newcommand{\y}{0}
\draw (0,\d/2) node (X) {};
\renewcommand{\x}{0}
\renewcommand{\y}{\d}
\draw (\x,\y)[color=blue, very thick] circle (\r);
\draw (\x,\y) node{$\psi$};
\draw [midarrow={latex reversed}](\x+\r,\y) -- (\x+\r+\aH,\y); 
\draw [midarrow={latex}](\x-\r-\aH,\y) -- (\x-\r,\y); 
\draw [midarrow={latex reversed}](\x,\y+\r) -- (\x,\y+\r+\aH);
\renewcommand{\x}{0}
\renewcommand{\y}{0}
\draw (\x,\y)[color=red, very thick] circle (\r);
\draw (\x,\y) node{$s_n$};
\draw [midarrow={latex reversed}](\x,\y+\r) -- (\x,\y-\r+\d); 
\renewcommand{\x}{-\dH/2}
\renewcommand{\y}{-\d}
\draw [midarrow={latex}](\x+\r*0.5,\y+\r*0.866) -- (0-\r*0.5,\y+\d-\r*0.866); 
\renewcommand{\x}{\dH/2}
\draw [midarrow={latex}](\x-\r*0.5,\y+\r*0.866) -- (0+\r*0.5,\y+\d-\r*0.866); 
\end{tikzpicture}
\approx
\begin{tikzpicture}[baseline = (X.base),every node/.style={scale=0.8},scale=0.8]
  \draw (0,\d/2) node (X) {};
  \newcommand{\x}{-\dH/2}
  \newcommand{\y}{0}
  \draw (\x,\y)[color=black, very thick] circle (\r);
  \draw (\x,\y) node{$a_n$};
  \draw (\x,\y+0.5) node{$k$};
  \draw [midarrow={latex reversed}] (\x-\r,\y) -- (\x-\r-\a,0);
  \draw [midarrow={latex}] (\x+\r,\y) -- (\x+\r+\a,0);
  \draw [midarrow={latex reversed}] (\x,\y-\r) -- (\x,\y-\r-\a);
  \draw [midarrow={latex}] (\x+\r*0.5,\y+\r*0.866) -- (0-\r*0.5,\dH*0.866-\r*0.866);
  \renewcommand{\x}{\dH/2}
  \renewcommand{\y}{0}
  \draw (\x,\y)[color=blue, very thick] circle (\r);
  \draw (\x,\y) node{$\phi_n$};
   \draw (\x,\y+0.5) node{$i$};
  \draw [midarrow={latex reversed}] (\x+\r,\y) -- (\x+\r+\a,0);
  \draw [midarrow={latex reversed}] (\x,\y-\r) -- (\x,\y-\r-\a);
  \draw [midarrow={latex}] (\x-\r*0.5,\y+\r*0.866) -- (0+\r*0.5,\dH*0.866-\r*0.866);
  \renewcommand{\x}{0}
  \renewcommand{\y}{\dH*0.866}
  \draw (\x,\y)[color=red, very thick] ellipse (0.5 and \r); 
  \draw (\x,\y) node{$s_{n+1}$};
    \draw (\x,\y-1.2) node{$j$};
    \draw (\x+0.2,\y+0.5) node{$l$};
  \draw [midarrow={latex reversed}] (\x,\y+\r) -- (\x,\y+\r+\a);
\end{tikzpicture}. 
\end{equation}
This iteration terminates if $\norm{s_n-s_{n+1}}$ is smaller than certain threshold or $n$ exceeds certain iteration limit. 
One then takes $A$ and $\ket{\Phi}$ to be composed, respectively, of the last $a_n$ and $\phi_n$ in the iteration.

For faster runtime and better convergence, we found it important to minimize R\'enyi-2 entropy in the tripartite splitting here. We suggest the interested reader to checkout the details in previous works~\cite{zaletel2020isometric,lin2022efficient} and describe the additional modification in the following.
The tripartite splitting has unitary gauge redundancies on the internal bonds $i, j$, and $k$. 
To increase the chance that the MM iterations converge along the infinite direction, we need to the fix the gauge redundancies on bond $k$ and $i$. 
We do not need to fix the redundancy on bond $j$, because only $S$ is used for the convergence criteria.
To fix $k$ and $i$, we do an SVD on matrix $S_{k:li} = UsV^\dag$, which is formed by grouping bond $i$ and $l$ as the column index, and replace $S$ with $sV^\dag$. 
We then do an SVD on matrix $S_{lk:i} = UsV^\dag$, formed by grouping bond $k$ and $l$, and replace $S$ with $Us$.
The two SVDs almost fix the gauge redundancy except the gauge freedom in the SVD itself: every column of $U$ from an SVD can be multiplied by a phase, as long as the corresponding row of the $V^\dag$ is multiplied by the inverse of that phase. 
To fix this phase freedom, we view $S_{l=0}$ as a matrix and demand its first column and first row to be all positive numbers. 
This can be achieved by a unitary diagonal matrix on bond $i$ and $k$ independently.
When the singular values of $S_{k:li}$ and $S_{lk:i}$ are not degenerate, these operations fix the gauge freedom on bond $i$ and $k$ completely. 
Despite the gauge-fixing, $\iMML$\, still often fails to converge.
Thus, it is best used to provide an initial guess for the $\iMMP$\, algorithm which is guaranteed to converge along the infinite direction.

\subsection{\iMMP}
\label{sec:iMMP}

To overcome the convergence issue of \iMML, we consider \iMMP\ to directly optimize the approximate fixed point equation of the tripartite decomposition in Eq.~\eqref{eq:renyi_iteration}.
In other words, we enforce translational invariance of the iMM by requiring $s_{n} = s_{n+1}=s$ in Eq.~\eqref{eq:renyi_iteration} and optimize over $s$, $a$, and $\phi$ to make the equality close to exact.
More precisely, we maximize the real part of the overlap between the left and right-hand side of the Eq.~\eqref{eq:renyi_iteration} under the constraints that $\norm{s} = 1$, $a$ is an isometry, and $\phi$ is a normalized MPS tensor in isometric form~\footnote{The reason to take the real part in Eq.~\eqref{eq:iMM_polar} is that for any normalized states $\ket{x}$ and $\ket{y}$, $\norm{\ket{x}-\ket{y}}^2 = 2 (1-\Re \braket{x|y})$.}:
\begin{equation}
    \label{eq:iMM_polar}
    \newcommand{\LL}{1}      
\renewcommand{\d}{1.0}   
\renewcommand{\r}{0.3}   
\renewcommand{\a}{0.4}   
\hspace{-5mm}
F=
\Re 
\hspace{-3mm}
\begin{tikzpicture}[baseline = (X.base),every node/.style={scale=0.6},scale=0.6]
\newcommand{\x}{0}
\newcommand{\y}{0}
\draw (0,\d/2) node (X) {};
\renewcommand{\x}{0}
\renewcommand{\y}{\d}
\draw (\x,\y)[color=blue, very thick] circle (\r);
\draw (\x,\y) node{$\overline\psi$};
\draw (\x+\r,\y) -- (\x+\r+\a,\y); 
\draw (\x-\r-\a,\y) -- (\x-\r,\y); 
\draw (\x,\y+\r) -- (\x,\y+\r+\a);
\renewcommand{\x}{0}
\renewcommand{\y}{0}
\draw (\x,\y)[color=red, very thick] circle (\r);
\draw (\x,\y) node{$\overline s$};
\draw (\x,\y+\r) -- (\x,\y-\r+\d); 
\renewcommand{\y}{-\d}
\draw (-\d/2+\r*0.5,\y+\r*0.866) -- (0-\r*0.5,\y+\d-\r*0.866); 
\draw (\d/2-\r*0.5,\y+\r*0.866) -- (0+\r*0.5,\y+\d-\r*0.866); 
\newcommand{\xx}{3}
\renewcommand{\x}{-\d/2+\xx}
\renewcommand{\y}{0}
\draw (\x,\y)[color=black, very thick] circle (\r);
\draw (\x,\y) node{$a$};
\draw ({\x-\r},\y) -- (\x-\r-\a,0);
\draw (\x+\r,\y) -- (\x+\r+\a,0);
\draw (\x,\y-\r) -- (\x,\y-\r-\a);
\draw (\x+\r*0.5,\y+\r*0.866) -- (\xx-\r*0.5,\y+\d-\r*0.866);
\renewcommand{\x}{\d/2+\xx}
\renewcommand{\y}{0}
\draw (\x,\y)[color=blue, very thick] circle (\r);
\draw (\x,\y) node{$\phi$};
\draw ({\x+\r},\y) -- (\x+\r+\a,0);
\draw (\x,\y-\r) -- (\x,\y-\r-\a);
\draw (\x-\r*0.5,\y+\r*0.866) -- (\xx+\r*0.5,\y+\d-\r*0.866);
\renewcommand{\x}{\xx}
\renewcommand{\y}{\d}
\draw (\x,\y)[color=red, very thick] circle (\r); 
\draw (\x,\y) node{$s$};
\draw (\x,\y+\r) -- (\x,\y+\r+\a);
\draw (0,\y+\r+\a) edge[in=90,out=90] (\x,\y+\r+\a);  
\draw (-\d/2+\r*0.5,-\d+\r*0.866) edge[in=-90,out=-120] (\xx-\d/2,-\r-\a);  
\draw (\d/2-\r*0.5,-\d+\r*0.866) edge[in=-90,out=-70] (\xx+\d/2,-\r-\a);  
\draw (-\r-\a,\d) edge[in=180,out=180] (\xx-\r-\a-\a,0); 
\draw (\r+\a,\d) edge[in=0,out=0] (\xx+\r+\a+\a,0); 
\end{tikzpicture} 
    \hspace{-5mm}
    = \Re \, s^\dag T s.
\hspace{-5mm}
\end{equation}
We note here that the fitness function above is not the same as the one in Eq.~\eqref{eq:lambda1}, in which the $s$ tensor plays no role. 
\iMMP\,is thus not variationally optimal. 
This way of solving Eq.~\eqref{eq:lambda1} approximately is entirely motivated by the finite MM, and as shown later, is very close to being variationally optimal.

Here we describe the steps of \iMMP\, which consists of maximizing Eq.~\eqref{eq:iMM_polar} alternately:
\begin{enumerate}
    \item When $s$ and $a$ are fixed, one forms the environment $E_\phi$ of $\phi$ in $F$ such that $F = \Tr(E_\phi^\dag \phi)$. 
    Here $\phi$ is viewed as a matrix with its incoming and outgoing indices grouped as the row and column index respectively.
    $E_\phi$ is grouped as a matrix such that its row and column index contract with the row and column index of $\phi$ respectively.
    Then the optimal isometry $\phi'$ maximizing $F$ is given by the polar decomposition of $E_\phi$: $\phi' =\argmax_\phi F  = U_\phi$ where $U_\phi P_\phi = E_\phi$ is the polar decomposition of $E_\phi$ \cite{evenbly2009MERA}.  
    
    \item When $s$ and $\phi$ are fixed, one analogously find the optimal update $a'=\argmax_a F$ through the polar decomposition of $E_a$.  
    
    \item When $a$ and $\phi$ are fixed, $F =\Re \, s^\dag T s= s^\dag \left( (T+T^\dag)/2 \right)  s$, and optimal $s'=\argmax_s F$ is given by the leading eigenvector of the Hermitian matrix $(T+T^\dag)/2$. 
\end{enumerate}
We repeat steps 1 to 3 to update $\phi, a, s$ many times until convergence is reached within a threshold.
Note that $F$ is strictly increasing at each step.
Like DMRG, this style of alternate optimization is not convex, and a good initial guess, given by \iMML\,, can speed up the optimization greatly.
Typically such an optimization results in local extremum.

\subsection{\iMMM}
\label{sec:iMMM}

The \iMMM\ method maximizes the overlap density $\Re\ \lambda_1$ through alternatively updating $a$ and $\phi$ tensor until convergence.
We describe the two alternative steps of \iMMM\ as follows:
\begin{enumerate}
    \item Given $a$, which determines an matrix product operator (MPO) $A^\dagger$ acting on the state $\ket{\Psi}$, the optimal update of $\phi$ is determined by the variational MPO-MPS compression algorithm developed for uniform matrix product states~\cite{vanhecke2021tangent}.
    The update for $\phi$ is optimal since the compression algorithm maximizes the overlap density.
    We perform only one update step in the MPO-MPS compression algorithm~\cite{vanhecke2021tangent} instead of finding the converged solution.
    
    \item Given $\phi$, we linearize the overlap  $\Re\ \langle \Psi | A \Phi \rangle$, i.e., viewing each $a$ tensor in $\ket{A}$ as an independent tensor, and find the update for $a$ by the polar decomposition over the environment of $a$ tensor,
    \begin{equation}
        \label{eq:Ea}
        E_{a} = \newcommand{\LL}{1}      
\renewcommand{\d}{1.0}   
\renewcommand{\r}{0.25}   
\renewcommand{\a}{0.5}   
\begin{tikzpicture}[baseline = (X.base),every node/.style={scale=0.6},scale=0.6]
\newcommand{\x}{0}
\newcommand{\y}{0}
\draw (0,\d/2) node (X) {};
\renewcommand{\x}{0}
\renewcommand{\y}{\d}
\draw (\x,\y)[color=blue, very thick] circle (\r);
\draw (\x,\y) node{$\overline\psi$};
\draw (\x+\r,\y) -- (\x+\r+\a,\y); 
\draw (\x-\r-\a,\y) -- (\x-\r,\y); 
\draw (\x,\y+\r) -- (\x,\y+\r+\a);
\renewcommand{\x}{0}
\renewcommand{\y}{0}
\draw (\x,\y) circle (\r);
\draw (\x,\y) node{$L$};
\draw (\x,\y+\r) -- (\x,\y-\r+\d); 
\renewcommand{\y}{-\d}
\draw (-\d/2+\r*0.5,\y+\r*0.866) -- (0-\r*0.5,\y+\d-\r*0.866); 
\draw (\d/2-\r*0.5,\y+\r*0.866) -- (0+\r*0.5,\y+\d-\r*0.866); 
\newcommand{\xx}{3}
\renewcommand{\x}{-\d/2+\xx}
\renewcommand{\y}{0}
\draw (\x+\r,\y) -- (\x+\r+\a,0);
\draw (\x+\r*0.5,\y+\r*0.866) -- (\xx-\r*0.5,\y+\d-\r*0.866);
\renewcommand{\x}{\d/2+\xx}
\renewcommand{\y}{0}
\draw (\x,\y)[color=blue, very thick] circle (\r);
\draw (\x,\y) node{$\phi$};
\draw ({\x+\r},\y) -- (\x+\r+\a,0);
\draw (\x,\y-\r) -- (\x,\y-\r-\a);
\draw (\x-\r*0.5,\y+\r*0.866) -- (\xx+\r*0.5,\y+\d-\r*0.866);
\renewcommand{\x}{\xx}
\renewcommand{\y}{\d}
\draw (\x,\y) circle (\r); 
\draw (\x,\y) node{$R$};
\draw (\x,\y+\r) -- (\x,\y+\r+\a);
\draw (0,\y+\r+\a) edge[in=90,out=90] (\x,\y+\r+\a);  
\draw (-\d/2+\r*0.5,-\d+\r*0.866) edge[in=-90,out=-120] (\xx-\d/2,-\r-\a);  
\draw (\d/2-\r*0.5,-\d+\r*0.866) edge[in=-90,out=-70] (\xx+\d/2,-\r-\a);  
\draw (-\r-\a,\d) edge[in=180,out=180] (\xx-\r-\a-\a,0); 
\draw (\r+\a,\d) edge[in=0,out=0] (\xx+\r+\a+\a,0); 
\end{tikzpicture}
    \end{equation}
    where $L^T$ (row vector) and $R$ (column vector) are the left and right leading eigenvectors of mixed transfer matrix $T$.
\end{enumerate}

We repeat steps 1 to 2 to update $\phi, a$ until convergence is reached within a threshold or the overlap density $\Re\ \lambda_1$ starts increasing in the update for the $a$ tensor.
Note that the update in step 2 usually, but is not guaranteed to, increase the overlap density.
In fact, it is related to a gradient ascent update~\cite{hauru2021riemannian,luchnikov2021qgopt} on tensor $a$.
We observe this update is efficient in increasing the overlap at the initial stage but is slower in the final stage of the convergence comparing to a non-linear conjugate gradient update.
Overall, this method tends to give slightly more accurate results at the cost of slightly longer runtimes compared to \iMMP .

\subsection{\iMMC}
\label{sec:iMMC}

An alternative way to maximize the overlap density $\Re\ \lambda_1$ is to perform  non-linear conjugate gradient ascent on isometries $a$ and $\phi$. 
To respect the isometric constraint on $A$ and $\Phi$, we parametrize the tensors as $a = U_a a_0 = \exp(X_a) a_0$ and $\phi = U_\phi \phi_0 = \exp(X_\phi) \phi_0$, where $a_0$ and $\phi_0$ are any fixed isometries, $U_a$ and $U_\phi$ are unitary matrices acting on the incoming legs of $a$ and $\phi$, respectively, and $X_a$ and $X_\phi$ are anti-Hermitian matrices.
Assuming all bonds to be dimension $\chi$, $X_a$ is a $\chi^2 \times \chi^2$ matrix, while $X_\phi$ is $\chi^3 \times \chi^3$.
The variational space is now the vector space of the anti-Hermitian matrices $X_a$ and $X_\phi$, and conjugate gradient ascent can readily be applied.

Denote the fitness function as 
\begin{equation}
    O[X_a, X_\phi] \equiv \Re\ \lambda_1 (T) \equiv \Re \ \lambda_1(T_{\psi:a[X_a] \phi[X_\phi]}).   
\end{equation}
The change in the objective due to $dX_a$ can be computed as
\begin{align}
    \label{eq:dO}
    dO &= \Re L^T \cdot dT \cdot R = \Re \,\, \tTr(E_{X_a}, dX_a) \\
    &= \tTr(\Re E_{X_a}, d\Re X_a) - \tTr(\Im E_{X_a}, d\Im X_a) \nonumber
\end{align}
where $L^T$ (row vector) and $R$ (column vector) are the left and right leading eigenvectors of $T$. We assume they are normalized so that $L^T R = 1$.
Above, $\tTr$ denotes tensor contraction, and $E_{X_a}$ is the environment of $X_a$ in the tensor contraction:  
\begin{equation}
    \label{eq:EXa}
    E_{X_a} = \newcommand{\LL}{1}      
\renewcommand{\d}{1.0}   
\renewcommand{\r}{0.25}   
\renewcommand{\a}{0.5}   
\begin{tikzpicture}[baseline = (X.base),every node/.style={scale=0.6},scale=0.6]
\newcommand{\x}{0}
\newcommand{\y}{0}
\draw (0,\d/2) node (X) {};
\renewcommand{\x}{0}
\renewcommand{\y}{\d}
\draw (\x,\y)[color=blue, very thick] circle (\r);
\draw (\x,\y) node{$\overline\psi$};
\draw (\x+\r,\y) -- (\x+\r+\a,\y); 
\draw (\x-\r-\a,\y) -- (\x-\r,\y); 
\draw (\x,\y+\r) -- (\x,\y+\r+\a);
\renewcommand{\x}{0}
\renewcommand{\y}{0}
\draw (\x,\y) circle (\r);
\draw (\x,\y) node{$L$};
\draw (\x,\y+\r) -- (\x,\y-\r+\d); 
\renewcommand{\y}{-\d}
\draw (-\d/2+\r*0.5,\y+\r*0.866) -- (0-\r*0.5,\y+\d-\r*0.866); 
\draw (\d/2-\r*0.5,\y+\r*0.866) -- (0+\r*0.5,\y+\d-\r*0.866); 
\newcommand{\xx}{3}
\renewcommand{\x}{-\d/2+\xx}
\renewcommand{\y}{0}
\draw (\x,\y)[color=black, very thick] circle (\r);
\draw (\x,\y) node{$a$};
\draw ({\x-\r},\y) -- (\x-\r-\a,0);
\draw (\x+\r,\y) -- (\x+\r+\a,0);
\draw (\x,\y-\r) -- (\x,\y-\r-\a);
\draw (\x+\r*0.5,\y+\r*0.866) -- (\xx-\r*0.5,\y+\d-\r*0.866);
\renewcommand{\x}{\d/2+\xx}
\renewcommand{\y}{0}
\draw (\x,\y)[color=blue, very thick] circle (\r);
\draw (\x,\y) node{$\phi$};
\draw ({\x+\r},\y) -- (\x+\r+\a,0);
\draw (\x,\y-\r) -- (\x,\y-\r-\a);
\draw (\x-\r*0.5,\y+\r*0.866) -- (\xx+\r*0.5,\y+\d-\r*0.866);
\renewcommand{\x}{\xx}
\renewcommand{\y}{\d}
\draw (\x,\y) circle (\r); 
\draw (\x,\y) node{$R$};
\draw (\x,\y+\r) -- (\x,\y+\r+\a);
\draw (0,\y+\r+\a) edge[in=90,out=90] (\x,\y+\r+\a);  
\draw (-\d/2+\r*0.5,-\d+\r*0.866) edge[in=-90,out=-120] (\xx-\d/2,-\r-\a*2);  
\draw (\d/2-\r*0.5,-\d+\r*0.866) edge[in=-90,out=-70] (\xx+\d/2,-\r-\a);  
\draw (-\r-\a,\d) edge[in=180,out=180] (\xx-\r-\a-\a*2,0); 
\draw (\r+\a,\d) edge[in=0,out=0] (\xx+\r+\a+\a,0); 
\end{tikzpicture}
\end{equation}
Note the implicit dependency of $L$ and $R$ on $X_a$ in the current form give zero contribution to the change in the objective~\cite{xie2020eigensolver}.

The ascent direction for the maximization is thus given by the derivative: 
\begin{equation}
    \frac{\partial O}{\partial \Re X_a} = \Re E_{X_a},  
    \hspace{5mm} 
    \frac{\partial O}{\partial \Im X_a} = -\Im E_{X_a}
\end{equation}
Thus, the ascent direction of $X_a$ is $\overline{E_{X_a}}$~\footnote{
Here we treat the real and imaginary part of $X_a$ as independent variables. An alternative way to arrive at the same result is to treat $X_a$ as complex-valued variables and the ascent direction is given by ${dO}/{d \overline{X_a}} = \frac{1}{2} R^\dagger (dT^\dagger /{d \overline{X_a}}) L = \frac{1}{2}\overline{E_{X_a}}$.
}.
In fact, $dO$ is manifestly positive if $dX_a = \overline{E_{X_a}}$, where the overline denotes complex conjugation. 
Note that the $E_{X_a}$ computed in Eq.~\eqref{eq:EXa} and the ascent direction $\overline{E_{X_a}}$ are generally not anti-Hermitian.
Therefore, one needs to anti-Hermitian-ize $dX_a$ so that the updated $X_a$ is still anti-Hermitian.
One can analogously compute the ascent direction for $X_\phi$.

With these ingredients, the conjugate gradient ascent is done as follows.
\begin{enumerate}
    \item At CG step $k$, compute $E_{X_a}$ and $E_{X_\phi}$ and anti-Hermitian-ize them. 
    We overload notation and use $E_X$ to refer to the anti-Hermitian environments below.
    \item Set the ascent direction, $H$, using the gradient and the ascent direction from the previous step: 
        \begin{align}
            \begin{split}
                H_a(k) = \overline{E_{X_a}} - \beta H_a(k-1)
                \\
                H_\phi(k) = \overline{E_{X_\phi}} - \beta H_\phi(k-1) 
            \end{split}
        \end{align}
    where $\beta$ is determined by a non-linear CG $\beta$-mixer, e.g. Polak-Ribi\'ere. 
    
    \item Parametrize $a(t)$ and $\phi(t)$ along the ascent direction: $a(t) = \exp(t H_a) a$ and $\phi(t) = \exp(t H_\phi) \phi$.
    Via the linesearch algorithm, look for $t_\max$ at which $O(t_\max)$ is maximized along the $t$-curve. 
    This needs the computation of $dO(t)/dt$:
    \begin{align}
        \begin{split}
            \frac{dO(t)}{dt} =& \Re\,\, \tTr(E_{X_{a(t)}}, H_a) 
            \\&+  \Re\,\, \tTr(E_{X_{\phi(t)}}, H_\phi) 
            \end{split}
    \end{align}
    
    \item Update $a \rightarrow a(t_\max)$ and $\phi \rightarrow \phi(t_\max)$ and $k \rightarrow k+1$. 
\end{enumerate}
This process is iterated until $a$ and $\phi$ converge to within the desired threshold.

\subsection{Error measures and a structure theorem}
\label{sec:error}

Before presenting benchmarks and applications of the iMM, we first discuss error measures for the splitting problem in Eq.~\eqref{eq:iQR} and introduce a structure theorem governing the assignment of bond dimensions to tensors $a$ and $\phi$.

As before, let $\ket{\Psi}$ be the input to the iMM, and $A$ and $\ket{\Phi}$ be the output. 
$\ket{\Psi}$ and $\ket{\Phi}$ are always normalized.
We consider the error of the fidelity density, 
\begin{equation}
    \label{eq:total_error}
    \epsilon \equiv 1 - (\Re \, \lambda_1(T))^2 = 1 -  \lambda_1(T)^2,  
\end{equation}
where we assume the iMM algorithm finds $A$ and $\ket{\Phi}$ such that the dominant eigenvalue of $T$ is real.
To motivate this definition, let us consider the example of splitting a finite and uniform system of size $L$, as in Eq.~\eqref{eq:finite_MM}.
The error of the splitting is given by 
\begin{align}
    \label{eq:epsilon}
    \begin{split}
        \norm{\ket{\Psi}-A\ket{\Phi}}^2 =& 2 - 2 \Re \bra{\Psi}A\ket{\Phi} \\
        &\approx 2-2 (\sqrt{1-\epsilon})^L \approx \epsilon L.
    \end{split}
\end{align}
Thus, $\epsilon$ is the intensive error density.

It can be shown that $\epsilon$ is a sum of two errors (see Appendix~\ref{sec:error_decomp} for the derivation):
\begin{align}
  \epsilon &= \epsilon_p + \epsilon_t + O(\epsilon_p^2, \epsilon_t^2, \epsilon_p \epsilon_t) 
  \label{eq:error_decomposition}
\end{align}
where 
\begin{align}
    \epsilon_p &\equiv 1 - \lambda_1(T_{A^\dag\Psi:A^\dag\Psi}), \nonumber\\
    \epsilon_t &\equiv 1 - \left( \lambda_1(T_{\widetilde{A^\dag\Psi}:\Phi}) \right)^2.
    \label{eq:error_defs}
\end{align}
Here $\widetilde{A^\dag \ket{\Psi}} = A^\dag \ket{\Psi}/\norm{A^\dag \ket{\Psi}}$ is normalized.
$\epsilon_p$ measures the norm that $A^\dag \ket{\Psi}$ loses due to the projection and $\epsilon_t$ measures the truncation error due to approximating $\widetilde{A^\dag\ket{\Psi}}$ with the MPS $\ket{\Phi}$.

In practice, these errors guide the choice of internal bond dimensions of the iMM. 
Let us denote the bond dimensions of the iMM as the following: 
\begin{equation}
    \label{eq:bond_dims}
    \renewcommand{\d}{1.0}   
\renewcommand{\r}{0.25}   
\renewcommand{\a}{0.5}   
\newcommand{\dH}{1.0}   
\newcommand{\aH}{0.5}
\newcommand{\x}{0}
\newcommand{\y}{0}
\begin{tikzpicture}[baseline = (X.base),every node/.style={scale=1.0},scale=1.0]
\renewcommand{\x}{0}
\renewcommand{\y}{0}
\draw (0, -0.1) node (X) {};

\draw [color=blue, very thick](\x,\y) circle (\r);
\draw (\x,\y) node {$\psi$};
\draw [midarrow={latex reversed}](\x+\r,\y) -- (\x+\r+\aH,\y); 
\draw [midarrow={latex}](\x-\r-\aH,\y) -- (\x-\r,\y); 
\draw [midarrow={latex reversed}](\x,\y+\r) -- (\x,\y+\r+\a);
\draw [midarrow={latex}](\x,\y-\r) -- (\x,\y-\r-\a);
\draw (\x,\d) node{$\chi_0$}; 
\draw (\x,-\d) node{$\chi_0$}; 
\draw (\x-\d,0) node{$\chi_\ell$}; 
\draw (\x+\d,0) node{$\chi_r$}; 
\draw (0,-1.5*\d-0.1) node{$\ket{\Psi}$}; 
\end{tikzpicture}
\approx
\begin{tikzpicture}[baseline = (X.base),every node/.style={scale=1.0},scale=1.0]
\renewcommand{\x}{0}
\renewcommand{\y}{0}
\draw (0, -0.1) node (X) {};

\foreach \i in {0,...,1}
{
    \renewcommand{\x}{\i*\dH}
    \renewcommand{\y}{0}
    \pgfmathsetmacro{\Hdir}{ifthenelse(\i == 0, "latex", "latex reversed")}
    \ifthenelse{\i<1}{
        \draw [color=black, very thick] (\x,\y) circle (\r);\draw (\x,\y) node{$a$};
    }{
        \draw [color=blue, very thick] (\x,\y) circle (\r);\draw (\x,\y) node{$\phi$};
    }
    \draw [midarrow={\Hdir}](\x+\r,\y) -- (\x+\r+\aH,\y); 
    \draw [midarrow={latex}](\x,\y+\r) -- (\x,\y+\r+\a);
    \draw [midarrow={latex}](\x,0-\r-\a) -- (\x,0-\r);
}
\renewcommand{\x}{0}
\renewcommand{\y}{0}
\draw [midarrow={latex reversed}](\x-\r,\y) -- (\x-\r-\aH,\y); 

\draw (0,-\d) node{$\chi_v$};
\draw (0,\d) node{$\chi_v$};
\draw (\d,-\d) node{$\eta$};
\draw (\d,\d) node{$\eta$};
\draw (-\d,0) node{$\chi_\ell$}; 
\draw (2*\d,0) node{$\chi_r$};
\draw (\d/2,0.3) node{$\chi_h$}; 
\draw (0,-1.5*\d-0.1) node{$A$}; 
\draw (\d,-1.5*\d-0.1) node{$\ket{\Phi}$}; 
\end{tikzpicture}.
\end{equation}
Assuming the non-convex optimization in the iMM is successful, increasing $\eta$ decreases $\epsilon_t$, and increasing $\chi_v$ and $\chi_h$ decreases $\epsilon_p$.   
Perhaps less obvious is that increasing $\chi_v$ and $\chi_h$ can also decrease $\epsilon_t$, because $A$, in addition to being a projector, also serves as a disentangler of $\Psi$; see Sec.~\ref{sec:disentangling_iMM}.
That is, if $A$ is well-chosen, $\widetilde{A^\dag \ket{\Psi}}$ will have less entanglement entropy than $\Psi$, and, when truncated to an MPS with smaller bond dimension, will have less truncation error compared to directly truncating $\ket{\Psi}$.
Implicitly in iMM, when minimizing $\epsilon$, the optimization reaches a balance between the projecting ($\epsilon_p$) and the disentangling ($\epsilon_t$) role of $A$ so that their collective effect minimizes $\epsilon$.
The disentangling effect of $A$ will be reflected in $\ket{\Phi}$ having less entanglement entropy than $\ket{\Psi}$, as demonstrated in Sec.~\ref{sec:area_law}.

The above discussion would appear to indicates that, as long as the computational cost is affordable, the internal bond dimensions  should be as large as possible to reduce $\epsilon$. 
In regular iMPS compression algorithms, one can set the bond dimension of the trial state to be as large as desired, and the only side-effect is that the code will run for longer.
This is indeed the case  $\eta$, as it is just the bond dimension of the new MPS $\ket{\Phi}$.
For $\chi_h$, we have to choose $\chi_h \leq \chi_\ell$ in order for $A$ to be an isometry (see Eq.~\eqref{eq:bond_dims}).

However, care must be taken in choosing the bond dimension $\chi_v$.
As we explain below, carelessly increasing $\chi_v$ may cause the failure of both iMM algorithms and subsequent iTEBD algorithms.
We now present a structure theorem to guide the choice of $\chi_v$ given choices for $\eta$ and $\chi_h$.
\begin{thm}
\label{thm:structure}
Suppose the two-sided iMPS $\ket{\Psi}$ can be exactly split via iMM: $\ket{\Psi}=A\ket{\Phi}$, with bond dimensions $\eta$, $\chi_h=\chi_\ell$, and $\chi_v$. 
If the iMM for $\ket{\Psi}$ is performed with bond dimensions $\eta$, $\chi_h$ and $\chi_v'>\chi_v$, there is an exact solution $\ket{\Psi}=A'\ket{\Phi'}$, where $A'$ is an isoTNO such that, for any normalized two-sided iMPS $\ket{B}$, the dominant eigenvalue of $T_{A'\ket{B}:A'\ket{B}}$ is degenerate.
\end{thm}
\begin{proof}
Suppose that a set of internal bond dimensions $\chi_v$, $\chi_h$, and $\eta$ suffices to give an exact iMM splitting: $\ket{\Psi} = A \ket{\Phi}$. 
Then when $\chi'_v > \chi_v$, $\chi'_h = \chi_h$, and $\eta' = \eta$ are used for the same splitting problem, clearly there is an exact iMM splitting $\ket{\Psi} = A'\ket{\Phi'}$, where $\ket{\Phi'}$ is the same MPS as $\ket{\Phi}$ up to some gauge difference on the virtual bond, and $A'$ is $A$ enlarged from bond dimension $\chi_v$ to $\chi'_v$ while maintaining the isometric constraints.
Let $a$ and $a'$ be the tensors making up $A$ and $A'$. 
When $\chi_h = \chi_l$, namely when $a$ is unitary, enlarging the bond dimension $\chi_v$ of the isometry while maintaining the isometric constraint yields $a' = a\oplus a^\bot$, where $a^\bot$ is the orthogonal complement to $a$ in the enlarged space; this is shown in Appendix~\ref{sec:isofill}.
Then, for {\it any} normalized two-sided MPS $\ket{B}$, 
\begin{align}
    \begin{split}
        T_{A'B:A'B} &= T_{AB:AB} \oplus T_{A^\bot B:AB} \\
        &\oplus T_{AB:A^\bot B} \oplus T_{A^\bot B:A^\bot B}.
    \end{split}
\end{align}
In particular, the spectrum of both $T_{AB:AB}$ and $T_{A^\bot B:A^\bot B}$ contains a copy of the spectrum of $T_{B:B}$.
As $\ket{B}$ is normalized, $T_{A'B:A'B}$ will have at least two eigenvalues equal to 1, meaning that $A'\ket{B}$ becomes a ``cat-state'' (non-injective) iMPS. 
\end{proof}

If in addition, the solution to iMM with $\chi'_v$ is unique then the cat-state described above will be the only exact solution, and if the iMM optimization is successful, the optimization will result in such cat-states.
In practice, we find that such cat-states, if they exist, are always found. 
Such states are pathological because the convergence of many iMPS algorithms, including iDMRG and iTEBD, scales inversely with the gap of the transfer matrix of the iMPS, which for $A'\ket{B}$ is zero.

In particular, if $\eta=\chi_0$ and $\chi_h = \chi_l$, then $\chi_v = 1$ is sufficient to give an exact iMM splitting.
As a consequence of Theorem~\ref{thm:structure}, any larger $\chi_v$ will lead to degeneracies in the spectrum.
Choosing $\chi_h$ strictly less than $\chi_l$ is a one way to avoid the conditions of the Theorem~\ref{thm:structure} from being satisfied in simulations --- this motivates the choice of the boundary bond dimension, $D_b$, to be larger than 1.

\section{iMM Benchmarks}
\label{sec:benchmarks}

Having introduced algorithms for performing the iMM splitting procedure depicted in Eq.~\eqref{eq:iQR} and decomposed the resulting error as the sum of a truncation and projection error terms, we now perform several experiments.
The first is to compare the four different iMM algorithms introduced in Sec.~\ref{sec:iMM} and various combinations of these algorithms.
From this, we conclude that the combination of \iMML\, and \iMMP\, produces accurate results without significant computational cost, and we use this combination for all future experiments.
We then investigate the disentangling properties of the isoTNO $A$, finding that varying its vertical bond dimension $\chi_v$ can greatly decrease both $\epsilon_p$ and $\epsilon_t$.
Finally, we repeatedly apply iMM to a strip isoTNS, sweeping back and forth, and find that the accumulated error saturates after about 20 iterations.

In this section and those that follow, we consider the 2D transverse field Ising (TFI) model,
\begin{align}
    H &= -\sum_{\langle i, j \rangle} Z_i Z_j - g\sum_{i} X_i,
    \label{eq:2DTFI}
\end{align}
to benchmark the performance of the iMM algorithms and the iTEBD\textsuperscript{2} algorithm.
Unless otherwise noted, we choose $g=3.5$ to be in the paramagnetic phase. 
The critical coupling for this model, obtained via cluster Monte Carlo simulations, in the thermodynamic limit in both directions is $g_C^{2D} \approx 3.04438$ \cite{blote2002TFI}, while from iDMRG the critical coupling on infinite cylinders increases from $g_C^{1D}=1$ towards $g_C^{2D}$ with width \cite{hashizume2022TFI}.

\begin{figure*}[ht!]
\begin{tabular}{>{\centering\arraybackslash}b{0.30\linewidth-1\tabcolsep\relax} >{\centering\arraybackslash}b{0.7\linewidth-1\tabcolsep\relax}}
\begin{align*}
    \renewcommand{\d}{1.0}   
\renewcommand{\r}{0.25}   
\renewcommand{\a}{0.5}   
\newcommand{\dH}{1.0}   
\newcommand{\aH}{0.5}
\newcommand{\x}{0}
\newcommand{\y}{0}
\begin{tikzpicture}[baseline = (X.base),every node/.style={scale=1.0},scale=1.0]
\renewcommand{\x}{0}
\renewcommand{\y}{0}
\draw (0, -0.1) node (X) {};

\draw [color=blue, very thick](\x,\y) circle (\r);
\draw (\x,\y) node {$\psi$};
\draw [midarrow={latex reversed}](\x+\r,\y) -- (\x+\r+\aH,\y); 
\draw [midarrow={latex}](\x-\r-\aH,\y) -- (\x-\r,\y); 
\draw [midarrow={latex reversed}](\x,\y+\r) -- (\x,\y+\r+\a);
\draw [midarrow={latex}](\x,\y-\r) -- (\x,\y-\r-\a);
\draw (\x,\d) node{$128$}; 
\draw (\x,-\d) node{$128$}; 
\draw (\x-\d,0) node{$4$}; 
\draw (\x+\d,0) node{$4$}; 
\draw (0,-1.5*\d-0.1) node{$\ket{\Psi}$}; 

\draw (\x+\d,-\d) node{$\approx$}; 

\foreach \i in {0,...,1}
{
    \renewcommand{\x}{\i*\dH+2}
    \renewcommand{\y}{-2}
    \pgfmathsetmacro{\Hdir}{ifthenelse(\i == 0, "latex", "latex reversed")}
    \ifthenelse{\i<1}{
        \draw [color=black, very thick] (\x,\y) circle (\r);\draw (\x,\y) node{$a$};
    }{
        \draw [color=blue, very thick] (\x,\y) circle (\r);\draw (\x,\y) node{$\phi$};
    }
    \draw [midarrow={\Hdir}](\x+\r,\y) -- (\x+\r+\aH,\y); 
    \draw [midarrow={latex}](\x,\y+\r) -- (\x,\y+\r+\a);
    \draw [midarrow={latex}](\x,\y-\r-\a) -- (\x,\y-\r);
}
\renewcommand{\x}{2}
\renewcommand{\y}{-2}
\draw [midarrow={latex reversed}](\x-\r,\y) -- (\x-\r-\aH,\y); 

\draw (\x,\y-\d) node{$2$};
\draw (\x,\y+\d) node{$2$};
\draw (\x+\d,\y-\d) node{$10$};
\draw (\x+\d,\y+\d) node{$10$};
\draw (\x-\d,\y) node{$4$}; 
\draw (\x+2*\d,\y) node{$4$};
\draw (\x+\d/2,\y+0.3) node{$2$}; 
\draw (\x,\y-1.5*\d-0.1) node{$A$}; 
\draw (\x+\d,\y-1.5*\d-0.1) node{$\ket{\Phi}$}; 
\end{tikzpicture}
\end{align*}
\captionof{figure}{Splitting problem considered for iMM method benchmarking results in Table~\ref{tab:single}. 
This two-sided iMPS of $\chi=128$ and $\chi_\ell = \chi_r = 4$ represents the ground state of 2D TFI model on $L_x=4$ infinite strip, as described in Sec.~\ref{sec:single_iMM}.}
\label{fig:single}
    &
\renewcommand{\arraystretch}{1.3}
\begin{tabular}{
|| >{\centering\arraybackslash}p{2.5cm} | >{\centering\arraybackslash}p{3cm} | >{\centering\arraybackslash}p{3cm} | >{\raggedleft\arraybackslash}p{1cm} c >{\raggedright\arraybackslash}p{1cm} ||
}
 \hline
 \hline
 iMM Method & $\epsilon_p$ & $\epsilon_t$ & \multicolumn{3}{c ||}{time (s)} \\
 \hline
\texttt{L} & 1.8e-03 $\pm$ 1.0e-06 & 3.4e-05 $\pm$ 8.2e-06 & 0.7 & $\pm$ & 0.01 \\ 
\hline
\texttt{P} & 1.3e-03 $\pm$ 1.4e-03 & 2.3e-05 $\pm$ 4.4e-06 & 38 & $\pm$ & 15 \\ 
\hline
 \texttt{M} & 1.0e-02 $\pm$ 6.7e-03 & 1.4e-03 $\pm$ 1.7e-03 & 12 & $\pm$ & 21 \\ 
\hline
\texttt{C} & 2.2e-01 $\pm$ 2.6e-01 & 3.7e-01 $\pm$ 3.6e-01 & 21 & $\pm$ & 10 \\ 
\hline
\texttt{LP} & 6.2e-04 $\pm$ 1.9e-08 & 2.6e-05 $\pm$ 1.8e-08 & 11 & $\pm$ & 0.1 \\ 
\hline
\texttt{LM} & 5.0e-04 $\pm$ 2.0e-10 & 9.4e-06 $\pm$ 5.7e-10 & 62 & $\pm$ & 0.6 \\ 
\hline
\texttt{LC} & 5.0e-04 $\pm$ 1.2e-07 & 4.2e-05 $\pm$ 1.1e-06 & 29 & $\pm$ & 1.9 \\ 
\hline
\texttt{LPM} & 4.6e-04 $\pm$ 2.0e-08 & 5.8e-07 $\pm$ 8.8e-11 & 12 & $\pm$ & 0.05 \\ 
\hline
\texttt{LPC} & 4.6e-04 $\pm$ 4.2e-06 & 5.2e-06 $\pm$ 2.6e-07 & 19 & $\pm$ & 1.4 \\
\hline
\texttt{LMC} & 5.0e-04 $\pm$ 2.1e-11 & 9.3e-06 $\pm$ 2.1e-11 & 72 & $\pm$ & 0.4 \\ 
\hline 
\texttt{LPMC} & 4.5e-04 $\pm$ 3.7e-10 & 2.4e-07 $\pm$ 1.0e-10 & 20 & $\pm$ & 0.07 \\ 
\hline
\hline
\end{tabular}
\captionof{table}{Projection and truncation errors as defined in Eq.~\eqref{eq:error_defs} in Sec.~\ref{sec:single_iMM} and runtimes for different iMM methods applied to the splitting problem in Fig.~\ref{fig:single}. 
Results from five different runs are average to give the standard deviations. We conclude that the combination of \iMML\, and \iMMP\, provides an accurate result without incurring a large computational cost.
}
\label{tab:single}
\end{tabular}
\end{figure*}

\subsection{Benchmark for a single run of iMM}
\label{sec:single_iMM}

We begin by investigating the performance of the iMM algorithms introduced earlier. 
The input state is the ground state of the 2D TFI model on an infinite strip of width $L_x=4$ obtained from iDMRG with bond dimension $\chi=128$.
We contract the four iMPS tensors in the unit cell (one row) to one iMPS tensor with four physical legs, each of dimension $d=2$.
We then group the two physical legs corresponding to the left two sites in the row into the left leg and the remaining two legs into the right leg of the two-sided iMPS. 
This produces the iMPS shown in Fig.~\ref{fig:single}.
The splitting problem considered involves both projection and truncation, as $\chi_h < \chi_\ell$.

We compare all four iMM algorithms discussed above and various combinations of them to the two-sided iMPS.
We report errors and the runtimes on a standard workstation, averaged over five runs, in Table~\ref{tab:single}.
In this discussion, we use the shorthand notation \texttt{L} to represent \iMML\,, \texttt{P} for \iMMP\,, \texttt{M} for \iMMM\,, and \texttt{C} for \iMMC\, for convenience.
We find that \texttt{L} is indeed the fastest but does not achieve the accuracy of other standalone methods.
However, even though this method is not guaranteed to converge, the gauge fixing procedure reduces variance, indicating that this method can be used to provide a stable starting point for further methods.
Seeding methods \texttt{P} and \texttt{M} with the $a$ and $\phi$ tensors provided by \texttt{L}, yielding methods \texttt{LP} and \texttt{LM}, decreases the error and standard deviation, while also for \texttt{LP} greatly reducing the runtime, compared to \texttt{P} and \texttt{M} alone.
Further improving the results with \texttt{C} is possible but the improvements are not significant yet have added computational cost.
Thus for all future experiments, we use \texttt{LP}.

Additionally, empirically we find that when $\chi_v$ is large, it is best to use \texttt{L} to produce starting solutions with a smaller vertical bond dimension, say $\chi_v'=4$, and then isometrically expand the vertical dimensions of the $a$ tensor gradually.
At each intermediate bond dimension between $\chi_v'$ and $\chi_v$, we use \texttt{P} to improve the result. 
We then expand the isometric tensor $a$ by viewing it as a matrix by grouping the incoming legs into a row index and the outgoing legs into a column index. 
The incoming row index can be increased by zero padding, while the outgoing leg must be increased by adding orthogonal columns.
Such a gradual iMM procedure improves the stability and performance of the splitting as $\chi_v$ grows.
Theorem~\ref{thm:structure} indicates that isometric filling an $A$ produced by an already exactly splitting will lead to a degenerate spectrum, so we do the expansion only if either $\eta < \chi_0$ or $\chi_h < \chi_\ell$.

\subsection{Disentangling effect of $A$}
\label{sec:disentangling_iMM}

We now investigate the disentangling properties of $A$. 
As discussed in Sec.~\ref{sec:error}, by tuning the bond dimensions $\chi_v$, $\chi_h$, and $\eta$, we can control the projection and truncation errors of the splitting procedure; see Eq.~\eqref{eq:bond_dims} for definitions of these bond dimensions.
We know that (1) increasing $\eta$ will decrease $\epsilon_t$ as this error is simply MPS truncation error; and that (2) increasing $\chi_h$ up to the maximally allowed $\chi_\ell$ will reduce the projection error.
Here, we want to understand the effect of increasing $\chi_v$ on both type of errors.
Following a similar setup as in previous section, we obtain the result shown in Fig.~\ref{fig:disentangling}.
We find that increasing $\chi_v$ can significantly reduce both types of error with a fixed $\eta$ and $\chi_h$.
While we expect the decrease in $\varepsilon_p$ with increasing $\chi_v$, the less expected decrease in $\varepsilon_t$ indicates $A$ has a disentangling effect on $\Phi$, making it easier to truncate and thus reducing $\epsilon_t$. 
Increasing $\chi_v$ increases the complexity of the isoTNO, and thus a more complex operator allows for more successful disentangling.

Since the $A$ is acting only on one side of the two-sided MPS, it is similar to the disentangler used for the purified wavefunction of a mixed state, which only acts on the auxiliary indices.
This indicate that iMM algorithm can also be applied to produce more efficient purified MPS description, a problem that was considered in the time evolution of thermofield double states~\cite{hauschild2018finding}.
Another implication is that there will be a minimal (purification) entanglement which cannot be removed from the two-sided MPS by applying $A^\dagger$.
We see such effect in Fig.~\ref{fig:disentangling} that the errors of both truncation $\epsilon_t$ and projection $\epsilon_p$ saturate with the increasing $\chi_v$.

\begin{figure}
\centering
\input{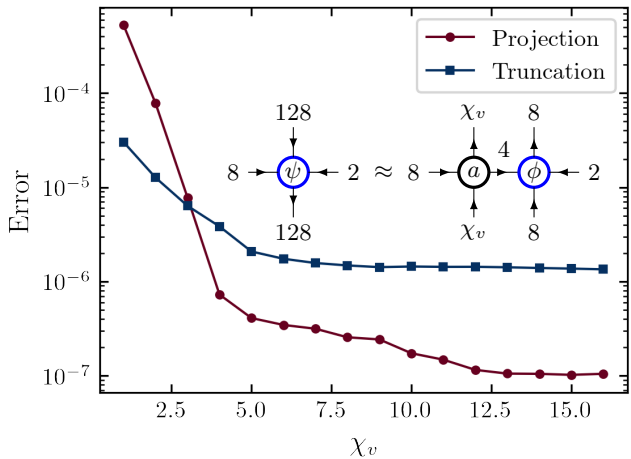}
\caption{
Projection error (red), $\epsilon_p$, and truncation error (blue), $\epsilon_t$, as a function of $\chi_v$, as calculated by Eq.\eqref{eq:error_defs}, using the combination of \iMML\, and \iMMP. 
\textbf{(Inset)} Splitting problem considered. The two-sided iMPS is the ground state of the $L_x=4$ 2D TFI from iDMRG, similar to Fig.~\ref{fig:single}. 
But three physical indices are grouped to the left and one to the right.}
\label{fig:disentangling}
\end{figure}

\subsection{Repeated application of iMM to an isoTNS}
\label{sec:repeated_iMM}

The iMM algorithm performs the splitting in Eq.~\eqref{eq:iQR} approximately and so, unlike moving the OC by QR decomposition in 1D MPS, iMM will not exactly preserve the state.
Thus it is important to quantify how repeated applications of iMM affect an isoTNS and perturb it away from the starting state.
Here we consider specifically the case where the starting state is an approximate ground state of the 2D TFI model on an $L_x=6$ strip represented by an isoTNS.
We measure the deviation of the state after consecutive iMM from the original state.

Suppose that the current (after $n-1$ iterations) isoTNS $\ket{\Psi^{(n-1)}}$ has its orthogonality column $\ket{\Psi^{(n-1)}_{1}}$ as the leftmost column of the strip; we then write the tensor network as $\ket{\Psi^{(n-1)}} = \ket{\Psi^{(n-1)}_{1}} B^{(n-1)}_{2} \ldots B^{(n-1)}_{L_x}$, where $B$ denotes an isoTNO with horizontal isometry arrows pointing left.
We fuse the columns $\ket{\Psi^{(n-1)}_{1}}$ and $B^{(n-1)}_{2}$ to form a doubled column $\ket{\Psi^{(n)}_{1,2}}$ with two physical sites per tensor;
we split this doubled orthogonality column using iMM into the isoTNO $A^{(n)}_{1}$ with rightward pointing isometry arrows and the new orthogonality column $\ket{\Psi^{(n)}_{2}}$, with each column having a physical index.
We perform this procedure of merging and splitting two columns a total of $L_x-1$ times, moving the orthogonality column entirely to the right of the strip.
This completes one sweep and produces a new isoTNS $\ket{\Psi^{(n)}} = A^{(n)}_{1} \ldots A^{(n)}_{L_x-1} \ket{\Psi^{(n)}_{L_x}}$. 
We can now repeat the process moving to the left (or in practice horizontally mirroring the isoTNS and again moving to the right) to produce $\ket{\Psi^{(n+1)}}$. 
We measure the error of fidelity density as in Eq.~\eqref{eq:total_error}, where the transfer matrix $T$ is formed by one row of the original state $\ket{\Psi^{(0)}}$ and the state at $n$-iteration $\ket{\Psi^{(n)}}$.

\begin{figure}
    \centering
    \includegraphics[width=\columnwidth]{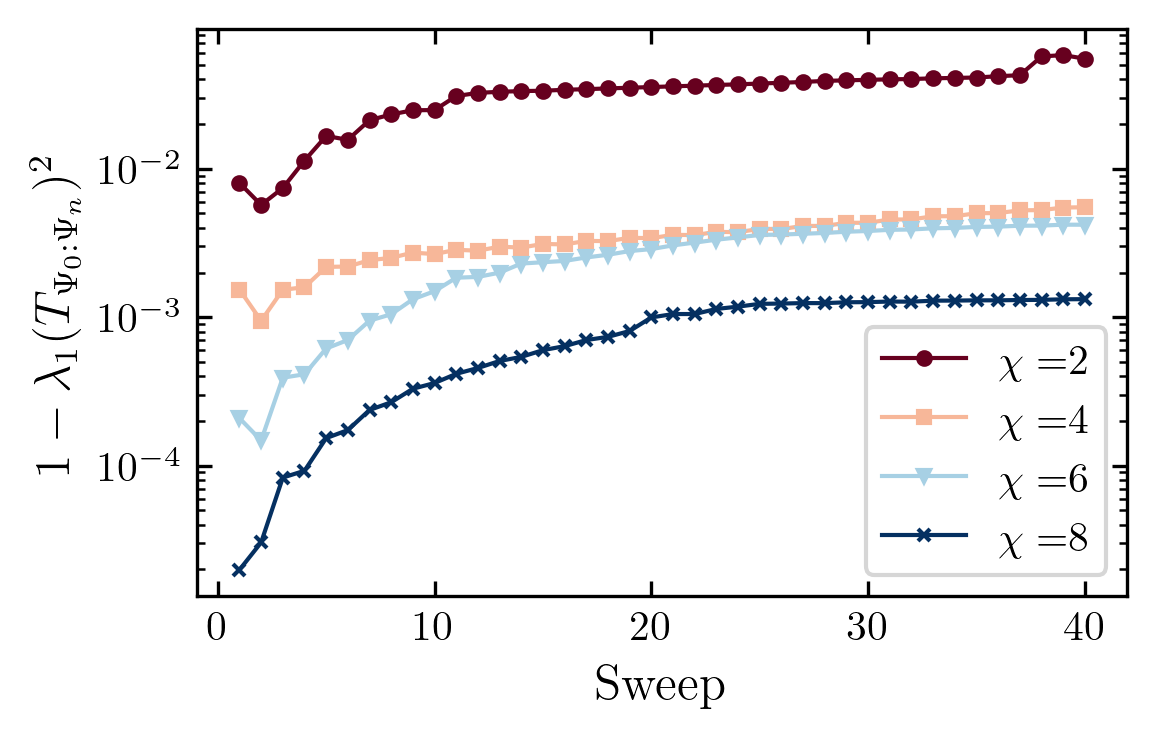}
    \caption{Error between original isoTNS $\ket{\Psi^{(0)}}$ and isoTNS $\ket{\Psi^{(n)}}$ produced by repeated $n$ iMM sweeps. $\ket{\Psi^{(0)}}$ is an isoTNS with $L_x=6$, $\chi_{\text{init}} = 2$, $D_b = 1$ produced by peeling of iDMRG GS of the 2D TFI model, as described in Sec.~\ref{sec:iDMRG_to_isoTNS}.
    }
    \label{fig:repeated_iMM}
\end{figure}

The results of this procedure for an $L_x=6$, $\chi_{\text{init}}=2$ isoTNS is shown in Fig.~\ref{fig:repeated_iMM}. 
This isoTNS is produced by the peeling procedure discussed in Sec.~\ref{sec:iDMRG_to_isoTNS} and thus is an approximation of the 2D TFI ground state. 
We choose bond dimensions in the iMM splitting to be $\chi = \chi_v = \chi_h$ and $\eta = 24$.
After each sweep, the maximum bond dimension in $\ket{\Psi_{n}}$ is $\chi$.
Some bonds will have values smaller than $\chi$ so that the isometric conditions on each tensor are satisfied.
We find that for each of the bond dimensions used in the iMM and thus the bond dimensions of the resulting isoTNS, the accumulated error saturates after 20 sweeps. 
Additionally, we find that increasing $\chi$ decreases the error as expected, as the iMMs in the sweep can be done more accurately and thus have a less corrupting effect on the state.

\section{Transforming iMPS into 2D isoTNS}
\label{sec:iDMRG_to_isoTNS}

As the first application of the iMM algorithm, we show that we can transform an iMPS representing a 2D ground state into a 2D isoTNS with an approximation error controlled by the bond dimension.
We further investigate the entanglement properties of the states produced during the procedure.

\subsection{Peeling}
\label{sec:peeling}

To obtain an isoTNS approximating an iMPS, we use a peeling process, depicted in Fig.~\ref{fig:peeling}, analogous to that used for finite isoTNS in~\cite{zaletel2020isometric}.
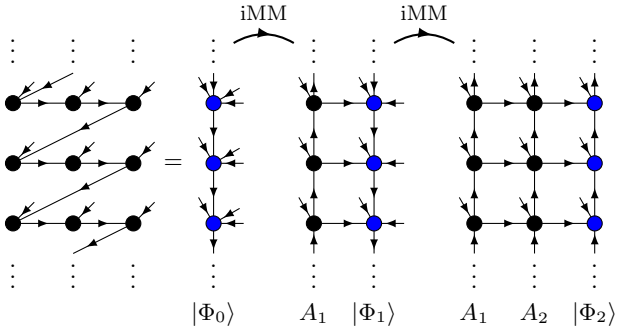
\begin{figure}[b!]
    \centering
    \newcommand{\LLy}{2}      
\newcommand{\LLx}{2}      
\newcommand{\ic}{0}
\newcommand{\jc}{10}
\renewcommand{\d}{0.8}   
\renewcommand{\r}{0.1} 
\renewcommand{\a}{0.6}   
\newcommand{\x}{0}
\newcommand{\y}{0}

\begin{tikzpicture}[baseline = (X.base),every node/.style={scale=1.0},scale=1.0]
		
\draw (\LLx*\d/2,\LLy*\d/2) node (X) {};

\foreach \j in {0,...,\LLy} {
    \foreach \i in {0,...,\LLx} {
        \renewcommand{\x}{\i*\d}
        \renewcommand{\y}{\j*\d}
        \draw [fill=black] (\x,\y) circle (\r);
        \ifthenelse{\i<\LLx}{
            \draw [midarrow={latex}](\x+\r,\y) -- (\x+\r+\a,\y);
        }{}
        \draw [midarrow={latex reversed}] (\x+\r*0.707,\y+\r*0.707) -- (\x+\d/2*0.707,\y+\d/2*0.707);
    }
    
    \ifthenelse{\j<\LLy}{
        \draw [midarrow={latex reversed}](0,\j*\d) -- (\LLx*\d,\j*\d+\d);
    }{}
}
\draw [beginarrow={latex reversed}](\LLx*\d/2,-\d/2) -- (\LLx*\d,0);
\draw [midarrow={latex reversed}](0,\LLy*\d)--(\LLx*\d/2, \LLy*\d + \d/2);

\foreach \i in {0,...,\LLx} {
    \renewcommand{\x}{\i*\d}
    \renewcommand{\y}{0*\d}
    \draw (\x,\y-\a)node{$\vdots$};
    \draw (\x,\y+\a+\LLy*\d+0.2)node{$\vdots$};
}


\draw (\LLx*\d + 2*\d/3, \LLy*\d/2) node {$=$};

\foreach \j in {0,...,\LLy} {
    \renewcommand{\x}{\LLx*\d + 4*\d/3}
    \renewcommand{\y}{\j*\d}
    \draw [fill=blue] (\x,\y) circle (\r);
    \draw [midarrow={latex reversed}] (\x+\r*0.866,\y+\r*0.5) -- (\x+\d/2*0.866,\y+\d/2*0.5);
    \draw [midarrow={latex reversed}] (\x-\r*0.5,\y+\r*0.866) -- (\x-\d/2*0.5,\y+\d/2*0.866);
    \draw [midarrow={latex reversed}] (\x+\r,\y) -- (\x+\d/2,\y);
    \ifthenelse{\j<\LLy}{
        \draw [midarrow={latex reversed}](\x,\y+\r) -- (\x,\y+\r+\a);
    }{
        \draw [midarrow={latex reversed}](\x,\y+\r) -- (\x,\y+\r+\a/2);
    }
}

		
\renewcommand{\x}{\LLx*\d + 4*\d/3}
\renewcommand{\y}{0*\d}
\draw [midarrow={latex reversed}](\x,\y-\r-\a/2) -- (\x,\y-\r); 
\draw (\x,\y-\a) node{$\vdots$}; 
\draw (\x,\y+\d*\LLy+\a+0.2) node{$\vdots$}; 

\draw (\x,-1.2)node{$\ket{\Phi_0}$};


\draw [midarrow={latex}, thick] (\LLx*\d + 5*\d/3,\LLy*\d + \d) to[out=30,in=150] (\LLx*\d + 3*\d - \d/3,\LLy*\d + \d);

\draw (\LLx*\d + 6.5/3*\d,\LLy*\d + 3*\d/2) node{\footnotesize iMM};

\foreach \i in {0,...,1} {
    \foreach \j in {0,...,\LLy} {
        \renewcommand{\x}{\LLx*\d + 3*\d+\d*\i}
        \renewcommand{\y}{\j*\d}
        \pgfmathsetmacro{\color}{ifthenelse(\i<1,"black", "blue")}
        \pgfmathsetmacro{\Vdir}{ifthenelse(\i<1,"latex", "latex reversed")}
        \draw [fill=\color] (\x,\y) circle (\r);
        \draw [midarrow={latex reversed}] (\x-\r*0.5,\y+\r*0.866) -- (\x-\d/2*0.5,\y+\d/2*0.866);
        \ifthenelse{\i<1}{
            \draw [midarrow={latex}] (\x+\r,\y) -- (\x+\d,\y);
        }{
            \draw [midarrow={latex reversed}] (\x+\r,\y) -- (\x+\d/2,\y);
        }
        \ifthenelse{\j<\LLy}{
            \draw [midarrow={\Vdir}](\x,\y+\r) -- (\x,\y+\r+\a);
        }{
            \draw [midarrow={\Vdir}](\x,\y+\r) -- (\x,\y+\r+\a/2);
        }
    }
}


\renewcommand{\x}{\LLx*\d + 3*\d}
\renewcommand{\y}{0*\d}

\draw [endarrow={latex}](\x,\y-\r-\a/2) -- (\x,\y-\r); 
\draw [midarrow={latex reversed}](\x+\d,\y-\r-\a/2) -- (\x+\d,\y-\r); 

\draw (\x,\y-\a) node{$\vdots$}; 
\draw (\x+\d,\y-\a) node{$\vdots$}; 

\draw (\x,\y+\d*\LLy+\a+0.2) node{$\vdots$};
\draw (\x+\d,\y+\d*\LLy+\a+0.2) node{$\vdots$};

\draw (\x,-1.2)node{$A_1$};
\draw (\x+\d,-1.2)node{$\ket{\Phi_1}$};


\draw [midarrow={latex}, thick] (\LLx*\d + 13*\d/3,\LLy*\d + \d) to[out=30,in=150] (\LLx*\d + 17*\d/3 - \d/3,\LLy*\d + \d);

\draw (\LLx*\d + 14.5/3*\d,\LLy*\d + 3*\d/2) node{\footnotesize iMM};

\foreach \i in {0,...,2} {
    \foreach \j in {0,...,\LLy} {
        \renewcommand{\x}{\LLx*\d + 17*\d/3+\d*\i}
        \renewcommand{\y}{\j*\d}
        \pgfmathsetmacro{\color}{ifthenelse(\i<2,"black", "blue")}
        \pgfmathsetmacro{\Vdir}{ifthenelse(\i<2,"latex", "latex")}
        \draw [fill=\color] (\x,\y) circle (\r);
        \draw [midarrow={latex reversed}] (\x-\r*0.5,\y+\r*0.866) -- (\x-\d/2*0.5,\y+\d/2*0.866);
        \ifthenelse{\i<2}{
            \draw [midarrow={latex}] (\x+\r,\y) -- (\x+\d,\y);
        }{}
        \ifthenelse{\j<\LLy}{
            \draw [midarrow={\Vdir}](\x,\y+\r) -- (\x,\y+\r+\a);
        }{
            \draw [endarrow={\Vdir}](\x,\y+\r) -- (\x,\y+\r+\a/2);
        }
    }
}


\renewcommand{\x}{\LLx*\d + 17*\d/3}
\renewcommand{\y}{0*\d}

\draw [endarrow={latex}](\x,\y-\r-\a/2) -- (\x,\y-\r); 
\draw [endarrow={latex}](\x+\d,\y-\r-\a/2) -- (\x+\d,\y-\r); 
\draw [endarrow={latex}](\x+2*\d,\y-\r-\a/2) -- (\x+2*\d,\y-\r);

\draw (\x,\y-\a) node{$\vdots$}; 
\draw (\x+\d,\y-\a) node{$\vdots$}; 
\draw (\x+2*\d,\y-\a) node{$\vdots$}; 

\draw (\x,\y+\d*\LLy+\a+0.2) node{$\vdots$};
\draw (\x+\d,\y+\d*\LLy+\a+0.2) node{$\vdots$};
\draw (\x+2*\d,\y+\d*\LLy+\a+0.2) node{$\vdots$};

\draw (\x,-1.2)node{$A_1$};
\draw (\x+\d,-1.2)node{$A_2$};
\draw (\x+2*\d,-1.2)node{$\ket{\Phi_2}$};

\end{tikzpicture}
    \caption{Peeling procedure to convert iMPS to isoTNS. 
    First the iMPS unit cell is collapsed to form an iMPS with $L_x$ physical legs per site. 
    Then iMM is applied $L_x-1$ times to strip off columns. Example shown is for $L_x=3$.}
    \label{fig:peeling}
\end{figure}
Starting from an iMPS found by iDMRG on an infinite strip of width $L_x$, we collapse the unit cell (a row) of $L_x$ sites to form an iMPS with $L_x$ physical legs per tensor.
We then repeatedly apply iMM to iteratively \textit{peel} isoTNO columns $A$ off of the iMPS, one physical site at a time. 
We use $\ket{\Phi_\ell}$ to denote the orthogonality column of the isoTNS after $\ell$ applications of iMMs:  
\begin{equation}
  \begin{split}
    \ket{\text{iMPS}} &= \ket{\Phi_0} \\
    &\approxeq A_1 \ket{\Phi_1} \approxeq A_1 A_2 \ket{\Phi_2} \\
    &\approxeq A_1A_2...A_{L_x-1} \ket{\Phi_{L_x-1}} =\ket{\text{isoTNS}}
\end{split}
\label{eq:Peeling}
\end{equation}

As an example of this peeling procedure, we peel an $L_x=6$ iMPS with $\chi=128$ representing ground state of TFI model with $g=3.5$. 
The error in fidelity density for a $\chi=6$ isoTNS produced by peeling is $2.4 \times 10^{-4}$, where the mixed transfer matrix $T_{\text{iMPS:isoTNS}}$ is formed by the iMPS and isoTNS.
Increasing to $\chi=10$ decreases the error to $7.3\times10^{-6}$, indicating that this conversion process can be done more accurately by increasing $\chi$.

\subsection{Area Law}
\label{sec:area_law}

For an area-law 2D ground state, one expects that the iMPS has an extensive amount of half-chain entanglement, i.e., on the order of $L_x$, before the peeling procedure.
In order for an isoTNS with finite bond dimension to represent this highly entangled iMPS with finite error density as $L_x \rightarrow \infty$, it is necessary that (1) the error in each iMM application does not increase with the number of iMM applied, and that (2) during the peeling process, the half-chain entanglement entropy has been decreased from $O(L_x)$ in $\ket{\Phi_0}$ to $O(1)$ in $\ket{\Phi_{L_x-1}}$. 
As $O(L_x)$ number of iMMs are applied in the process as in Eq.~\eqref{eq:Peeling}, we expect that the half-chain entanglement entropy of the $\ket{\Phi_\ell}$ is $O(L_x-\ell)$. 
We will see that this is indeed the case.

For an area-law state, we expect the half-chain bipartite entropy of an $L_x$-strip iMPS to be given by $S_{L_x} = \alpha(g) \cdot L_x + O(1)$, where $\alpha(g)$ is coupling dependent slope.
During the peeling process, we calculate the half-chain entropy of the iMPS $\ket{\Phi_{\ell}}$:
\begin{align}
    S(\ket{\Phi_{\ell}}) &= -\sum_i \sigma_i^2 \log \sigma_i^2,
    \label{eq:bipartiteS}
\end{align}
where $\sigma_i$ are the singular values on a vertical bond of $\ket{\Phi_\ell}$. 
This entropy is for the half-chain subsystem composed of the physical indices and the virtual indices of $\ket{\Phi_\ell}$ and is thus not a physical entropy.  
However, it controls the bond dimension needed for the orthogonality column $\ket{\Phi_\ell}$ and is thus a key quantity in the isoTNS representation.

In Fig.~\ref{fig:area_law}, we present the entropy when peeling off columns of an $L_x=8$ iMPS with $\chi=128$ representing ground state of the 2D TFI model at $g=3.50$ found by iDMRG.
We perform the iMM with $\chi=2$, so the first iMM is exact with $\chi_v=1$, and thus $A_1$ is a tensor product of on-site unitaries.
Remarkably, even though $S(\ket{\Phi_\ell})$ is not physical, after an initial delay, the iMM procedure removes essentially $\alpha(g)$ entanglement per iteration, where $\alpha(g)$ is the coupling-dependent physical entropy density in the entropy area-law equation.
We include a more detailed study of the entanglement structure in Appendix~\ref{appendix:area_law} and data from analogous experiments done at the critical coupling in Appendix~\ref{appendix:critical_numerics}. 

\begin{figure}[h]
\centering
\includegraphics[scale=1]{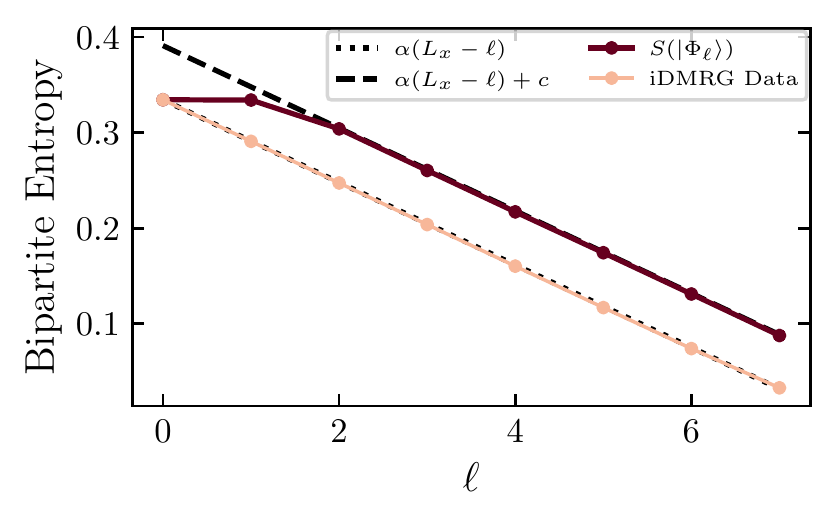}
\caption{
Entropy of $\ket{\Phi_\ell}$ as a function of number of columns $\ell$ removed by peeling the ground state of the 2D TFI Hamiltonian with width $L_x=8$ and $g=3.50$.
After an initial delay, iMM removes an amount of entanglement consistent with the area law.
}
\label{fig:area_law}
\end{figure}

\section{Evaluation of local observables: Energy as an example}
\label{sec:evaluation_energy}

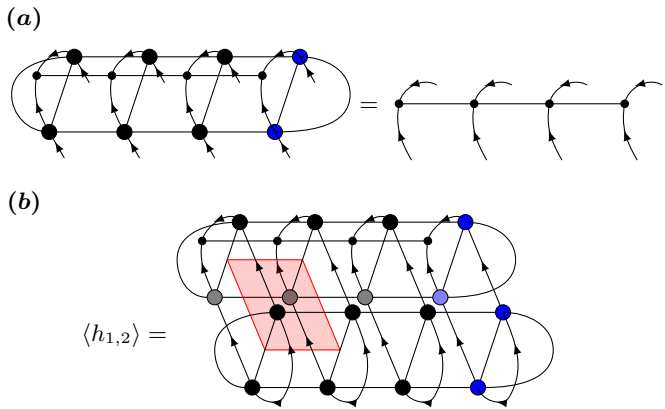
\begin{figure}[b]
    \centering
    \newcommand{\LL}{3}      
\renewcommand{\d}{1.0}   
\renewcommand{\r}{0.1} 
\renewcommand{\a}{0.8}   
\newcommand{\x}{0}
\newcommand{\y}{0}

\begin{tikzpicture}[baseline = (X.base),every node/.style={scale=1.0},scale=1.0]
\draw (0,0) node (X) {};

\draw (-\d/3, 3*\d/2) node{$\bm{(a)}$};

\foreach \i in {0,...,\LL}{
    \renewcommand{\x}{\i*\d}
    \draw (\x, 0) -- (\x+\d/3, \d);
    
    \foreach \j in {0,...,1}{
        \renewcommand{\x}{\i*\d+\j*\d/3}
        \renewcommand{\y}{\j*\d}
        
        \pgfmathsetmacro{\color}{ifthenelse(\i<\LL, "black", "blue")}
        \draw [fill=\color] (\x,\y) circle (\r);
        \ifthenelse{\i<\LL}{
            \draw [] (\x+\r, \y) -- (\x+\r+\a, \y);
        }{}
        
        \draw [midarrow=latex] (\x+\a/2*0.5, \y-\a/2*0.866) -- (\x,\y);
    }
}

\foreach \i in {0,...,\LL}{
    \renewcommand{\x}{\i*\d-\d/6}
    \renewcommand{\y}{3/4*\d}
    
    \draw [fill=black] (\x,\y) circle (\r/2);
    \ifthenelse{\i<\LL}{
        \draw (\x+\r/2, \y) -- (\x+3*\r/2+\a, \y);
    }{}
    
    \draw[midarrow=latex reversed] (\x, \y) to[out=80, in=120] (\x+\d/6+\d/3, \y-3/4*\d+\d);
    \draw[midarrow=latex reversed] (\x, \y) to[out=260, in=120] (\x+\d/6, \y-3/4*\d);
}

\draw [] (\LL*\d+\d, \d/2) to[out=90,in=0] (\LL*\d+\d/3, \d);
\draw [] (\LL*\d+\d, \d/2) to[out=270,in=0] (\LL*\d, 0);

\draw [] (-\d/2, \d/2) to[out=270,in=180] (0, 0);
\draw [] (-\d/2, \d/2) to[out=90,in=180] (\d/3, \d);

\draw (\LL*\d + 10*\d/8, \d/3) node{$=$};

\foreach \i in {0,...,\LL}{
    \renewcommand{\x}{\i*\d-\d/4 + \LL*\d + 1.9*\d}
    \renewcommand{\y}{3/8*\d}
    
    \draw [fill=black] (\x,\y) circle (\r/2);
    \ifthenelse{\i<\LL}{
        \draw (\x+\r/2, \y) -- (\x+3*\r/2+\a, \y);
    }{}
    
    \draw[midarrow=latex reversed] (\x, \y) to[out=80, in=150] (\x+\d/6+\d/3, \y-3/4*\d+\d);
    \draw[midarrow=latex reversed] (\x, \y) to[out=260, in=120] (\x+\d/6, \y-3/4*\d);
}


\pgfmathsetmacro{\ox}{2.2*\d}
\pgfmathsetmacro{\oy}{2.2*\d}

\draw (-\d/3, 5*\d/4-\oy) node{$\bm{(b)}$};

\foreach \i in {0,...,\LL}{
    \renewcommand{\x}{\i*\d-\d/6+\ox}
    \renewcommand{\y}{3/4*\d-\oy}
    
    \draw [fill=black] (\x,\y) circle (\r/2);
    \ifthenelse{\i<\LL}{
        \draw (\x+\r/2, \y) -- (\x+3*\r/2+\a, \y);
    }{}
    
    \draw[midarrow=latex reversed] (\x, \y) to[out=80, in=113] (\x+\d/6+\d/3, \y-3/4*\d+\d);
    \draw[midarrow=latex reversed] (\x, \y) to[out=260, in=113] (\x+\d/6-0.384*\r, \y-3/4*\d+0.923*\r);
}
\foreach \k in {0,...,1}{
    \foreach \i in {0,...,\LL}{
        \foreach \j in {0,...,1}{
            \renewcommand{\x}{\i*\d+\j*\d/3+\k*\d/2+\ox}
            \renewcommand{\y}{\j*\d-\oy-\k*\d*1.2}
            
            \ifthenelse{\j<1}{
                \draw (\x+0.316*\r, \y+0.948*\r) -- (\x+\d/3, \y+\d);
            }{}
            
            \pgfmathsetmacro{\color}{ifthenelse(\i<\LL, "black", "blue")}
            \pgfmathsetmacro{\opacity}{ifthenelse(\j<1, ifthenelse(\k<1 ,0.5, 1), 1)}
            \draw [fill=\color, fill opacity=\opacity] (\x,\y) circle (\r);
            \ifthenelse{\i<\LL}{
                \draw [] (\x+\r, \y) -- (\x+\r+\a, \y);
            }{}
            
            \ifthenelse{\k<1}{
                \draw [midarrow=latex] (\x+\d*0.5, \y-\d*1.2) -- (\x+0.384*\r,\y-0.923*\r);
                
                \ifthenelse{\j<1}{
                    \ifthenelse{\i<1}{
                        \pgfmathsetmacro{\os}{(\i+2.5)/5} 
                        \coordinate (A) at (\x+\d/3*\os,\y+\d*\os);
                        \coordinate (B) at (\x+\d/3*\os+\d,\y+\d*\os);
                        \coordinate (C) at (\x+\d/3*\os+\d/2,\y+\d*\os-\d*1.2);
                        \coordinate (D) at (\x+\d/3*\os+\d/2+\d,\y+\d*\os-\d*1.2);
                        \draw [draw=red, fill=red, fill opacity=0.2] 
                            (A) -- (B) -- (D) -- (C) -- cycle;
                   }{}
                }{}
            }{
                \ifthenelse{\j<1}{
                    \draw [beginarrow=latex] (\x+5*\d/12, \y-\d/8) to[out=240,in=280] (\x, \y);
                }{
                    \draw [midarrow=latex] (\x+\d/12, \y-9*\d/8) to[out=60,in=293] (\x, \y);
                }
            }
        }
    }
    
    \draw [] (\LL*\d+\d+\k*\d/2+\ox, \d/2-\k*\d*1.2-\oy) to[out=90,in=0] (\LL*\d+\d/3+\k*\d/2+\ox, \d-\k*\d*1.2-\oy);
    \draw [] (\LL*\d+\d+\k*\d/2+\ox, \d/2-\k*\d*1.2-\oy) to[out=270,in=0] (\LL*\d+\k*\d/2+\ox+\r, -\k*\d*1.2-\oy);
    
    \draw [] (-\d/2+\k*\d/2+\ox, \d/2-\k*\d*1.2-\oy) to[out=270,in=180] (\k*\d/2+\ox-\r, -\k*\d*1.2-\oy);
    \draw [] (-\d/2+\k*\d/2+\ox, \d/2-\k*\d*1.2-\oy) to[out=90,in=180] (\k*\d/2+\d/3+\ox, \d-\k*\d*1.2-\oy);
}


\draw (\ox-1.2*\d,-1*\oy-0.5*\d) node{$\langle h_{1,2} \rangle = $};

\end{tikzpicture}
    \caption{bMPO energy evaluation.
    \textbf{(a)} boundary MPO representing fixed point of isoTNS transfer matrix. 
    $L_x \times \infty$ strip is oriented so that infinite direction is vertical.
    Horizontal arrows are omitted for clarity.
    Fixed point only needs to be found against the direction of vertical arrows, as fixed point in the direction of arrows is the identity. 
    \textbf{(b)} Evaluation of $\langle h_{1,2}$ using bMPO fixed point and trivial fixed point. The energy of a row is found by $\sum_c \langle h_{c,c+1} \rangle$.}
    \label{fig:bMPO_energy}
\end{figure}

The evaluation of expectation values of local observables is crucial and non-trivial for 2D finite and infinite TNS.
As the second application of the iMM algorithm, we show that, up to some small error density, we can efficiently evaluate the expectation of local observables of the given strip isoTNS using iMM.
We benchmark the result by comparing to results obtained by approximate PEPS contractions using boundary MPO (bMPO) approach~\cite{lubasch2014unifying,ran2020tensor} and the exact contraction by collapsing the state into iMPS.
Here, we consider the evaluation of energy as an example, but the method can be applied to the evaluation of other local observables.

\subsection{Energy evaluated by iMM}
\label{sec:iMM_energy}

Suppose we wish to evaluate the energy of a state given a Hamiltonian $H = \sum_c H_{c,c+1}$ where $H_{c,c+1}$ acting on columns $c$ and $c+1$ is composed of a translationally invariant local operator $h_{c,c+1}$. 
The expectation value of $H_{c,c+1}$ can be found by evaluating the energy of each local operator $h_{c,c+1}$ with an isoTNS with doubled column $c,c+1$ as the orthogonality column.
This is because calculating expectation values of operators contained entirely in the orthogonality column reduces to an efficient 1D iMPS problem.
If we could move the orthogonality column freely around the strip in an exact fashion, this would give us an exact method for evaluating the energy.
However, as moving the orthogonality column requires the iMM, this method inherently incurs an approximation error on the order of the iMM errors.
As demonstrated in Sec.~\ref{sec:repeated_iMM}, the state represented by an isoTNS is not significantly affected by iMM applications, where the individual errors of each iMM application were presented in Table~\ref{tab:single}.
Thus we can perform a full sweep of iMM iterations over the entire strip and use the two-site orthogonality center of each doubled column to evaluate the expectation value of $h_{c,c+1}$.
Doing this for each two-column term gives us the energy of one row, $\langle E_{\mathrm{row}}\rangle = \sum_{c} \langle h_{c,c+1} \rangle$, which can be converted to a per-site energy by dividing by strip width.

Note that the maximum bond dimension $\chi' = \chi'_h = \chi'_v$, and $\eta'$ of the iMM used to calculate observables does not need to be the same as $\chi$ and $\eta$, that of the original isoTNS.
The accuracy of calculated observables increases as we increase $\chi'$, as this improves the accuracy of iMM.
This method for calculating observables scales linearly with the strip width and requires $L_x-1$ applications of iMM, which as stated previously has a computational cost of $\mathcal{O}(\chi'^4\eta'^3)$.

\subsection{Energy evaluated by boundary MPO}
To check the energy evaluation from the iMM independently, we compute the energy of an isoTNS strip without using the iMM.
To do this, we consider the isoTNS row transfer matrix and find its fixed point. 
We can compute local observables efficiently given the fixed points of the row transfer matrix.
However, as the dimension of the transfer matrix grows exponentially in the width $L_x$, we approximate the fixed point ``vector'' by the boundary MPO (bMPO)~\cite{verstraete2004renormalization} as depicted in Fig.~\ref{fig:bMPO_energy}(a).
Note that the strip is oriented such that the infinite direction is vertical.
Combining the bMPO contraction methods developed for finite TNS~\cite{verstraete2004renormalization,lubasch2014unifying,lubasch2014algorithms} and the power method, we find the dominant eigenvector of the transfer matrix represented by a bMPO with bond dimension $D_{\mathrm{bMPO}}$.

We only have to converge the fixed point \textit{against} the direction of the vertical isometric arrows (from the top down in Fig.~\ref{fig:bMPO_energy}(b)).
This is because the fixed point along the direction of the arrows is, by definition, an identity operator over each column.
The non-trivial fixed point \textit{against} the isometric arrow direction admits a spectral decomposition $U \rho U^\dagger$, where $\rho$ is diagonal and positive definite, encoding the square of the Schmidt values. However, we do not utilize this property but use a bMPO directly to parametrize the fixed point vector.
With this bMPO representing the fixed point of the isoTNS row transfer matrix, we calculate the energy of a row by evaluating $\langle\sum_{c} h_{c,c+1} \rangle = \sum_c \langle h_{c,c+1} \rangle$.
This is done by contracting the network shown in Fig.~\ref{fig:bMPO_energy}(b) for each two-column local operator $h_{c,c+1}$, shown here to act on two neighboring rows.

This method can be made arbitrarily accurate by increasing the bond dimension $D_{\mathrm{bMPO}}$, but we note that this method is very costly as strip width $L_x$ and the bond dimension $\chi$ of the isoTNS grows.
Calculating the fixed point scales as $\mathcal{O}\left( N_\text{iter}L_x (\chi^4 D_{\mathrm{bMPO}}^3 + d\chi^6 D_{\mathrm{bMPO}}^2) \right)$, where $d$ is the local Hilbert space dimension and $\chi$ is the dimension of all virtual legs in the isoTNS.
Typically $D_{\mathrm{bMPO}} \sim \chi^2$, so this method scales as $\mathcal{O}(\chi^{10})$.
%
The $N_\text{iter}$ is the number of the transfer matrix-vector multiplications required for convergence which is related to the gap in the transfer matrix.

\subsection{Energy Benchmarks}

The exact but most computationally intensive $\mathcal{O}(\exp(L_x))$  method is to find the exact fixed point and perform an exact contraction.
To this end, we collapse each row of an isoTNS to form an iMPS with physical dimension $d^{L_x}$, representing $L_x$ physical sites per tensor. 
Then, standard MPS methods can be used to find the exact fixed point and evaluate the energy.
We use this essentially exact method only to benchmark the previous two methods.

The benchmark result is shown in Fig.~\ref{fig:energy_estimate}.
We apply both methods to a $L_x=8$, $\chi=4$ isoTNS produced by peeling of a 2D TFI GS iMPS, as described in Sec.~\ref{sec:iDMRG_to_isoTNS}.
We find that the bMPO method is essentially exact for large enough $D_{\mathrm{bMPO}}$, while the accuracy of the iMM energy increases with both $\chi'$ and $\eta'$, as the splitting can be done more accurately with the larger space of available tensors.
From this, we note that the iMM energy tends to underestimate the true energy; yet for large strip widths and large $\chi$, other methods are infeasible due to computational costs.

\begin{figure}[htb]
    \centering
    \includegraphics[width=\columnwidth]{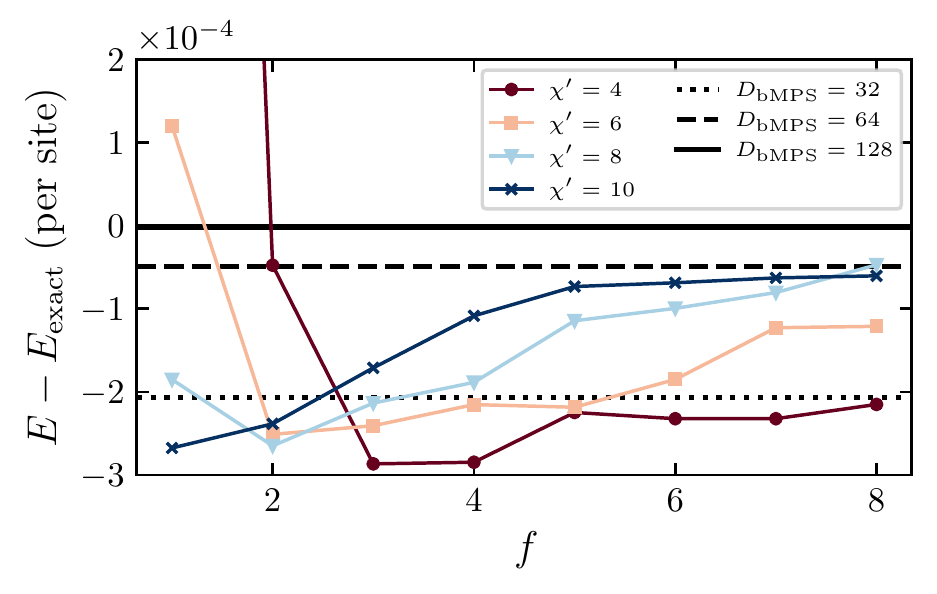}
    \caption{
    Energy of the $L_x=8$, $\chi=4$ isoTNS produced by peeling evaluated with the iMM and bMPO methods.
    We compare the energies to the exact energy evaluated by exact iMPS contraction.
    The energy from the bMPO method is close to exact for large $D_{\mathrm{bMPO}}$ but is very costly as $D_{\mathrm{bMPO}}$ grows due to complexity of finding the fixed point.
    The accuracy of iMM energy increases with both $\eta'$ and $\chi'$ used in iMM. 
    In the figure, we use $f \equiv \eta'/\chi'$.}
    \label{fig:energy_estimate}
\end{figure}

\section{iTEBD\textsuperscript{2}}
\label{sec:iTEBD2}

We now introduce the TEBD-based time evolution algorithm for infinite strip isoTNS, dubbed iTEBD\textsuperscript{2}.
We then demonstrate the algorithm by performing imaginary time evolution to find the ground states of the two dimensional transverse field Ising model.

\begin{figure*}[htb]
    \centering
    \input{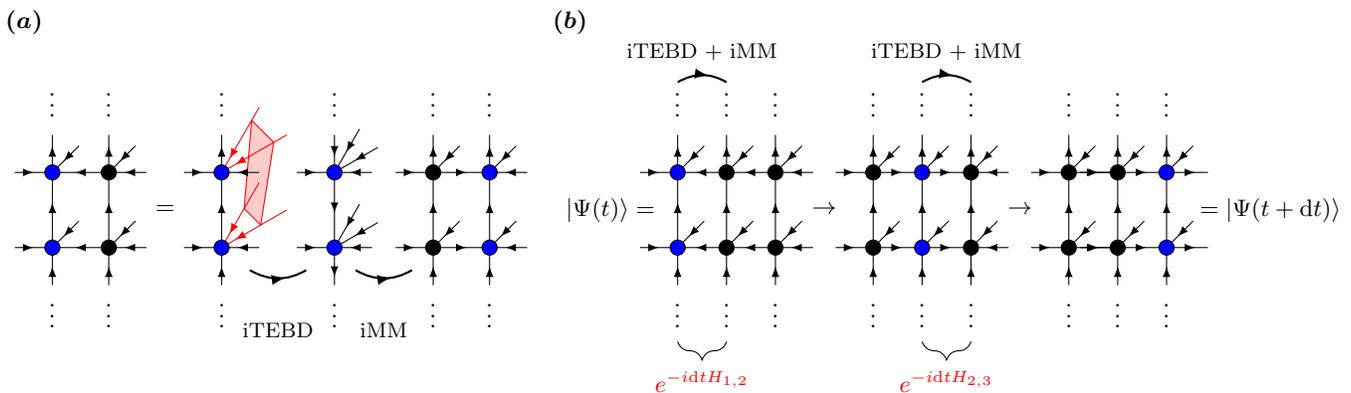}
    \caption{iTEBD\textsuperscript{2} Algorithm. 
    \textbf{(a)} Subroutine acting on two columns, the left of which is the orthogonality column. 
    First the two columns are fused. 1D iTEBD is applied to the physical legs (colored red) of the double column using gates $e^{-i\mathrm{d}t h_{c,c+1}}$, acting on a plaquette of four physical sites; note that the gate does not act on any virtual legs.
    iMM is then applied to split the doubled column and move the orthogonality center to the right. 
    \textbf{(b)} The iTEBD\textsuperscript{2} algorithm involves $L_x-1$ iterations of the subroutines described in (a) to implement a total Trotterized time evolution step of $\mathrm{d}t$. 
    Subsequent iTEBD\textsuperscript{2} sweeps alternate sweep directions, left-to-right and vice-versa.}
    \label{fig:iTEBD2}
\end{figure*}

\subsection{iTEBD\textsuperscript{2} Algorithm}
\label{sec:iTEBD2_algorithm}

TEBD-like algorithms perform time evolution by approximating it as successive local time evolutions via the Suzuki-Trotter decomposition of a Hamiltonian, i.e., the sum of local terms:
\begin{align}
    \ket{\Psi(t+\mathrm{d}t)} &= e^{-iH\mathrm{d}t }\ket{\Psi(t)}\\
    &\approx \prod_j e^{-ih_j \mathrm{d}t}\ket{\Psi(t)}.
\end{align}
After each local time evolution operation $e^{-ih_j\mathrm{d}t}$, we find the closest state within the ansatz manifold to represent the time-evolved state.
This general approach leads to tensor network implementations as the TEBD algorithm for finite MPS~\cite{vidal2003efficient,vidal2004efficient}, the iTEBD algorithm for iMPS~\cite{vidal2007classical,orus2008iTEBD}, the TEBD\textsuperscript{2} algorithm for 2D isoTNS~\cite{zaletel2020isometric}, and the simple and full update for TNS time evolution~\cite{jiang2008simple,jordan2008full,phien2015infinite}. 
Here we focus on the application of iTEBD\textsuperscript{2} to strip isoTNS.

For an width $L_x$ strip, we write the Hamiltonian as
\begin{align}
    H&=\sum_{c=1}^{L_x-1} H_{c,c+1},
\end{align}
where each $H_{c,c+1}$ is the infinite collection of local operators acting on columns $c$ and $c+1$.
Each infinite two-column operator is the sum of local terms, and we will assume the form
\begin{align}
    H_{c,c+1}&=\sum_{i} h^{i,i+1}_{c,c+1},
\end{align}
which acts on a plaquette of four spins on columns $c$ and $c+1$ and rows $i$ and $i+1$.
We work with models where the local term $h^{i,i+1}_{c,c+1}$ is translationally invariant in the vertical direction, so we will drop the row superscripts.

To perform time evolution, we Trotterize the full Hamiltonian according to
\begin{align}
    e^{-i\mathrm{d}t H} =  e^{-iH_{1,2}\mathrm{d}t} e^{-iH_{2,3}\mathrm{d}t} \ldots e^{-iH_{L_x-1,L_x} \mathrm{d}t} 
\end{align}
which is a first-order splitting of the column operators $H_{c,c+1}$. 
We then perform a first-order splitting of the individual plaquette terms $h_{c,c+1}$ within the column operators. 
The iTEBD\textsuperscript{2} algorithm on an infinite strip using a Hamiltonian of this form is depicted graphically in Fig.~\ref{fig:iTEBD2}.
Our initial isoTNS has the orthogonality column with arrows pointing up as the left-most column.
As shown in Fig.~\ref{fig:iTEBD2}(a), we first merge columns $1$ and $2$ to form a doubled column with two physical sites per tensor.
We then apply the time evolution operator $e^{-i\mathrm{d}t H_{1,2}/2}\approx\prod e^{-i\mathrm{d}t h_{1,2}/2}$ to only the physical legs of the doubled column, which flips the isometric arrows to point down.
We note that we do not use the standard iTEBD algorithm that enforces at least a two-site unit cell along the column~\cite{vidal2007classical}, as we do not wish to have a non-trivial unit cell. 
Instead we simply apply the two-site gates and do SVD truncation with gauge-fixing to sweep downward until convergence.
Following this 1D iTEBD on a doubled column, we apply the chosen iMM algorithm to split this doubled column into an isoTNO $A_1$ and a new orthogonality column $\ket{\Phi_2}$, both of which have isometry arrows pointing up and a single physical leg. 
We can now repeat this process with columns $2$ and $3$.
Proceeding in this way from left (right) to right (left) on odd (even) iterations of iTEBD\textsuperscript{2}, we perform one time evolution step $dt$ using $L_x-1$ applications of the iMM and 1D iTEBD; this procedure is summarized in Fig.~\ref{fig:iTEBD2}(b).

Both of these subroutines utilize deterministic SVDs and have complexity $\mathcal{O}(\chi^4\eta^3)$~\footnote{
After merging two columns, the merged column would have bond dimension $\eta\chi$.
Optionally, a compression can be performed after merging two columns and before the iTEBD. 
In this work, we compress the bond dimension from $\eta\chi$ down to $\eta$ before performing the iTEBD. 
If a compression to bond dimension $\eta$ is performed before the iTEBD, the complexity is $\mathcal{O}(\chi^4 \eta^3)$; otherwise, the complexity is slightly higher, $\mathcal{O}(\chi^5\eta^3)$. 
In practice, we observe a slight deterioration in the result with compression. But we can obtain better results overall by using larger bond dimensions within the same run-time. 
}, where we assume that $\chi_h=\chi_v=\chi$, and allow the bond dimension $\eta$ along the orthogonality column to be different from $\chi$. 
Hence we see that increasing the strip width incurs a linear increase in the cost of algorithms for systems obeying area law, compared to an exponential increase in costs for 1D algorithms applied to 2D strips.

\subsection{Ground state search}
\begin{figure*}[htb]
    \centering
    \includegraphics[scale=1.0,width=\textwidth]{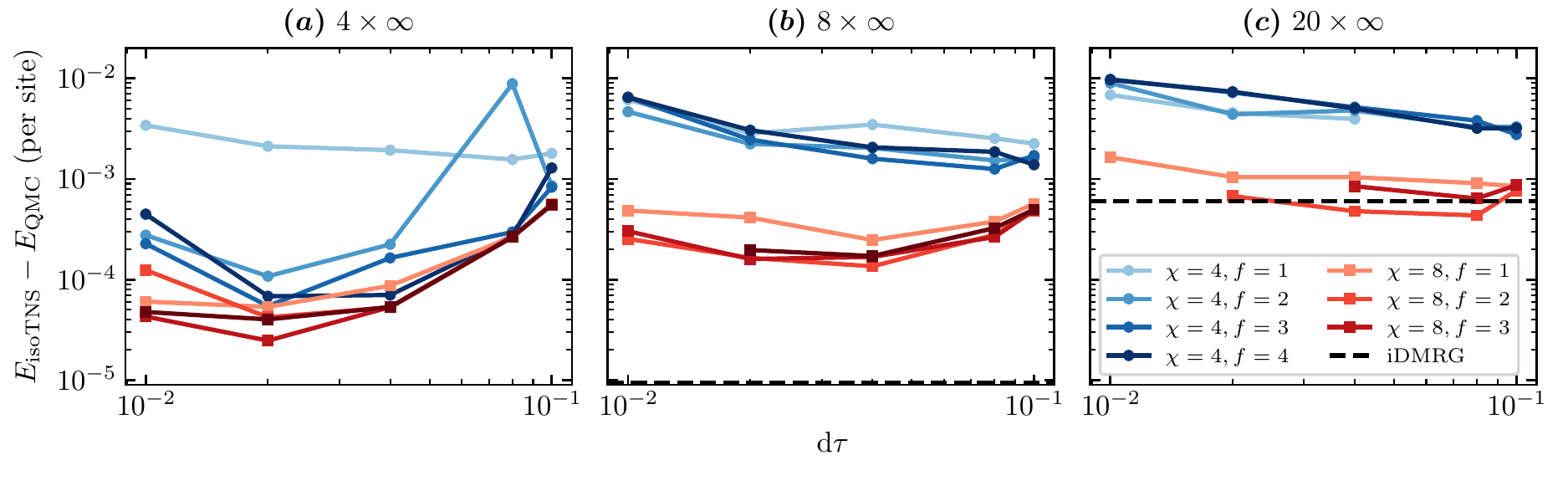}
    \caption{
    Ground state energies achieved with iTEBD$^2$ for paramagnetic 2D TFI with $g=3.50$ on an $L_x=4, 8,$ and $20$ strip.
    We compare the isoTNS energies against essentially exact energies from quantum Monte Carlo (QMC) extrapolated from strips of finite length.
    As a comparison, the dashed line is the result of an iDMRG calculation with bond dimension $\chi=512$. For $L_x=4$, the iDMRG result is below the bottom axis of the plot.
    }
    \label{fig:iTEBD_energy}
\end{figure*}
With the iTEBD\textsuperscript{2} algorithm, we can perform imaginary time evolution to find the ground state of Hamiltonian starting from an initial state $\ket{\Psi_0}$:
\begin{align}
    \ket{\Psi_\mathrm{GS}} &= \lim_{\tau \rightarrow \infty} \frac{e^{-\tau H} \ket{\Psi_0}}{\norm{e^{-\tau H} \ket{\Psi_0}}}.
\end{align}

We benchmark the algorithm with 2D TFI model on strips of width $L_x=4, 8, 20$ using isoTNS of bond dimensions $\chi=4, 8$ with $D_b=\chi$.
We investigate a range of $f \equiv \eta/\chi$ values, where again $f$ controls the bond dimension of the intermediate $\ket{\Phi}$ columns produced by iMM during the iTEBD\textsuperscript{2} sweeps.
As an essentially exact benchmark, we compute the energy via quantum Monte Carlo (QMC) with the ALPS library~\cite{bauer2011alps} on strips of both finite width and length; we find the energy of strips of increasing length and extrapolate to infinite strips.
The results are presented in Fig.~\ref{fig:iTEBD_energy}. Here, the energy of each isoTNS is calculated by the iMM method discussed in Sec. \ref{sec:iMM_energy}, using $\chi=10$ and $f=4$ to give an accurate energy estimate.

If iMM were exact, then the energy would decrease monotonically with time step $\mathrm{d}\tau$. 
Yet we clearly see that there exists an energy minimum at intermediate $\mathrm{d}\tau$.
There are competing effects between less Trotter error from a smaller time step but then more error accumulated from an increased number of iTEBD and iMM iterations~\cite{zaletel2020isometric,lin2022efficient}.
Additionally, the largest $f$ does not provide the lowest energy, as one would naively suspect.
Larger $f$ leads to increased vertical bond dimensions during the sweep, leading to larger truncation errors during the iTEBD on the doubled column.

\section{Conclusions}
\label{sec:conclusion}

In this work, we have extended isometric tensor networks on a square lattice to infinite strip geometries and introduced algorithms to both manipulate and time-evolve the ansatz.
We introduced four different infinite Moses Move algorithms and found the combination of \iMML\ and \iMMP\ to be efficient and stable.
The isoTNOs produced by iMM have a significant disentangling effect on a two-sided MPS.
We verify this effect by showing that iMM can remove an amount of entanglement consistent with the area law of the underlying phase at each iteration, similar to MM~\cite{zaletel2020isometric}.
We demonstrated three different applications based on iMM algorithms: (i) transforming an iMPS into an 2D isoTNS, (ii) the evaluation of local observables, and (iii) the iTEBD\textsuperscript{2} algorithm, which enables ground state optimization via imaginary time evolution. 
These results demonstrate that the isoTNS is a 2D network ansatz permitting both efficient optimization and calculations, and we expect the benefits of this method to become apparent as strip width increases beyond the reach of 1D methods.

We conclude by commenting interesting applications of the our work. 
Of foremost interest is the extension of our methods to the simulation of strongly correlated many-body systems that are infinite in both directions.
Having now dealt with infinite columns through the methods we've introduced, a remaining challenge is to find a fixed point solution to the splitting problem $A \Psi = \Psi B$, where $A$ ($B$) is an infinite isometric column with horizontal arrows pointing right (left) and $\Psi$ is an infinite orthogonality column.
While this is solved in 1D by iterative gauge-fixed QR decompositions, the iMM algorithms as currently formulated do not fix the gauge on the horizontal legs, so there is no a priori reason that repeatedly applying the algorithm will converge.
Additionally, many interesting physical systems display spontaneous translational symmetry breaking.
While in our current prescription, the tensors in a row can differ from one another, each row is repeated along the vertical direction.
To allow for a non-trivial unit cell in the infinite direction, we must generalize the iMM methods discussed in Sec.~\ref{sec:iMM} to multi-site unit cells.
Such modifications are simple extensions of the iMM algorithms introduced earlier.
We leave the generalization of iMM to infinite width and exploration of non-trivial unit cells as future works.

A second application is motivated by the use of tensor networks as state preparing circuits on quantum computers~\cite{schon2005sequential,banuls2008sequentially,foss-feig2021holographic}.
This relies on the network having an isometric structure so that tensors can be interpreted as unitaries acting on qubits and that there is a unidirectional flow of time opposite to the isometry directions.
In general, the isometric tensor of bond dimension $\chi$ would translate into a gate acting across $\log{\chi}$ qubits.
To really construct the circuits that could run on quantum computers, we have to further decompose such ``dense'' unitaries into quantum circuits consisting of two-site gate~\cite{haghshenas2022variational}, resulting in the so-called quantum circuit tensor network.
Quantum circuits of finite 2D isoTNS~\cite{slattery2021quantum,wei2022sequential} and infinite 1D isoTNS~\cite{Barratt2021,Dborin2022,astrakhantsev2022time} have been numerically and analytically explored for this purpose.
Adapting changes mentioned above, the infinite strip networks developed here can be prepared on a quantum computer, allowing for calculation of expectation values by directly measuring the state without expensive and approximate boundary contraction methods.
Additionally, the finite MM algorithm was recently used to prepare isometric circuits encoding entanglement renormalization principles to accurately measure long-range correlations in critical quantum chains~\cite{anand2022holographic}.
The iMM algorithm developed here can be used to extend this work to the thermodynamic limit.

\begin{acknowledgments}
    YW was supported by the RIKEN Interdisciplinary Theoretical and Mathematical Sciences Program.
	SA and MZ were supported by the U.S. Department of Energy, Office of Science, Basic Energy Sciences, under Early Career Award No. DE-SC0022716. 
    F.P. acknowledges the support of the Deutsche Forschungsgemeinschaft (DFG, German Research Foundation) under Germany's Excellence Strategy EXC-2111-390814868. 
	S.L. and F.P. were supported by the DFG TRR80 and the Bavarian state government with funds from the Hightech Agenda Bayern Plus.
    This work was supported by the European Research Council (ERC) under the European Union's Horizon 2020 research and innovation program (Grant Agreement No. 771537). %
	The research is part of the Munich Quantum Valley, which is supported by the Bavarian state government with funds from the Hightech Agenda Bayern Plus.  
	Computing resources were provided by National Energy Research Scientific Computing Center (NERSC), a U.S. Department of Energy Office of Science User Facility located at Lawrence Berkeley National Laboratory, operated under Contract No. DE-AC02-05CH11231 using NERSC Award No. BES-ERCAP0020043.

\textbf{Data and materials availability} – Data analysis and simulation codes are available upon reasonable request.
\end{acknowledgments}

\appendix
\section{Proof of Eq.~\eqref{eq:error_decomposition}}
\label{sec:error_decomp}

Recall $A$ is an isometry such that $A^\dagger A = \mathbbm{1}$ and $A A^\dagger = \mathcal{P}$.
As a result, we have the following identity:
\begin{align}
    \label{eq:error_identity}
    \norm{\ket{\Psi} - A\ket{\Phi}}^2 =& \norm{\ket{\Psi} - AA^\dag \ket{\Psi}}^2 \nonumber \\
    &+ \norm{A^\dag \ket{\Psi} - \ket{\Phi}}^2.
\end{align}
This identity introduces the unnormalized intermediate states $A^\dag \ket{\Psi}$ and invites an interpretation of the MM error as the sum of a projection error and an MPS truncation error, which are respectively the first term and the second term in Eq.~\eqref{eq:error_identity}.

As we prove below, for a uniform finite system of size $L$, to the first order in the errors, both errors are proportional to $L$:  
\begin{align}
    \label{eq:errors2}
    \begin{split}
        \norm{\ket{\Psi} - AA^\dag \ket{\Psi}}^2 
        \approx [1 - \lambda_1(T_{A^\dag\Psi:A^\dag\Psi})]L \\
        \norm{A^\dag \ket{\Psi} - \ket{\Phi}}^2 
        \approx [1 - \left( \lambda_1(T_{\widetilde{A^\dag\Psi}:\Phi}) \right)^2]L,
    \end{split}
\end{align}
where $\widetilde{A^\dag \ket{\Psi}} \propto A^\dag \ket{\Psi}$ is the normalized state \footnote{Note that it is wrong to use $T_{A^\dag \Psi:\Lambda}$ in Eq.~\eqref{eq:errors2} as the states making up the transfer matrix are not properly normalized}.
Recall the definitions as in Eq.~\eqref{eq:error_defs}:
\begin{align*}
    \epsilon_p &\equiv 1 - \lambda_1(T_{A^\dag\Psi:A^\dag\Psi}), \nonumber\\
    \epsilon_t &\equiv 1 - \left( \lambda_1(T_{\widetilde{A^\dag\Psi}:\Phi}) \right)^2.
\end{align*}
Evidently, the total error density in Eq.~\eqref{eq:total_error} can be decomposed as
\begin{align}
  \epsilon &= \epsilon_p + \epsilon_t + O(\epsilon_p^2, \epsilon_t^2, \epsilon_p \epsilon_t),
\end{align}
which is Eq.~\eqref{eq:error_decomposition}.

To complete this appendix, we now give the derivation of Eq.~\eqref{eq:errors2}.
The first term in Eq.~\eqref{eq:error_identity} is due to the projection $AA^\dag$: 
\begin{align}
  \norm{\ket{\Psi}-AA^\dag \ket{\Psi}}^2 &= 1 - \norm{A^\dag \ket{\Psi}}^2 \nonumber\\
  &\equiv 1-(1-\epsilon_p)^L \approx \epsilon_p L.
\end{align}
The second source of $\epsilon$ comes from the truncation error of representing the unnormalized state $A^\dag \ket{\Psi}$ with the MPS $\ket{\Phi}$: 
\begin{equation}
  \begin{split}
  &\norm{A^\dag \ket{\Psi} - \ket{\Phi}}^2 = 1 +  \norm{A^\dag \ket{\Psi}}^2 - 2 \Re \braket{\Phi|A^\dag|\Psi}
  \\
  &\hspace{5mm}= 2 - \epsilon_p L  - 2 \Re \bra{\Phi} \frac{A^\dag\ket{\Psi}}{\left(\sqrt{1-\epsilon_p}\right)^L} \left(\sqrt{1-\epsilon_p}\right)^L
\\
&\hspace{5mm}\equiv 2 - \epsilon_p L  - 2 \Re \bra{\Phi} \widetilde{A^\dag\ket{\Psi}} (\sqrt{1-\epsilon_p})^L,
\end{split}
\end{equation}
where we defined $\widetilde{A^\dag \ket{\Psi}} = A^\dag \ket{\Psi}/{\left(\sqrt{1-\epsilon_p}\right)^L}$ as the normalized state.  
Analogously to Eq.~\eqref{eq:epsilon},  
\begin{equation}
   \norm{\ket{\Phi} - \widetilde{A^\dag \ket{\Psi}}}^2 = 2- 2 \Re \bra{\Phi} \widetilde{A^\dag\ket{\Psi}} \approx \epsilon_t L.
\end{equation} 
Thus, putting everything together, we have
\begin{align}
  \norm{A^\dag \ket{\Psi} - \ket{\Phi}}^2 &\approx 2-\epsilon_p L - (2-\epsilon_t L)\left(1-\frac{\epsilon_p L}{2}\right)\nonumber \\
  &= \epsilon_t L.
\end{align}

\section{Isometric filling of $A$}
\label{sec:isofill}

Let $a$ denote the tensor making up $A$. 
Label its indices as below:
\begin{equation}
    \label{eq:bond_dims}
    \begin{tikzpicture}[baseline = (X.base),every node/.style={scale=1.0},scale=1.0]
\renewcommand{\d}{1.0}   
\renewcommand{\r}{0.25}  
\renewcommand{\a}{0.5}   
\newcommand{\dH}{1.0}   
\newcommand{\aH}{0.5}
\newcommand{\x}{0}
\newcommand{\y}{0}
\draw (0, 0) node (X) {};
\draw (\x,\y) circle (\r);
\draw (\x,\y) node {$a$};
\draw (\x - \r - 0.75,\y) node {$i_0$};
\draw (\x + \r + 0.75,\y) node {$i_1$.};
\draw (\x, \y - \r - 0.75) node {$i_2$};
\draw (\x, \y + \r + 0.75) node {$i_3$};
\draw [midarrow={latex}](\x+\r,\y) -- (\x+\r+\aH,\y); 
\draw [midarrow={latex}](\x-\r-\aH,\y) -- (\x-\r,\y); 
\draw [midarrow={latex}](\x,\y+\r) -- (\x,\y+\r+\a);
\draw [midarrow={latex reversed}](\x,\y-\r) -- (\x,\y-\r-\a);
\end{tikzpicture}.
\end{equation}
To enlarge the bond dimension $\chi_v = \dim(i_2) = \dim(i_3)$ to $\chi'_v$ while keeping the operator that $A$ represents invariant and the isometric condition intact, one groups $i_2i_0$ and $i_3i_1$ respectively as the row and column index of the isometric matrix $a$. 
$i_2$ and $i_3$ are the ``slow" index of their respective combined indices. 
Here we assume $\dim(i_0) = \dim(i_1)$, and thus $a$ is square. 
To enlarge the bond dimensions, one first zero-pads on the index $i_2$ and then adds orthogonal columns on index $i_3$:
\begin{equation}
 [a] \rightarrow \begin{bmatrix} a \\ 0 \end{bmatrix} \rightarrow \begin{bmatrix} a & 0 \\ 0 & a^\bot \end{bmatrix},
\end{equation}
where $a^\bot$ is an arbitrary unitary matrix with $(\chi'_v - \chi_v)\dim(i_0)$ number of rows and columns. 
Thus, when $\dim(i_0) = \dim(i_1)$, the result of isometric filling is $a \rightarrow a' = a \oplus a^\bot$.

\section{Area law in Pealing iMPS}
\label{appendix:area_law}

As discussed in Sec.~\ref{sec:disentangling_iMM}, iMM has a disentangling effect on the new iMPS $\ket{\Phi_{\ell}}$ and we have shown by computing the half-chain entropy of the iMPS $\ket{\Phi_{\ell}}$ over the peeling process:
\begin{align*}
    S(\ket{\Phi_{\ell}}) &= -\sum_i \sigma_i^2 \log \sigma_i^2,
\end{align*}
where $\sigma_i$ are the singular values on a vertical bond of $\ket{\Phi_\ell}$; see Fig.~\ref{fig:area_law_appendix}(a).
This entropy is for the half-chain subsystem composed of the physical indices and the virtual indices of $\ket{\Phi_\ell}$ and is thus not physical.

To dig deeper, we ask whether iMM can  extract $\alpha(g)$ using only physical quantities~\footnote{Here by physical quantities, we mean any quantity obtainable in principle from the initial wavefunction $\ket{\Phi_0}$}. 
The answer is yes.
For each $\ket{\Phi_\ell}$, a further iMM can be applied to $\ket{\Phi_\ell}$ with $\chi_h = 1$:  
\begin{equation}
  \ket{\Phi_\ell} \approx A'_\ell \ket{\Phi_\ell'},
\end{equation}
where the isoTNO $A_\ell'$ has no physical legs. 
The original $\ket{\Phi_l}$ has virtual bonds and can be viewed as the purification state of the density matrix over the remaining physical spins. 
Now, the $\ket{\Phi_\ell'}$ is a pure state and the dominant eigenstate of such a density matrix.
Remarkably, we find that $\ket{\Phi_\ell'}$ has an error density on the order of $10^{-5}$ when compared to the iMPS representing width $L_x-\ell$ TFI ground state and thus has an $S(\ket{\Phi'_\ell})$ that obeys the entropy area law almost perfectly while also agreeing with the values found via iDMRG (see blue crosses Fig.~\ref{fig:area_law_appendix}(b)). 
This suggests that the TFI Hamiltonian of width $L$ and the entanglement Hamiltonian of a subsystem of width $L$ in a larger TFI system have approximately the same ground state.
We repeat this experiment at the critical transverse field $g=g_C^{2D}$ and find the same behavior. Results are shown in Appendix~\ref{appendix:critical_numerics}.

\begin{figure}[h]
\centering
\input{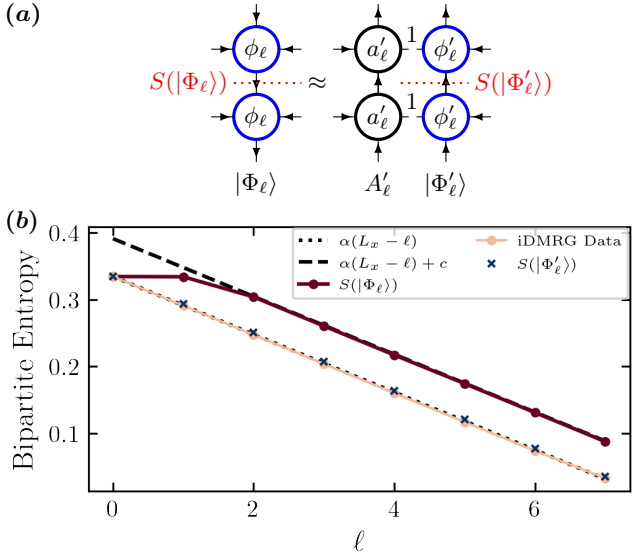}
\caption{
Area law properties of peeled isoTNS.
\textbf{(a)} $S(\ket{\Phi_\ell})$ is the half-chain entropy of $\ket{\Phi_\ell}$, including contributions from left virtual legs.
$S(\ket{\Phi_\ell'})$ is a physical entropy between physical legs above and below the bipartition, found by projecting out the left virtual legs via iMM. 
\textbf{(b)} Entropy of $\ket{\Phi_\ell}$ as a function of number of columns $\ell$ removed by peeling the ground state of the 2D TFI Hamiltonian with width $L_x=8$ and $g=3.50$.
After an initial delay, iMM removes an amount of entanglement consistent with the area law.
Physical entropy $S(\ket{\Phi_\ell'})$ and entropy from iDMRG GS of different width strips agree and obey area law.
}
\label{fig:area_law_appendix}
\end{figure}

\section{Numerical results for $g=3.04438$}
\label{appendix:critical_numerics}

Here we present numerics at $g=g_C^{2D}=3.04438$, the critical transverse field for the two-dimensional TFI in the thermodynamic limit. 
Close to the critical point, we expect this model on finite width strips to be more difficult to capture by an iMPS due to increased entanglement.

First we repeat the area law experiment of Sec.~\ref{sec:area_law} and show the results in Fig.~\ref{fig:area_law_g3.04438}.
We again see that the iMM algorithm removes an amount of entanglement $\alpha(g)$ per column consistent with the area law.
Additionally, the physical entropy $S(\ket{\Phi_\ell'})$ agrees with the entropy from iDMRG, again indicating that the orthogonality column of $L_x - \ell$ columns and the iMPS of the same width have large overlap. 
We note that for $g=3.50$, the amount of entanglement removed per column saturates to $\alpha(g)$ sooner than in the $g=304438$ case, indicating that more horizontal entropy is present in latter case.

\begin{figure}[htb]
    \centering
    \includegraphics[width=\columnwidth]{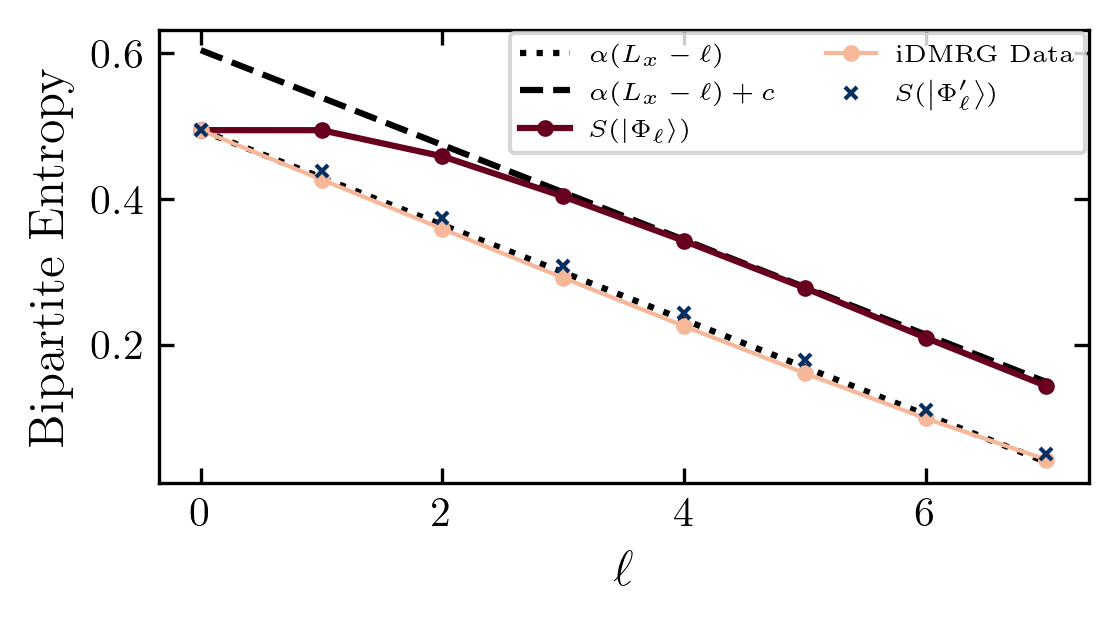}
    \caption{Entropy of $\ket{\Phi_\ell}$ as a function of number of columns $\ell$ removed by peeling the ground state of the 2D TFI Hamiltonian with width $L_x=8$ and $g=3.04438$. 
    After an initial delay, explained by $S(\ket{\Phi_\ell})$ not being a physical entropy, iMM removes an amount of entanglement consistent with the area law. 
    Physical entropy $S(\ket{\Phi_\ell'})$ (as defined in Fig.~\ref{fig:area_law_appendix}(a)) and entropy from iDMRG GS of different width strips agree and obey area law.}
    \label{fig:area_law_g3.04438}
\end{figure}

Next we use the iTEBD\textsuperscript{2} algorithm to search for the ground state of the $g=3.04438$ 2D TFI. 
Again we compare isoTNS energies evaluated by iMM against iDMRG energies for $\chi=512$.
Results for $L_x=4$, $L_x=8$, and $L_x=20$ infinite strips are shown in Fig.~\ref{fig:iTEBD_energy_g3.04438}, where again we find an intermediate $\mathrm{d}t$, which balances iTEBD and iMM errors, leads to the optimal energies.

\begin{figure*}[htb]
    \centering
    \includegraphics[width=\textwidth]{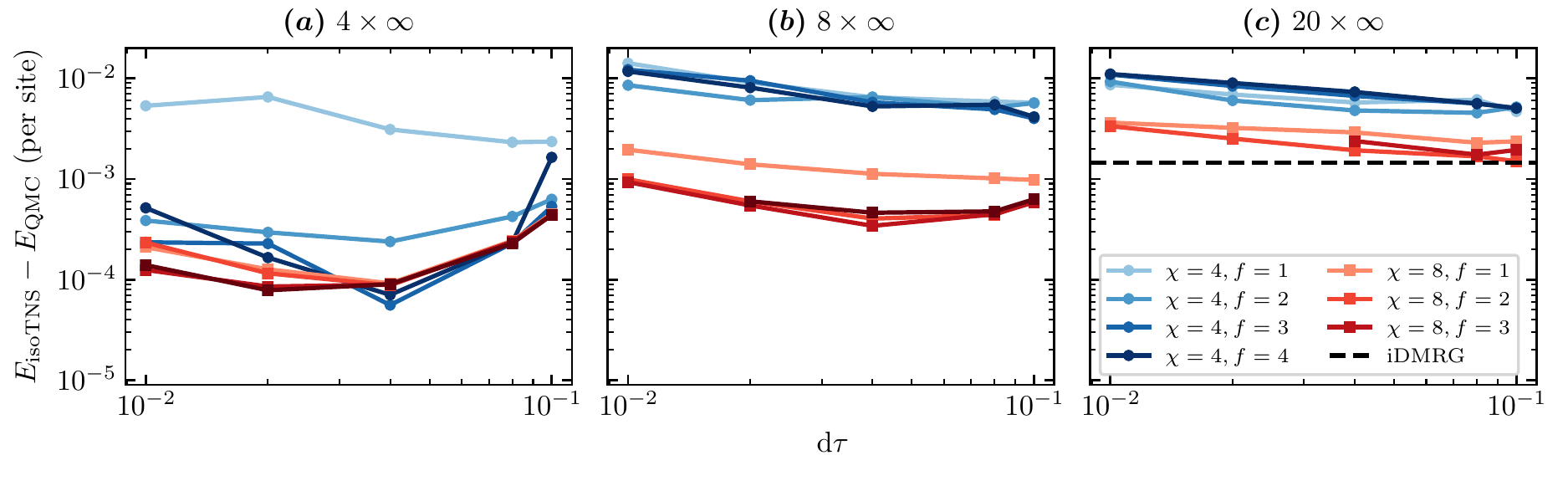}
    \caption{Ground state energies achieved with iTEBD\textsuperscript{2} for critical $g=3.04438$ 2D TFI on $L_x=4, 8, 20$  strip.
    An intermediate $\mathrm{d}\tau$ and $f$ yield the best energy, while $\chi=8$ outperforms $\chi=4$.
    We compare the isoTNS energies against essentially exact energies from quantum Monte Carlo (QMC) extrapolated from strips of finite length.
    As a comparison, the dashed line is the result of an iDMRG calculation with bond dimension $\chi = 512$. For $L_x=4, 8$, the iDMRG result is below the bottom axis of the plot.
    }
    \label{fig:iTEBD_energy_g3.04438}
\end{figure*}

\clearpage
\bibliography{ref}

\begin{thebibliography}{76}%
\makeatletter
\providecommand \@ifxundefined [1]{%
 \@ifx{#1\undefined}
}%
\providecommand \@ifnum [1]{%
 \ifnum #1\expandafter \@firstoftwo
 \else \expandafter \@secondoftwo
 \fi
}%
\providecommand \@ifx [1]{%
 \ifx #1\expandafter \@firstoftwo
 \else \expandafter \@secondoftwo
 \fi
}%
\providecommand \natexlab [1]{#1}%
\providecommand \enquote  [1]{``#1''}%
\providecommand \bibnamefont  [1]{#1}%
\providecommand \bibfnamefont [1]{#1}%
\providecommand \citenamefont [1]{#1}%
\providecommand \href@noop [0]{\@secondoftwo}%
\providecommand \href [0]{\begingroup \@sanitize@url \@href}%
\providecommand \@href[1]{\@@startlink{#1}\@@href}%
\providecommand \@@href[1]{\endgroup#1\@@endlink}%
\providecommand \@sanitize@url [0]{\catcode `\\12\catcode `\$12\catcode
  `\&12\catcode `\#12\catcode `\^12\catcode `\_12\catcode `\%12\relax}%
\providecommand \@@startlink[1]{}%
\providecommand \@@endlink[0]{}%
\providecommand \url  [0]{\begingroup\@sanitize@url \@url }%
\providecommand \@url [1]{\endgroup\@href {#1}{\urlprefix }}%
\providecommand \urlprefix  [0]{URL }%
\providecommand \Eprint [0]{\href }%
\providecommand \doibase [0]{https://doi.org/}%
\providecommand \selectlanguage [0]{\@gobble}%
\providecommand \bibinfo  [0]{\@secondoftwo}%
\providecommand \bibfield  [0]{\@secondoftwo}%
\providecommand \translation [1]{[#1]}%
\providecommand \BibitemOpen [0]{}%
\providecommand \bibitemStop [0]{}%
\providecommand \bibitemNoStop [0]{.\EOS\space}%
\providecommand \EOS [0]{\spacefactor3000\relax}%
\providecommand \BibitemShut  [1]{\csname bibitem#1\endcsname}%
\let\auto@bib@innerbib\@empty
\bibitem [{\citenamefont {Cirac}\ \emph {et~al.}(2021)\citenamefont {Cirac},
  \citenamefont {Pérez-García}, \citenamefont {Schuch},\ and\ \citenamefont
  {Verstraete}}]{cirac2021review}%
  \BibitemOpen
  \bibfield  {author} {\bibinfo {author} {\bibfnamefont {J.~I.}\ \bibnamefont
  {Cirac}}, \bibinfo {author} {\bibfnamefont {D.}~\bibnamefont
  {Pérez-García}}, \bibinfo {author} {\bibfnamefont {N.}~\bibnamefont
  {Schuch}},\ and\ \bibinfo {author} {\bibfnamefont {F.}~\bibnamefont
  {Verstraete}},\ }\bibfield  {title} {\bibinfo {title} {Matrix product states
  and projected entangled pair states: Concepts, symmetries, theorems},\
  }\bibfield  {journal} {\bibinfo  {journal} {Reviews of Modern Physics}\
  }\textbf {\bibinfo {volume} {93}},\ \href
  {https://doi.org/10.1103/revmodphys.93.045003} {10.1103/revmodphys.93.045003}
  (\bibinfo {year} {2021})\BibitemShut {NoStop}%
\bibitem [{\citenamefont {Fannes}\ \emph {et~al.}(1992)\citenamefont {Fannes},
  \citenamefont {Nachtergaele},\ and\ \citenamefont
  {Werner}}]{fannes1992abundance}%
  \BibitemOpen
  \bibfield  {author} {\bibinfo {author} {\bibfnamefont {M.}~\bibnamefont
  {Fannes}}, \bibinfo {author} {\bibfnamefont {B.}~\bibnamefont
  {Nachtergaele}},\ and\ \bibinfo {author} {\bibfnamefont {R.}~\bibnamefont
  {Werner}},\ }\bibfield  {title} {\bibinfo {title} {Abundance of translation
  invariant pure states on quantum spin chains},\ }\href@noop {} {\bibfield
  {journal} {\bibinfo  {journal} {letters in mathematical physics}\ }\textbf
  {\bibinfo {volume} {25}},\ \bibinfo {pages} {249} (\bibinfo {year}
  {1992})}\BibitemShut {NoStop}%
\bibitem [{\citenamefont {Schollw{\"o}ck}(2011)}]{schollwock2011density}%
  \BibitemOpen
  \bibfield  {author} {\bibinfo {author} {\bibfnamefont {U.}~\bibnamefont
  {Schollw{\"o}ck}},\ }\bibfield  {title} {\bibinfo {title} {The density-matrix
  renormalization group in the age of matrix product states},\ }\href@noop {}
  {\bibfield  {journal} {\bibinfo  {journal} {Annals of physics}\ }\textbf
  {\bibinfo {volume} {326}},\ \bibinfo {pages} {96} (\bibinfo {year}
  {2011})}\BibitemShut {NoStop}%
\bibitem [{\citenamefont {Bridgeman}\ and\ \citenamefont
  {Chubb}(2017)}]{bridgeman2017hand}%
  \BibitemOpen
  \bibfield  {author} {\bibinfo {author} {\bibfnamefont {J.~C.}\ \bibnamefont
  {Bridgeman}}\ and\ \bibinfo {author} {\bibfnamefont {C.~T.}\ \bibnamefont
  {Chubb}},\ }\bibfield  {title} {\bibinfo {title} {Hand-waving and
  interpretive dance: an introductory course on tensor networks},\ }\href@noop
  {} {\bibfield  {journal} {\bibinfo  {journal} {Journal of Physics A:
  Mathematical and Theoretical}\ }\textbf {\bibinfo {volume} {50}},\ \bibinfo
  {pages} {223001} (\bibinfo {year} {2017})}\BibitemShut {NoStop}%
\bibitem [{\citenamefont {Paeckel}\ \emph {et~al.}(2019)\citenamefont
  {Paeckel}, \citenamefont {Köhler}, \citenamefont {Swoboda}, \citenamefont
  {Manmana}, \citenamefont {Schollwöck},\ and\ \citenamefont
  {Hubig}}]{paeckel2019time}%
  \BibitemOpen
  \bibfield  {author} {\bibinfo {author} {\bibfnamefont {S.}~\bibnamefont
  {Paeckel}}, \bibinfo {author} {\bibfnamefont {T.}~\bibnamefont {Köhler}},
  \bibinfo {author} {\bibfnamefont {A.}~\bibnamefont {Swoboda}}, \bibinfo
  {author} {\bibfnamefont {S.~R.}\ \bibnamefont {Manmana}}, \bibinfo {author}
  {\bibfnamefont {U.}~\bibnamefont {Schollwöck}},\ and\ \bibinfo {author}
  {\bibfnamefont {C.}~\bibnamefont {Hubig}},\ }\bibfield  {title} {\bibinfo
  {title} {Time-evolution methods for matrix-product states},\ }\href
  {https://doi.org/10.1016/j.aop.2019.167998} {\bibfield  {journal} {\bibinfo
  {journal} {Annals of Physics}\ }\textbf {\bibinfo {volume} {411}},\ \bibinfo
  {pages} {167998} (\bibinfo {year} {2019})}\BibitemShut {NoStop}%
\bibitem [{\citenamefont {Vanderstraeten}\ \emph {et~al.}(2019)\citenamefont
  {Vanderstraeten}, \citenamefont {Haegeman},\ and\ \citenamefont
  {Verstraete}}]{vanderstraeten2019tangent}%
  \BibitemOpen
  \bibfield  {author} {\bibinfo {author} {\bibfnamefont {L.}~\bibnamefont
  {Vanderstraeten}}, \bibinfo {author} {\bibfnamefont {J.}~\bibnamefont
  {Haegeman}},\ and\ \bibinfo {author} {\bibfnamefont {F.}~\bibnamefont
  {Verstraete}},\ }\bibfield  {title} {\bibinfo {title} {Tangent-space methods
  for uniform matrix product states},\ }\href@noop {} {\bibfield  {journal}
  {\bibinfo  {journal} {SciPost Physics Lecture Notes}\ } (\bibinfo {year}
  {2019})}\BibitemShut {NoStop}%
\bibitem [{\citenamefont {White}(1992)}]{white1992density}%
  \BibitemOpen
  \bibfield  {author} {\bibinfo {author} {\bibfnamefont {S.~R.}\ \bibnamefont
  {White}},\ }\bibfield  {title} {\bibinfo {title} {Density matrix formulation
  for quantum renormalization groups},\ }\href@noop {} {\bibfield  {journal}
  {\bibinfo  {journal} {Physical review letters}\ }\textbf {\bibinfo {volume}
  {69}},\ \bibinfo {pages} {2863} (\bibinfo {year} {1992})}\BibitemShut
  {NoStop}%
\bibitem [{\citenamefont {White}(1993)}]{white1993density}%
  \BibitemOpen
  \bibfield  {author} {\bibinfo {author} {\bibfnamefont {S.~R.}\ \bibnamefont
  {White}},\ }\bibfield  {title} {\bibinfo {title} {Density-matrix algorithms
  for quantum renormalization groups},\ }\href@noop {} {\bibfield  {journal}
  {\bibinfo  {journal} {Physical review b}\ }\textbf {\bibinfo {volume} {48}},\
  \bibinfo {pages} {10345} (\bibinfo {year} {1993})}\BibitemShut {NoStop}%
\bibitem [{\citenamefont {Vidal}(2003)}]{vidal2003efficient}%
  \BibitemOpen
  \bibfield  {author} {\bibinfo {author} {\bibfnamefont {G.}~\bibnamefont
  {Vidal}},\ }\bibfield  {title} {\bibinfo {title} {Efficient classical
  simulation of slightly entangled quantum computations},\ }\href@noop {}
  {\bibfield  {journal} {\bibinfo  {journal} {Physical review letters}\
  }\textbf {\bibinfo {volume} {91}},\ \bibinfo {pages} {147902} (\bibinfo
  {year} {2003})}\BibitemShut {NoStop}%
\bibitem [{\citenamefont {Vidal}(2004)}]{vidal2004efficient}%
  \BibitemOpen
  \bibfield  {author} {\bibinfo {author} {\bibfnamefont {G.}~\bibnamefont
  {Vidal}},\ }\bibfield  {title} {\bibinfo {title} {Efficient simulation of
  one-dimensional quantum many-body systems},\ }\href@noop {} {\bibfield
  {journal} {\bibinfo  {journal} {Physical review letters}\ }\textbf {\bibinfo
  {volume} {93}},\ \bibinfo {pages} {040502} (\bibinfo {year}
  {2004})}\BibitemShut {NoStop}%
\bibitem [{\citenamefont {Haegeman}\ \emph {et~al.}(2011)\citenamefont
  {Haegeman}, \citenamefont {Cirac}, \citenamefont {Osborne}, \citenamefont
  {Pižorn}, \citenamefont {Verschelde},\ and\ \citenamefont
  {Verstraete}}]{haegeman2011tdvp}%
  \BibitemOpen
  \bibfield  {author} {\bibinfo {author} {\bibfnamefont {J.}~\bibnamefont
  {Haegeman}}, \bibinfo {author} {\bibfnamefont {J.~I.}\ \bibnamefont {Cirac}},
  \bibinfo {author} {\bibfnamefont {T.~J.}\ \bibnamefont {Osborne}}, \bibinfo
  {author} {\bibfnamefont {I.}~\bibnamefont {Pižorn}}, \bibinfo {author}
  {\bibfnamefont {H.}~\bibnamefont {Verschelde}},\ and\ \bibinfo {author}
  {\bibfnamefont {F.}~\bibnamefont {Verstraete}},\ }\bibfield  {title}
  {\bibinfo {title} {Time-dependent variational principle for quantum
  lattices},\ }\bibfield  {journal} {\bibinfo  {journal} {Physical Review
  Letters}\ }\textbf {\bibinfo {volume} {107}},\ \href
  {https://doi.org/10.1103/physrevlett.107.070601}
  {10.1103/physrevlett.107.070601} (\bibinfo {year} {2011})\BibitemShut
  {NoStop}%
\bibitem [{\citenamefont {Haegeman}\ \emph {et~al.}(2016)\citenamefont
  {Haegeman}, \citenamefont {Lubich}, \citenamefont {Oseledets}, \citenamefont
  {Vandereycken},\ and\ \citenamefont {Verstraete}}]{haegeman2016tdvp}%
  \BibitemOpen
  \bibfield  {author} {\bibinfo {author} {\bibfnamefont {J.}~\bibnamefont
  {Haegeman}}, \bibinfo {author} {\bibfnamefont {C.}~\bibnamefont {Lubich}},
  \bibinfo {author} {\bibfnamefont {I.}~\bibnamefont {Oseledets}}, \bibinfo
  {author} {\bibfnamefont {B.}~\bibnamefont {Vandereycken}},\ and\ \bibinfo
  {author} {\bibfnamefont {F.}~\bibnamefont {Verstraete}},\ }\bibfield  {title}
  {\bibinfo {title} {Unifying time evolution and optimization with matrix
  product states},\ }\bibfield  {journal} {\bibinfo  {journal} {Physical Review
  B}\ }\textbf {\bibinfo {volume} {94}},\ \href
  {https://doi.org/10.1103/physrevb.94.165116} {10.1103/physrevb.94.165116}
  (\bibinfo {year} {2016})\BibitemShut {NoStop}%
\bibitem [{Note1()}]{Note1}%
  \BibitemOpen
  \bibinfo {note} {The isometric form is the first of two conditions in the
  better known canonical form. An injective MPS is in canonical form if (i,
  isometry) $\DOTSB \sum@ \slimits@ _{i,\sigma } A^{\sigma }{i,j} \protect \bar
  {A}^{\sigma }_{i,k} = \delta _{j,k}$ and (ii, diagonal density matrix)
  $\DOTSB \sum@ \slimits@ _{j,j',\sigma } A^{\sigma }{i,j} \protect \bar
  {A}^{\sigma }_{i',j'} \Lambda _{j,j'} = \Lambda '_{i,i'}$, where $\Lambda ,
  \Lambda '$ are both positive, diagonal matrices.}\BibitemShut {Stop}%
\bibitem [{\citenamefont {Perez-Garcia}\ \emph {et~al.}(2007)\citenamefont
  {Perez-Garcia}, \citenamefont {Verstraete}, \citenamefont {Wolf},\ and\
  \citenamefont {Cirac}}]{perez2006matrix}%
  \BibitemOpen
  \bibfield  {author} {\bibinfo {author} {\bibfnamefont {D.}~\bibnamefont
  {Perez-Garcia}}, \bibinfo {author} {\bibfnamefont {F.}~\bibnamefont
  {Verstraete}}, \bibinfo {author} {\bibfnamefont {M.~M.}\ \bibnamefont
  {Wolf}},\ and\ \bibinfo {author} {\bibfnamefont {J.~I.}\ \bibnamefont
  {Cirac}},\ }\bibfield  {title} {\bibinfo {title} {Matrix product state
  representations},\ }\href@noop {} {\bibfield  {journal} {\bibinfo  {journal}
  {Quantum Info. Comput.}\ }\textbf {\bibinfo {volume} {7}},\ \bibinfo {pages}
  {401–430} (\bibinfo {year} {2007})}\BibitemShut {NoStop}%
\bibitem [{\citenamefont {Richter}(1995)}]{richter1995construction}%
  \BibitemOpen
  \bibfield  {author} {\bibinfo {author} {\bibfnamefont {S.}~\bibnamefont
  {Richter}},\ }\href@noop {} {\emph {\bibinfo {title} {Construction of states
  on two-dimensional lattices and quantum cellular automata}}}\ (\bibinfo
  {publisher} {Shaker},\ \bibinfo {year} {1995})\BibitemShut {NoStop}%
\bibitem [{\citenamefont {Niggemann}\ \emph {et~al.}(1997)\citenamefont
  {Niggemann}, \citenamefont {Kl{\"u}mper},\ and\ \citenamefont
  {Zittartz}}]{niggemann1997quantum}%
  \BibitemOpen
  \bibfield  {author} {\bibinfo {author} {\bibfnamefont {H.}~\bibnamefont
  {Niggemann}}, \bibinfo {author} {\bibfnamefont {A.}~\bibnamefont
  {Kl{\"u}mper}},\ and\ \bibinfo {author} {\bibfnamefont {J.}~\bibnamefont
  {Zittartz}},\ }\bibfield  {title} {\bibinfo {title} {Quantum phase transition
  in spin-3/2 systems on the hexagonal lattice—optimum ground state
  approach},\ }\href@noop {} {\bibfield  {journal} {\bibinfo  {journal}
  {Zeitschrift f{\"u}r Physik B Condensed Matter}\ }\textbf {\bibinfo {volume}
  {104}},\ \bibinfo {pages} {103} (\bibinfo {year} {1997})}\BibitemShut
  {NoStop}%
\bibitem [{\citenamefont {Sierra}(1998)}]{sierra1998density}%
  \BibitemOpen
  \bibfield  {author} {\bibinfo {author} {\bibfnamefont {G.}~\bibnamefont
  {Sierra}},\ }\bibfield  {title} {\bibinfo {title} {The density matrix
  renormalization group, quantum groups and conformal field theory},\ }in\
  \href@noop {} {\emph {\bibinfo {booktitle} {Proceedings of the Workshop on
  the Exact Renormalization Group}}}\ (\bibinfo  {publisher} {World
  Scientific},\ \bibinfo {year} {1998})\BibitemShut {NoStop}%
\bibitem [{\citenamefont {Nishino}\ and\ \citenamefont
  {Okunishi}(1998)}]{nishino1998density}%
  \BibitemOpen
  \bibfield  {author} {\bibinfo {author} {\bibfnamefont {T.}~\bibnamefont
  {Nishino}}\ and\ \bibinfo {author} {\bibfnamefont {K.}~\bibnamefont
  {Okunishi}},\ }\bibfield  {title} {\bibinfo {title} {A density matrix
  algorithm for 3d classical models},\ }\href@noop {} {\bibfield  {journal}
  {\bibinfo  {journal} {Journal of the Physical Society of Japan}\ }\textbf
  {\bibinfo {volume} {67}},\ \bibinfo {pages} {3066} (\bibinfo {year}
  {1998})}\BibitemShut {NoStop}%
\bibitem [{\citenamefont {Verstraete}\ and\ \citenamefont
  {Cirac}(2004)}]{verstraete2004renormalization}%
  \BibitemOpen
  \bibfield  {author} {\bibinfo {author} {\bibfnamefont {F.}~\bibnamefont
  {Verstraete}}\ and\ \bibinfo {author} {\bibfnamefont {J.~I.}\ \bibnamefont
  {Cirac}},\ }\bibfield  {title} {\bibinfo {title} {Renormalization algorithms
  for quantum-many body systems in two and higher dimensions},\ }\href@noop {}
  {\bibfield  {journal} {\bibinfo  {journal} {arXiv preprint cond-mat/0407066}\
  } (\bibinfo {year} {2004})}\BibitemShut {NoStop}%
\bibitem [{\citenamefont {Pi{\v{z}}orn}\ and\ \citenamefont
  {Verstraete}(2010)}]{pivzorn2010fermionic}%
  \BibitemOpen
  \bibfield  {author} {\bibinfo {author} {\bibfnamefont {I.}~\bibnamefont
  {Pi{\v{z}}orn}}\ and\ \bibinfo {author} {\bibfnamefont {F.}~\bibnamefont
  {Verstraete}},\ }\bibfield  {title} {\bibinfo {title} {Fermionic
  implementation of projected entangled pair states algorithm},\ }\href@noop {}
  {\bibfield  {journal} {\bibinfo  {journal} {Physical Review B}\ }\textbf
  {\bibinfo {volume} {81}},\ \bibinfo {pages} {245110} (\bibinfo {year}
  {2010})}\BibitemShut {NoStop}%
\bibitem [{\citenamefont {Lubasch}\ \emph
  {et~al.}(2014{\natexlab{a}})\citenamefont {Lubasch}, \citenamefont {Cirac},\
  and\ \citenamefont {Banuls}}]{lubasch2014algorithms}%
  \BibitemOpen
  \bibfield  {author} {\bibinfo {author} {\bibfnamefont {M.}~\bibnamefont
  {Lubasch}}, \bibinfo {author} {\bibfnamefont {J.~I.}\ \bibnamefont {Cirac}},\
  and\ \bibinfo {author} {\bibfnamefont {M.-C.}\ \bibnamefont {Banuls}},\
  }\bibfield  {title} {\bibinfo {title} {Algorithms for finite projected
  entangled pair states},\ }\href@noop {} {\bibfield  {journal} {\bibinfo
  {journal} {Physical Review B}\ }\textbf {\bibinfo {volume} {90}},\ \bibinfo
  {pages} {064425} (\bibinfo {year} {2014}{\natexlab{a}})}\BibitemShut
  {NoStop}%
\bibitem [{\citenamefont {Liu}\ \emph {et~al.}(2017)\citenamefont {Liu},
  \citenamefont {Dong}, \citenamefont {Han}, \citenamefont {Guo},\ and\
  \citenamefont {He}}]{liu2017gradient}%
  \BibitemOpen
  \bibfield  {author} {\bibinfo {author} {\bibfnamefont {W.-Y.}\ \bibnamefont
  {Liu}}, \bibinfo {author} {\bibfnamefont {S.-J.}\ \bibnamefont {Dong}},
  \bibinfo {author} {\bibfnamefont {Y.-J.}\ \bibnamefont {Han}}, \bibinfo
  {author} {\bibfnamefont {G.-C.}\ \bibnamefont {Guo}},\ and\ \bibinfo {author}
  {\bibfnamefont {L.}~\bibnamefont {He}},\ }\bibfield  {title} {\bibinfo
  {title} {Gradient optimization of finite projected entangled pair states},\
  }\href@noop {} {\bibfield  {journal} {\bibinfo  {journal} {Physical Review
  B}\ }\textbf {\bibinfo {volume} {95}},\ \bibinfo {pages} {195154} (\bibinfo
  {year} {2017})}\BibitemShut {NoStop}%
\bibitem [{\citenamefont {Liu}\ \emph {et~al.}(2021)\citenamefont {Liu},
  \citenamefont {Huang}, \citenamefont {Gong},\ and\ \citenamefont
  {Gu}}]{liu2021accurate}%
  \BibitemOpen
  \bibfield  {author} {\bibinfo {author} {\bibfnamefont {W.-Y.}\ \bibnamefont
  {Liu}}, \bibinfo {author} {\bibfnamefont {Y.-Z.}\ \bibnamefont {Huang}},
  \bibinfo {author} {\bibfnamefont {S.-S.}\ \bibnamefont {Gong}},\ and\
  \bibinfo {author} {\bibfnamefont {Z.-C.}\ \bibnamefont {Gu}},\ }\bibfield
  {title} {\bibinfo {title} {Accurate simulation for finite projected entangled
  pair states in two dimensions},\ }\bibfield  {journal} {\bibinfo  {journal}
  {Physical Review B}\ }\textbf {\bibinfo {volume} {103}},\ \href
  {https://doi.org/10.1103/physrevb.103.235155} {10.1103/physrevb.103.235155}
  (\bibinfo {year} {2021})\BibitemShut {NoStop}%
\bibitem [{\citenamefont {Vieijra}\ \emph {et~al.}(2021)\citenamefont
  {Vieijra}, \citenamefont {Haegeman}, \citenamefont {Verstraete},\ and\
  \citenamefont {Vanderstraeten}}]{vieijra2021direct}%
  \BibitemOpen
  \bibfield  {author} {\bibinfo {author} {\bibfnamefont {T.}~\bibnamefont
  {Vieijra}}, \bibinfo {author} {\bibfnamefont {J.}~\bibnamefont {Haegeman}},
  \bibinfo {author} {\bibfnamefont {F.}~\bibnamefont {Verstraete}},\ and\
  \bibinfo {author} {\bibfnamefont {L.}~\bibnamefont {Vanderstraeten}},\
  }\bibfield  {title} {\bibinfo {title} {Direct sampling of projected
  entangled-pair states},\ }\href@noop {} {\bibfield  {journal} {\bibinfo
  {journal} {arXiv preprint arXiv:2109.07356}\ } (\bibinfo {year}
  {2021})}\BibitemShut {NoStop}%
\bibitem [{\citenamefont {Phien}\ \emph {et~al.}(2015)\citenamefont {Phien},
  \citenamefont {Bengua}, \citenamefont {Tuan}, \citenamefont {Corboz},\ and\
  \citenamefont {Or{\'u}s}}]{phien2015infinite}%
  \BibitemOpen
  \bibfield  {author} {\bibinfo {author} {\bibfnamefont {H.~N.}\ \bibnamefont
  {Phien}}, \bibinfo {author} {\bibfnamefont {J.~A.}\ \bibnamefont {Bengua}},
  \bibinfo {author} {\bibfnamefont {H.~D.}\ \bibnamefont {Tuan}}, \bibinfo
  {author} {\bibfnamefont {P.}~\bibnamefont {Corboz}},\ and\ \bibinfo {author}
  {\bibfnamefont {R.}~\bibnamefont {Or{\'u}s}},\ }\bibfield  {title} {\bibinfo
  {title} {Infinite projected entangled pair states algorithm improved: Fast
  full update and gauge fixing},\ }\href@noop {} {\bibfield  {journal}
  {\bibinfo  {journal} {Physical Review B}\ }\textbf {\bibinfo {volume} {92}},\
  \bibinfo {pages} {035142} (\bibinfo {year} {2015})}\BibitemShut {NoStop}%
\bibitem [{\citenamefont {Vanderstraeten}\ \emph {et~al.}(2016)\citenamefont
  {Vanderstraeten}, \citenamefont {Haegeman}, \citenamefont {Corboz},\ and\
  \citenamefont {Verstraete}}]{vanderstraeten2016iPEPS}%
  \BibitemOpen
  \bibfield  {author} {\bibinfo {author} {\bibfnamefont {L.}~\bibnamefont
  {Vanderstraeten}}, \bibinfo {author} {\bibfnamefont {J.}~\bibnamefont
  {Haegeman}}, \bibinfo {author} {\bibfnamefont {P.}~\bibnamefont {Corboz}},\
  and\ \bibinfo {author} {\bibfnamefont {F.}~\bibnamefont {Verstraete}},\
  }\bibfield  {title} {\bibinfo {title} {Gradient methods for variational
  optimization of projected entangled-pair states},\ }\bibfield  {journal}
  {\bibinfo  {journal} {Physical Review B}\ }\textbf {\bibinfo {volume} {94}},\
  \href {https://doi.org/10.1103/physrevb.94.155123}
  {10.1103/physrevb.94.155123} (\bibinfo {year} {2016})\BibitemShut {NoStop}%
\bibitem [{\citenamefont {Corboz}(2016)}]{corboz2016iPEPS}%
  \BibitemOpen
  \bibfield  {author} {\bibinfo {author} {\bibfnamefont {P.}~\bibnamefont
  {Corboz}},\ }\bibfield  {title} {\bibinfo {title} {Variational optimization
  with infinite projected entangled-pair states},\ }\bibfield  {journal}
  {\bibinfo  {journal} {Physical Review B}\ }\textbf {\bibinfo {volume} {94}},\
  \href {https://doi.org/10.1103/physrevb.94.035133}
  {10.1103/physrevb.94.035133} (\bibinfo {year} {2016})\BibitemShut {NoStop}%
\bibitem [{\citenamefont {Liao}\ \emph {et~al.}(2019)\citenamefont {Liao},
  \citenamefont {Liu}, \citenamefont {Wang},\ and\ \citenamefont
  {Xiang}}]{liao2019differentiable}%
  \BibitemOpen
  \bibfield  {author} {\bibinfo {author} {\bibfnamefont {H.-J.}\ \bibnamefont
  {Liao}}, \bibinfo {author} {\bibfnamefont {J.-G.}\ \bibnamefont {Liu}},
  \bibinfo {author} {\bibfnamefont {L.}~\bibnamefont {Wang}},\ and\ \bibinfo
  {author} {\bibfnamefont {T.}~\bibnamefont {Xiang}},\ }\bibfield  {title}
  {\bibinfo {title} {Differentiable programming tensor networks},\ }\href@noop
  {} {\bibfield  {journal} {\bibinfo  {journal} {Physical Review X}\ }\textbf
  {\bibinfo {volume} {9}},\ \bibinfo {pages} {031041} (\bibinfo {year}
  {2019})}\BibitemShut {NoStop}%
\bibitem [{\citenamefont {Schuch}\ \emph {et~al.}(2007)\citenamefont {Schuch},
  \citenamefont {Wolf}, \citenamefont {Verstraete},\ and\ \citenamefont
  {Cirac}}]{schuch2007computational}%
  \BibitemOpen
  \bibfield  {author} {\bibinfo {author} {\bibfnamefont {N.}~\bibnamefont
  {Schuch}}, \bibinfo {author} {\bibfnamefont {M.~M.}\ \bibnamefont {Wolf}},
  \bibinfo {author} {\bibfnamefont {F.}~\bibnamefont {Verstraete}},\ and\
  \bibinfo {author} {\bibfnamefont {J.~I.}\ \bibnamefont {Cirac}},\ }\bibfield
  {title} {\bibinfo {title} {Computational complexity of projected entangled
  pair states},\ }\href@noop {} {\bibfield  {journal} {\bibinfo  {journal}
  {Physical review letters}\ }\textbf {\bibinfo {volume} {98}},\ \bibinfo
  {pages} {140506} (\bibinfo {year} {2007})}\BibitemShut {NoStop}%
\bibitem [{\citenamefont {Scarpa}\ \emph {et~al.}(2020)\citenamefont {Scarpa},
  \citenamefont {Molnár}, \citenamefont {Ge}, \citenamefont {García-Ripoll},
  \citenamefont {Schuch}, \citenamefont {Pérez-García},\ and\ \citenamefont
  {Iblisdir}}]{scarpa2020limitations}%
  \BibitemOpen
  \bibfield  {author} {\bibinfo {author} {\bibfnamefont {G.}~\bibnamefont
  {Scarpa}}, \bibinfo {author} {\bibfnamefont {A.}~\bibnamefont {Molnár}},
  \bibinfo {author} {\bibfnamefont {Y.}~\bibnamefont {Ge}}, \bibinfo {author}
  {\bibfnamefont {J.}~\bibnamefont {García-Ripoll}}, \bibinfo {author}
  {\bibfnamefont {N.}~\bibnamefont {Schuch}}, \bibinfo {author} {\bibfnamefont
  {D.}~\bibnamefont {Pérez-García}},\ and\ \bibinfo {author} {\bibfnamefont
  {S.}~\bibnamefont {Iblisdir}},\ }\bibfield  {title} {\bibinfo {title}
  {Projected entangled pair states: Fundamental analytical and numerical
  limitations},\ }\bibfield  {journal} {\bibinfo  {journal} {Physical Review
  Letters}\ }\textbf {\bibinfo {volume} {125}},\ \href
  {https://doi.org/10.1103/physrevlett.125.210504}
  {10.1103/physrevlett.125.210504} (\bibinfo {year} {2020})\BibitemShut
  {NoStop}%
\bibitem [{\citenamefont {Haferkamp}\ \emph {et~al.}(2020)\citenamefont
  {Haferkamp}, \citenamefont {Hangleiter}, \citenamefont {Eisert},\ and\
  \citenamefont {Gluza}}]{haferkamp2020hard}%
  \BibitemOpen
  \bibfield  {author} {\bibinfo {author} {\bibfnamefont {J.}~\bibnamefont
  {Haferkamp}}, \bibinfo {author} {\bibfnamefont {D.}~\bibnamefont
  {Hangleiter}}, \bibinfo {author} {\bibfnamefont {J.}~\bibnamefont {Eisert}},\
  and\ \bibinfo {author} {\bibfnamefont {M.}~\bibnamefont {Gluza}},\ }\bibfield
   {title} {\bibinfo {title} {Contracting projected entangled pair states is
  average-case hard},\ }\bibfield  {journal} {\bibinfo  {journal} {Physical
  Review Research}\ }\textbf {\bibinfo {volume} {2}},\ \href
  {https://doi.org/10.1103/physrevresearch.2.013010}
  {10.1103/physrevresearch.2.013010} (\bibinfo {year} {2020})\BibitemShut
  {NoStop}%
\bibitem [{\citenamefont {Or{\'u}s}\ and\ \citenamefont
  {Vidal}(2009)}]{orus2009simulation}%
  \BibitemOpen
  \bibfield  {author} {\bibinfo {author} {\bibfnamefont {R.}~\bibnamefont
  {Or{\'u}s}}\ and\ \bibinfo {author} {\bibfnamefont {G.}~\bibnamefont
  {Vidal}},\ }\bibfield  {title} {\bibinfo {title} {Simulation of
  two-dimensional quantum systems on an infinite lattice revisited: Corner
  transfer matrix for tensor contraction},\ }\href@noop {} {\bibfield
  {journal} {\bibinfo  {journal} {Physical Review B}\ }\textbf {\bibinfo
  {volume} {80}},\ \bibinfo {pages} {094403} (\bibinfo {year}
  {2009})}\BibitemShut {NoStop}%
\bibitem [{\citenamefont {Fishman}\ \emph {et~al.}(2017)\citenamefont
  {Fishman}, \citenamefont {Vanderstraeten}, \citenamefont {Zauner-Stauber},
  \citenamefont {Haegeman},\ and\ \citenamefont
  {Verstraete}}]{fishman2017faster}%
  \BibitemOpen
  \bibfield  {author} {\bibinfo {author} {\bibfnamefont {M.~T.}\ \bibnamefont
  {Fishman}}, \bibinfo {author} {\bibfnamefont {L.}~\bibnamefont
  {Vanderstraeten}}, \bibinfo {author} {\bibfnamefont {V.}~\bibnamefont
  {Zauner-Stauber}}, \bibinfo {author} {\bibfnamefont {J.}~\bibnamefont
  {Haegeman}},\ and\ \bibinfo {author} {\bibfnamefont {F.}~\bibnamefont
  {Verstraete}},\ }\bibfield  {title} {\bibinfo {title} {Faster methods for
  contracting infinite 2d tensor networks},\ }\href@noop {} {\bibfield
  {journal} {\bibinfo  {journal} {arXiv preprint arXiv:1711.05881}\ } (\bibinfo
  {year} {2017})}\BibitemShut {NoStop}%
\bibitem [{\citenamefont {Evenbly}\ and\ \citenamefont
  {Vidal}(2015)}]{evenbly2015TNR}%
  \BibitemOpen
  \bibfield  {author} {\bibinfo {author} {\bibfnamefont {G.}~\bibnamefont
  {Evenbly}}\ and\ \bibinfo {author} {\bibfnamefont {G.}~\bibnamefont
  {Vidal}},\ }\bibfield  {title} {\bibinfo {title} {Tensor network
  renormalization},\ }\bibfield  {journal} {\bibinfo  {journal} {Physical
  Review Letters}\ }\textbf {\bibinfo {volume} {115}},\ \href
  {https://doi.org/10.1103/physrevlett.115.180405}
  {10.1103/physrevlett.115.180405} (\bibinfo {year} {2015})\BibitemShut
  {NoStop}%
\bibitem [{\citenamefont {Zaletel}\ and\ \citenamefont
  {Pollmann}(2020)}]{zaletel2020isometric}%
  \BibitemOpen
  \bibfield  {author} {\bibinfo {author} {\bibfnamefont {M.~P.}\ \bibnamefont
  {Zaletel}}\ and\ \bibinfo {author} {\bibfnamefont {F.}~\bibnamefont
  {Pollmann}},\ }\bibfield  {title} {\bibinfo {title} {Isometric tensor network
  states in two dimensions},\ }\href@noop {} {\bibfield  {journal} {\bibinfo
  {journal} {Physical review letters}\ }\textbf {\bibinfo {volume} {124}},\
  \bibinfo {pages} {037201} (\bibinfo {year} {2020})}\BibitemShut {NoStop}%
\bibitem [{\citenamefont {Haghshenas}\ \emph {et~al.}(2019)\citenamefont
  {Haghshenas}, \citenamefont {O'Rourke},\ and\ \citenamefont
  {Chan}}]{haghshenas2019conversion}%
  \BibitemOpen
  \bibfield  {author} {\bibinfo {author} {\bibfnamefont {R.}~\bibnamefont
  {Haghshenas}}, \bibinfo {author} {\bibfnamefont {M.~J.}\ \bibnamefont
  {O'Rourke}},\ and\ \bibinfo {author} {\bibfnamefont {G.~K.-L.}\ \bibnamefont
  {Chan}},\ }\bibfield  {title} {\bibinfo {title} {Conversion of projected
  entangled pair states into a canonical form},\ }\href@noop {} {\bibfield
  {journal} {\bibinfo  {journal} {Physical Review B}\ }\textbf {\bibinfo
  {volume} {100}},\ \bibinfo {pages} {054404} (\bibinfo {year}
  {2019})}\BibitemShut {NoStop}%
\bibitem [{\citenamefont {Hyatt}\ and\ \citenamefont
  {Stoudenmire}(2019)}]{hyatt2019dmrg}%
  \BibitemOpen
  \bibfield  {author} {\bibinfo {author} {\bibfnamefont {K.}~\bibnamefont
  {Hyatt}}\ and\ \bibinfo {author} {\bibfnamefont {E.~M.}\ \bibnamefont
  {Stoudenmire}},\ }\bibfield  {title} {\bibinfo {title} {Dmrg approach to
  optimizing two-dimensional tensor networks},\ }\href@noop {} {\bibfield
  {journal} {\bibinfo  {journal} {arXiv preprint arXiv:1908.08833}\ } (\bibinfo
  {year} {2019})}\BibitemShut {NoStop}%
\bibitem [{\citenamefont {Tepaske}\ and\ \citenamefont
  {Luitz}(2021)}]{tepaske20213D}%
  \BibitemOpen
  \bibfield  {author} {\bibinfo {author} {\bibfnamefont {M.~S.~J.}\
  \bibnamefont {Tepaske}}\ and\ \bibinfo {author} {\bibfnamefont {D.~J.}\
  \bibnamefont {Luitz}},\ }\bibfield  {title} {\bibinfo {title}
  {Three-dimensional isometric tensor networks},\ }\bibfield  {journal}
  {\bibinfo  {journal} {Physical Review Research}\ }\textbf {\bibinfo {volume}
  {3}},\ \href {https://doi.org/10.1103/physrevresearch.3.023236}
  {10.1103/physrevresearch.3.023236} (\bibinfo {year} {2021})\BibitemShut
  {NoStop}%
\bibitem [{\citenamefont {Soejima}\ \emph {et~al.}(2020)\citenamefont
  {Soejima}, \citenamefont {Siva}, \citenamefont {Bultinck}, \citenamefont
  {Chatterjee}, \citenamefont {Pollmann}, \citenamefont {Zaletel} \emph
  {et~al.}}]{soejima2020isometric}%
  \BibitemOpen
  \bibfield  {author} {\bibinfo {author} {\bibfnamefont {T.}~\bibnamefont
  {Soejima}}, \bibinfo {author} {\bibfnamefont {K.}~\bibnamefont {Siva}},
  \bibinfo {author} {\bibfnamefont {N.}~\bibnamefont {Bultinck}}, \bibinfo
  {author} {\bibfnamefont {S.}~\bibnamefont {Chatterjee}}, \bibinfo {author}
  {\bibfnamefont {F.}~\bibnamefont {Pollmann}}, \bibinfo {author}
  {\bibfnamefont {M.~P.}\ \bibnamefont {Zaletel}}, \emph {et~al.},\ }\bibfield
  {title} {\bibinfo {title} {Isometric tensor network representation of
  string-net liquids},\ }\href@noop {} {\bibfield  {journal} {\bibinfo
  {journal} {Physical Review B}\ }\textbf {\bibinfo {volume} {101}},\ \bibinfo
  {pages} {085117} (\bibinfo {year} {2020})}\BibitemShut {NoStop}%
\bibitem [{Note2()}]{Note2}%
  \BibitemOpen
  \bibinfo {note} {The name comes from the biblical story of Moses splitting
  the Red Sea.}\BibitemShut {Stop}%
\bibitem [{Note3()}]{Note3}%
  \BibitemOpen
  \bibinfo {note} {Strictly speaking, it is not a QR decomposition as we do not
  require $\mathinner {|{\Phi }\rangle }$ to be upper-triangular.}\BibitemShut
  {Stop}%
\bibitem [{\citenamefont {Lin}\ \emph {et~al.}(2022)\citenamefont {Lin},
  \citenamefont {Zaletel},\ and\ \citenamefont {Pollmann}}]{lin2022efficient}%
  \BibitemOpen
  \bibfield  {author} {\bibinfo {author} {\bibfnamefont {S.-H.}\ \bibnamefont
  {Lin}}, \bibinfo {author} {\bibfnamefont {M.~P.}\ \bibnamefont {Zaletel}},\
  and\ \bibinfo {author} {\bibfnamefont {F.}~\bibnamefont {Pollmann}},\
  }\bibfield  {title} {\bibinfo {title} {Efficient simulation of dynamics in
  two-dimensional quantum spin systems with isometric tensor networks},\
  }\href@noop {} {\bibfield  {journal} {\bibinfo  {journal} {Physical Review
  B}\ }\textbf {\bibinfo {volume} {106}},\ \bibinfo {pages} {245102} (\bibinfo
  {year} {2022})}\BibitemShut {NoStop}%
\bibitem [{\citenamefont {Osorio~Iregui}\ \emph {et~al.}(2017)\citenamefont
  {Osorio~Iregui}, \citenamefont {Troyer},\ and\ \citenamefont
  {Corboz}}]{iregui2017cylinders}%
  \BibitemOpen
  \bibfield  {author} {\bibinfo {author} {\bibfnamefont {J.}~\bibnamefont
  {Osorio~Iregui}}, \bibinfo {author} {\bibfnamefont {M.}~\bibnamefont
  {Troyer}},\ and\ \bibinfo {author} {\bibfnamefont {P.}~\bibnamefont
  {Corboz}},\ }\bibfield  {title} {\bibinfo {title} {Infinite matrix product
  states versus infinite projected entangled-pair states on the cylinder: A
  comparative study},\ }\bibfield  {journal} {\bibinfo  {journal} {Physical
  Review B}\ }\textbf {\bibinfo {volume} {96}},\ \href
  {https://doi.org/10.1103/physrevb.96.115113} {10.1103/physrevb.96.115113}
  (\bibinfo {year} {2017})\BibitemShut {NoStop}%
\bibitem [{\citenamefont {Motruk}\ \emph {et~al.}(2016)\citenamefont {Motruk},
  \citenamefont {Zaletel}, \citenamefont {Mong},\ and\ \citenamefont
  {Pollmann}}]{motruk2016xk}%
  \BibitemOpen
  \bibfield  {author} {\bibinfo {author} {\bibfnamefont {J.}~\bibnamefont
  {Motruk}}, \bibinfo {author} {\bibfnamefont {M.~P.}\ \bibnamefont {Zaletel}},
  \bibinfo {author} {\bibfnamefont {R.~S.~K.}\ \bibnamefont {Mong}},\ and\
  \bibinfo {author} {\bibfnamefont {F.}~\bibnamefont {Pollmann}},\ }\bibfield
  {title} {\bibinfo {title} {Density matrix renormalization group on a cylinder
  in mixed real and momentum space},\ }\bibfield  {journal} {\bibinfo
  {journal} {Physical Review B}\ }\textbf {\bibinfo {volume} {93}},\ \href
  {https://doi.org/10.1103/physrevb.93.155139} {10.1103/physrevb.93.155139}
  (\bibinfo {year} {2016})\BibitemShut {NoStop}%
\bibitem [{Note4()}]{Note4}%
  \BibitemOpen
  \bibinfo {note} {It is possible to have more than one OCs in the TNS. In that
  case, the reduced density matrices of the OCs will be separable}\BibitemShut
  {NoStop}%
\bibitem [{Note5()}]{Note5}%
  \BibitemOpen
  \bibinfo {note} {The complexity is $\protect \mathcal {O}(\chi ^{12})$ for
  DMRG\protect \textsuperscript {2} with explicit $H_\protect \text {eff}$
  construction and $\protect \mathcal {O}(\chi ^{10})$ utilizing the sparse
  structure.}\BibitemShut {Stop}%
\bibitem [{Note6()}]{Note6}%
  \BibitemOpen
  \bibinfo {note} {The reason to take the real part in Eq.~\protect \textup
  {\hbox {\mathsurround \z@ \protect \normalfont (\ignorespaces \ref
  {eq:iMM_polar}\unskip \@@italiccorr )}} is that for any normalized states
  $\mathinner {|{x}\rangle }$ and $\mathinner {|{y}\rangle }$, $\norm
  {\mathinner {|{x}\rangle }-\mathinner {|{y}\rangle }}^2 = 2 (1-\protect \text
  {Re}\mathinner {\langle {x|y}\rangle })$.}\BibitemShut {Stop}%
\bibitem [{\citenamefont {Evenbly}\ and\ \citenamefont
  {Vidal}(2009)}]{evenbly2009MERA}%
  \BibitemOpen
  \bibfield  {author} {\bibinfo {author} {\bibfnamefont {G.}~\bibnamefont
  {Evenbly}}\ and\ \bibinfo {author} {\bibfnamefont {G.}~\bibnamefont
  {Vidal}},\ }\bibfield  {title} {\bibinfo {title} {Algorithms for entanglement
  renormalization},\ }\bibfield  {journal} {\bibinfo  {journal} {Physical
  Review B}\ }\textbf {\bibinfo {volume} {79}},\ \href
  {https://doi.org/10.1103/physrevb.79.144108} {10.1103/physrevb.79.144108}
  (\bibinfo {year} {2009})\BibitemShut {NoStop}%
\bibitem [{\citenamefont {Vanhecke}\ \emph {et~al.}(2021)\citenamefont
  {Vanhecke}, \citenamefont {Van~Damme}, \citenamefont {Haegeman},
  \citenamefont {Vanderstraeten},\ and\ \citenamefont
  {Verstraete}}]{vanhecke2021tangent}%
  \BibitemOpen
  \bibfield  {author} {\bibinfo {author} {\bibfnamefont {B.}~\bibnamefont
  {Vanhecke}}, \bibinfo {author} {\bibfnamefont {M.}~\bibnamefont {Van~Damme}},
  \bibinfo {author} {\bibfnamefont {J.}~\bibnamefont {Haegeman}}, \bibinfo
  {author} {\bibfnamefont {L.}~\bibnamefont {Vanderstraeten}},\ and\ \bibinfo
  {author} {\bibfnamefont {F.}~\bibnamefont {Verstraete}},\ }\bibfield  {title}
  {\bibinfo {title} {Tangent-space methods for truncating uniform mps},\
  }\href@noop {} {\bibfield  {journal} {\bibinfo  {journal} {SciPost Physics
  Core}\ }\textbf {\bibinfo {volume} {4}},\ \bibinfo {pages} {004} (\bibinfo
  {year} {2021})}\BibitemShut {NoStop}%
\bibitem [{\citenamefont {Hauru}\ \emph {et~al.}(2021)\citenamefont {Hauru},
  \citenamefont {Van~Damme},\ and\ \citenamefont
  {Haegeman}}]{hauru2021riemannian}%
  \BibitemOpen
  \bibfield  {author} {\bibinfo {author} {\bibfnamefont {M.}~\bibnamefont
  {Hauru}}, \bibinfo {author} {\bibfnamefont {M.}~\bibnamefont {Van~Damme}},\
  and\ \bibinfo {author} {\bibfnamefont {J.}~\bibnamefont {Haegeman}},\
  }\bibfield  {title} {\bibinfo {title} {Riemannian optimization of isometric
  tensor networks},\ }\href@noop {} {\bibfield  {journal} {\bibinfo  {journal}
  {SciPost Physics}\ }\textbf {\bibinfo {volume} {10}},\ \bibinfo {pages} {040}
  (\bibinfo {year} {2021})}\BibitemShut {NoStop}%
\bibitem [{\citenamefont {Luchnikov}\ \emph {et~al.}(2021)\citenamefont
  {Luchnikov}, \citenamefont {Ryzhov}, \citenamefont {Filippov},\ and\
  \citenamefont {Ouerdane}}]{luchnikov2021qgopt}%
  \BibitemOpen
  \bibfield  {author} {\bibinfo {author} {\bibfnamefont {I.}~\bibnamefont
  {Luchnikov}}, \bibinfo {author} {\bibfnamefont {A.}~\bibnamefont {Ryzhov}},
  \bibinfo {author} {\bibfnamefont {S.}~\bibnamefont {Filippov}},\ and\
  \bibinfo {author} {\bibfnamefont {H.}~\bibnamefont {Ouerdane}},\ }\bibfield
  {title} {\bibinfo {title} {Qgopt: Riemannian optimization for quantum
  technologies},\ }\href@noop {} {\bibfield  {journal} {\bibinfo  {journal}
  {SciPost Physics}\ }\textbf {\bibinfo {volume} {10}},\ \bibinfo {pages} {079}
  (\bibinfo {year} {2021})}\BibitemShut {NoStop}%
\bibitem [{\citenamefont {Xie}\ \emph {et~al.}(2020)\citenamefont {Xie},
  \citenamefont {Liu},\ and\ \citenamefont {Wang}}]{xie2020eigensolver}%
  \BibitemOpen
  \bibfield  {author} {\bibinfo {author} {\bibfnamefont {H.}~\bibnamefont
  {Xie}}, \bibinfo {author} {\bibfnamefont {J.-G.}\ \bibnamefont {Liu}},\ and\
  \bibinfo {author} {\bibfnamefont {L.}~\bibnamefont {Wang}},\ }\bibfield
  {title} {\bibinfo {title} {Automatic differentiation of dominant eigensolver
  and its applications in quantum physics},\ }\href
  {https://doi.org/10.1103/PhysRevB.101.245139} {\bibfield  {journal} {\bibinfo
   {journal} {Phys. Rev. B}\ }\textbf {\bibinfo {volume} {101}},\ \bibinfo
  {pages} {245139} (\bibinfo {year} {2020})}\BibitemShut {NoStop}%
\bibitem [{Note7()}]{Note7}%
  \BibitemOpen
  \bibinfo {note} {Here we treat the real and imaginary part of $X_a$ as
  independent variables. An alternative way to arrive at the same result is to
  treat $X_a$ as complex-valued variables and the ascent direction is given by
  ${dO}/{d \protect \overline {X_a}} = \protect \frac {1}{2} R^\dagger
  (dT^\dagger /{d \protect \overline {X_a}}) L = \protect \frac {1}{2}\protect
  \overline {E_{X_a}}$.}\BibitemShut {Stop}%
\bibitem [{\citenamefont {Bl\"ote}\ and\ \citenamefont
  {Deng}(2002)}]{blote2002TFI}%
  \BibitemOpen
  \bibfield  {author} {\bibinfo {author} {\bibfnamefont {H.~W.~J.}\
  \bibnamefont {Bl\"ote}}\ and\ \bibinfo {author} {\bibfnamefont
  {Y.}~\bibnamefont {Deng}},\ }\bibfield  {title} {\bibinfo {title} {Cluster
  monte carlo simulation of the transverse ising model},\ }\href
  {https://doi.org/10.1103/PhysRevE.66.066110} {\bibfield  {journal} {\bibinfo
  {journal} {Phys. Rev. E}\ }\textbf {\bibinfo {volume} {66}},\ \bibinfo
  {pages} {066110} (\bibinfo {year} {2002})}\BibitemShut {NoStop}%
\bibitem [{\citenamefont {Hashizume}\ \emph {et~al.}(2022)\citenamefont
  {Hashizume}, \citenamefont {McCulloch},\ and\ \citenamefont
  {Halimeh}}]{hashizume2022TFI}%
  \BibitemOpen
  \bibfield  {author} {\bibinfo {author} {\bibfnamefont {T.}~\bibnamefont
  {Hashizume}}, \bibinfo {author} {\bibfnamefont {I.~P.}\ \bibnamefont
  {McCulloch}},\ and\ \bibinfo {author} {\bibfnamefont {J.~C.}\ \bibnamefont
  {Halimeh}},\ }\bibfield  {title} {\bibinfo {title} {Dynamical phase
  transitions in the two-dimensional transverse-field ising model},\ }\href
  {https://doi.org/10.1103/PhysRevResearch.4.013250} {\bibfield  {journal}
  {\bibinfo  {journal} {Phys. Rev. Research}\ }\textbf {\bibinfo {volume}
  {4}},\ \bibinfo {pages} {013250} (\bibinfo {year} {2022})}\BibitemShut
  {NoStop}%
\bibitem [{\citenamefont {Hauschild}\ \emph {et~al.}(2018)\citenamefont
  {Hauschild}, \citenamefont {Leviatan}, \citenamefont {Bardarson},
  \citenamefont {Altman}, \citenamefont {Zaletel},\ and\ \citenamefont
  {Pollmann}}]{hauschild2018finding}%
  \BibitemOpen
  \bibfield  {author} {\bibinfo {author} {\bibfnamefont {J.}~\bibnamefont
  {Hauschild}}, \bibinfo {author} {\bibfnamefont {E.}~\bibnamefont {Leviatan}},
  \bibinfo {author} {\bibfnamefont {J.~H.}\ \bibnamefont {Bardarson}}, \bibinfo
  {author} {\bibfnamefont {E.}~\bibnamefont {Altman}}, \bibinfo {author}
  {\bibfnamefont {M.~P.}\ \bibnamefont {Zaletel}},\ and\ \bibinfo {author}
  {\bibfnamefont {F.}~\bibnamefont {Pollmann}},\ }\bibfield  {title} {\bibinfo
  {title} {Finding purifications with minimal entanglement},\ }\href@noop {}
  {\bibfield  {journal} {\bibinfo  {journal} {Physical Review B}\ }\textbf
  {\bibinfo {volume} {98}},\ \bibinfo {pages} {235163} (\bibinfo {year}
  {2018})}\BibitemShut {NoStop}%
\bibitem [{\citenamefont {Lubasch}\ \emph
  {et~al.}(2014{\natexlab{b}})\citenamefont {Lubasch}, \citenamefont {Cirac},\
  and\ \citenamefont {Banuls}}]{lubasch2014unifying}%
  \BibitemOpen
  \bibfield  {author} {\bibinfo {author} {\bibfnamefont {M.}~\bibnamefont
  {Lubasch}}, \bibinfo {author} {\bibfnamefont {J.~I.}\ \bibnamefont {Cirac}},\
  and\ \bibinfo {author} {\bibfnamefont {M.-C.}\ \bibnamefont {Banuls}},\
  }\bibfield  {title} {\bibinfo {title} {Unifying projected entangled pair
  state contractions},\ }\href@noop {} {\bibfield  {journal} {\bibinfo
  {journal} {New Journal of Physics}\ }\textbf {\bibinfo {volume} {16}},\
  \bibinfo {pages} {033014} (\bibinfo {year} {2014}{\natexlab{b}})}\BibitemShut
  {NoStop}%
\bibitem [{\citenamefont {Ran}\ \emph {et~al.}(2020)\citenamefont {Ran},
  \citenamefont {Tirrito}, \citenamefont {Peng}, \citenamefont {Chen},
  \citenamefont {Tagliacozzo}, \citenamefont {Su},\ and\ \citenamefont
  {Lewenstein}}]{ran2020tensor}%
  \BibitemOpen
  \bibfield  {author} {\bibinfo {author} {\bibfnamefont {S.-J.}\ \bibnamefont
  {Ran}}, \bibinfo {author} {\bibfnamefont {E.}~\bibnamefont {Tirrito}},
  \bibinfo {author} {\bibfnamefont {C.}~\bibnamefont {Peng}}, \bibinfo {author}
  {\bibfnamefont {X.}~\bibnamefont {Chen}}, \bibinfo {author} {\bibfnamefont
  {L.}~\bibnamefont {Tagliacozzo}}, \bibinfo {author} {\bibfnamefont
  {G.}~\bibnamefont {Su}},\ and\ \bibinfo {author} {\bibfnamefont
  {M.}~\bibnamefont {Lewenstein}},\ }\href@noop {} {\emph {\bibinfo {title}
  {Tensor network contractions: methods and applications to quantum many-body
  systems}}}\ (\bibinfo  {publisher} {Springer Nature},\ \bibinfo {year}
  {2020})\BibitemShut {NoStop}%
\bibitem [{\citenamefont {Vidal}(2007)}]{vidal2007classical}%
  \BibitemOpen
  \bibfield  {author} {\bibinfo {author} {\bibfnamefont {G.}~\bibnamefont
  {Vidal}},\ }\bibfield  {title} {\bibinfo {title} {Classical simulation of
  infinite-size quantum lattice systems in one spatial dimension},\ }\href
  {https://doi.org/10.1103/PhysRevLett.98.070201} {\bibfield  {journal}
  {\bibinfo  {journal} {Phys. Rev. Lett.}\ }\textbf {\bibinfo {volume} {98}},\
  \bibinfo {pages} {070201} (\bibinfo {year} {2007})}\BibitemShut {NoStop}%
\bibitem [{\citenamefont {Or\'us}\ and\ \citenamefont
  {Vidal}(2008)}]{orus2008iTEBD}%
  \BibitemOpen
  \bibfield  {author} {\bibinfo {author} {\bibfnamefont {R.}~\bibnamefont
  {Or\'us}}\ and\ \bibinfo {author} {\bibfnamefont {G.}~\bibnamefont {Vidal}},\
  }\bibfield  {title} {\bibinfo {title} {Infinite time-evolving block
  decimation algorithm beyond unitary evolution},\ }\href
  {https://doi.org/10.1103/PhysRevB.78.155117} {\bibfield  {journal} {\bibinfo
  {journal} {Phys. Rev. B}\ }\textbf {\bibinfo {volume} {78}},\ \bibinfo
  {pages} {155117} (\bibinfo {year} {2008})}\BibitemShut {NoStop}%
\bibitem [{\citenamefont {Jiang}\ \emph {et~al.}(2008)\citenamefont {Jiang},
  \citenamefont {Weng},\ and\ \citenamefont {Xiang}}]{jiang2008simple}%
  \BibitemOpen
  \bibfield  {author} {\bibinfo {author} {\bibfnamefont {H.~C.}\ \bibnamefont
  {Jiang}}, \bibinfo {author} {\bibfnamefont {Z.~Y.}\ \bibnamefont {Weng}},\
  and\ \bibinfo {author} {\bibfnamefont {T.}~\bibnamefont {Xiang}},\ }\bibfield
   {title} {\bibinfo {title} {Accurate determination of tensor network state of
  quantum lattice models in two dimensions},\ }\bibfield  {journal} {\bibinfo
  {journal} {Physical Review Letters}\ }\textbf {\bibinfo {volume} {101}},\
  \href {https://doi.org/10.1103/physrevlett.101.090603}
  {10.1103/physrevlett.101.090603} (\bibinfo {year} {2008})\BibitemShut
  {NoStop}%
\bibitem [{\citenamefont {Jordan}\ \emph {et~al.}(2008)\citenamefont {Jordan},
  \citenamefont {Orús}, \citenamefont {Vidal}, \citenamefont {Verstraete},\
  and\ \citenamefont {Cirac}}]{jordan2008full}%
  \BibitemOpen
  \bibfield  {author} {\bibinfo {author} {\bibfnamefont {J.}~\bibnamefont
  {Jordan}}, \bibinfo {author} {\bibfnamefont {R.}~\bibnamefont {Orús}},
  \bibinfo {author} {\bibfnamefont {G.}~\bibnamefont {Vidal}}, \bibinfo
  {author} {\bibfnamefont {F.}~\bibnamefont {Verstraete}},\ and\ \bibinfo
  {author} {\bibfnamefont {J.~I.}\ \bibnamefont {Cirac}},\ }\bibfield  {title}
  {\bibinfo {title} {Classical simulation of infinite-size quantum lattice
  systems in two spatial dimensions},\ }\bibfield  {journal} {\bibinfo
  {journal} {Physical Review Letters}\ }\textbf {\bibinfo {volume} {101}},\
  \href {https://doi.org/10.1103/physrevlett.101.250602}
  {10.1103/physrevlett.101.250602} (\bibinfo {year} {2008})\BibitemShut
  {NoStop}%
\bibitem [{Note8()}]{Note8}%
  \BibitemOpen
  \bibinfo {note} {After merging two columns, the merged column would have bond
  dimension $\eta \chi $. Optionally, a compression can be performed after
  merging two columns and before the iTEBD. In this work, we compress the bond
  dimension from $\eta \chi $ down to $\eta $ before performing the iTEBD. If a
  compression to bond dimension $\eta $ is performed before the iTEBD, the
  complexity is $\protect \mathcal {O}(\chi ^4 \eta ^3)$; otherwise, the
  complexity is slightly higher, $\protect \mathcal {O}(\chi ^5\eta ^3)$. In
  practice, we observe a slight deterioration in the result with compression.
  But we can obtain better results overall by using larger bond dimensions
  within the same run-time.}\BibitemShut {Stop}%
\bibitem [{\citenamefont {Bauer}\ \emph {et~al.}(2011)\citenamefont {Bauer},
  \citenamefont {Carr}, \citenamefont {Evertz}, \citenamefont {Feiguin},
  \citenamefont {Freire}, \citenamefont {Fuchs}, \citenamefont {Gamper},
  \citenamefont {Gukelberger}, \citenamefont {Gull}, \citenamefont {Guertler}
  \emph {et~al.}}]{bauer2011alps}%
  \BibitemOpen
  \bibfield  {author} {\bibinfo {author} {\bibfnamefont {B.}~\bibnamefont
  {Bauer}}, \bibinfo {author} {\bibfnamefont {L.}~\bibnamefont {Carr}},
  \bibinfo {author} {\bibfnamefont {H.~G.}\ \bibnamefont {Evertz}}, \bibinfo
  {author} {\bibfnamefont {A.}~\bibnamefont {Feiguin}}, \bibinfo {author}
  {\bibfnamefont {J.}~\bibnamefont {Freire}}, \bibinfo {author} {\bibfnamefont
  {S.}~\bibnamefont {Fuchs}}, \bibinfo {author} {\bibfnamefont
  {L.}~\bibnamefont {Gamper}}, \bibinfo {author} {\bibfnamefont
  {J.}~\bibnamefont {Gukelberger}}, \bibinfo {author} {\bibfnamefont
  {E.}~\bibnamefont {Gull}}, \bibinfo {author} {\bibfnamefont {S.}~\bibnamefont
  {Guertler}}, \emph {et~al.},\ }\bibfield  {title} {\bibinfo {title} {The alps
  project release 2.0: open source software for strongly correlated systems},\
  }\href@noop {} {\bibfield  {journal} {\bibinfo  {journal} {Journal of
  Statistical Mechanics: Theory and Experiment}\ }\textbf {\bibinfo {volume}
  {2011}},\ \bibinfo {pages} {P05001} (\bibinfo {year} {2011})}\BibitemShut
  {NoStop}%
\bibitem [{\citenamefont {Sch{\"o}n}\ \emph {et~al.}(2005)\citenamefont
  {Sch{\"o}n}, \citenamefont {Solano}, \citenamefont {Verstraete},
  \citenamefont {Cirac},\ and\ \citenamefont {Wolf}}]{schon2005sequential}%
  \BibitemOpen
  \bibfield  {author} {\bibinfo {author} {\bibfnamefont {C.}~\bibnamefont
  {Sch{\"o}n}}, \bibinfo {author} {\bibfnamefont {E.}~\bibnamefont {Solano}},
  \bibinfo {author} {\bibfnamefont {F.}~\bibnamefont {Verstraete}}, \bibinfo
  {author} {\bibfnamefont {J.~I.}\ \bibnamefont {Cirac}},\ and\ \bibinfo
  {author} {\bibfnamefont {M.~M.}\ \bibnamefont {Wolf}},\ }\bibfield  {title}
  {\bibinfo {title} {Sequential generation of entangled multiqubit states},\
  }\href@noop {} {\bibfield  {journal} {\bibinfo  {journal} {Physical review
  letters}\ }\textbf {\bibinfo {volume} {95}},\ \bibinfo {pages} {110503}
  (\bibinfo {year} {2005})}\BibitemShut {NoStop}%
\bibitem [{\citenamefont {Banuls}\ \emph {et~al.}(2008)\citenamefont {Banuls},
  \citenamefont {P{\'e}rez-Garc{\'\i}a}, \citenamefont {Wolf}, \citenamefont
  {Verstraete},\ and\ \citenamefont {Cirac}}]{banuls2008sequentially}%
  \BibitemOpen
  \bibfield  {author} {\bibinfo {author} {\bibfnamefont {M.-C.}\ \bibnamefont
  {Banuls}}, \bibinfo {author} {\bibfnamefont {D.}~\bibnamefont
  {P{\'e}rez-Garc{\'\i}a}}, \bibinfo {author} {\bibfnamefont {M.~M.}\
  \bibnamefont {Wolf}}, \bibinfo {author} {\bibfnamefont {F.}~\bibnamefont
  {Verstraete}},\ and\ \bibinfo {author} {\bibfnamefont {J.~I.}\ \bibnamefont
  {Cirac}},\ }\bibfield  {title} {\bibinfo {title} {Sequentially generated
  states for the study of two-dimensional systems},\ }\href@noop {} {\bibfield
  {journal} {\bibinfo  {journal} {Physical Review A}\ }\textbf {\bibinfo
  {volume} {77}},\ \bibinfo {pages} {052306} (\bibinfo {year}
  {2008})}\BibitemShut {NoStop}%
\bibitem [{\citenamefont {Foss-Feig}\ \emph {et~al.}(2021)\citenamefont
  {Foss-Feig}, \citenamefont {Hayes}, \citenamefont {Dreiling}, \citenamefont
  {Figgatt}, \citenamefont {Gaebler}, \citenamefont {Moses}, \citenamefont
  {Pino},\ and\ \citenamefont {Potter}}]{foss-feig2021holographic}%
  \BibitemOpen
  \bibfield  {author} {\bibinfo {author} {\bibfnamefont {M.}~\bibnamefont
  {Foss-Feig}}, \bibinfo {author} {\bibfnamefont {D.}~\bibnamefont {Hayes}},
  \bibinfo {author} {\bibfnamefont {J.~M.}\ \bibnamefont {Dreiling}}, \bibinfo
  {author} {\bibfnamefont {C.}~\bibnamefont {Figgatt}}, \bibinfo {author}
  {\bibfnamefont {J.~P.}\ \bibnamefont {Gaebler}}, \bibinfo {author}
  {\bibfnamefont {S.~A.}\ \bibnamefont {Moses}}, \bibinfo {author}
  {\bibfnamefont {J.~M.}\ \bibnamefont {Pino}},\ and\ \bibinfo {author}
  {\bibfnamefont {A.~C.}\ \bibnamefont {Potter}},\ }\bibfield  {title}
  {\bibinfo {title} {Holographic quantum algorithms for simulating correlated
  spin systems},\ }\bibfield  {journal} {\bibinfo  {journal} {Physical Review
  Research}\ }\textbf {\bibinfo {volume} {3}},\ \href
  {https://doi.org/10.1103/physrevresearch.3.033002}
  {10.1103/physrevresearch.3.033002} (\bibinfo {year} {2021})\BibitemShut
  {NoStop}%
\bibitem [{\citenamefont {Haghshenas}\ \emph {et~al.}(2022)\citenamefont
  {Haghshenas}, \citenamefont {Gray}, \citenamefont {Potter},\ and\
  \citenamefont {Chan}}]{haghshenas2022variational}%
  \BibitemOpen
  \bibfield  {author} {\bibinfo {author} {\bibfnamefont {R.}~\bibnamefont
  {Haghshenas}}, \bibinfo {author} {\bibfnamefont {J.}~\bibnamefont {Gray}},
  \bibinfo {author} {\bibfnamefont {A.~C.}\ \bibnamefont {Potter}},\ and\
  \bibinfo {author} {\bibfnamefont {G.~K.-L.}\ \bibnamefont {Chan}},\
  }\bibfield  {title} {\bibinfo {title} {Variational power of quantum circuit
  tensor networks},\ }\href@noop {} {\bibfield  {journal} {\bibinfo  {journal}
  {Physical Review X}\ }\textbf {\bibinfo {volume} {12}},\ \bibinfo {pages}
  {011047} (\bibinfo {year} {2022})}\BibitemShut {NoStop}%
\bibitem [{\citenamefont {Slattery}\ and\ \citenamefont
  {Clark}(2021)}]{slattery2021quantum}%
  \BibitemOpen
  \bibfield  {author} {\bibinfo {author} {\bibfnamefont {L.}~\bibnamefont
  {Slattery}}\ and\ \bibinfo {author} {\bibfnamefont {B.~K.}\ \bibnamefont
  {Clark}},\ }\href@noop {} {\bibinfo {title} {Quantum circuits for
  two-dimensional isometric tensor networks}} (\bibinfo {year} {2021}),\
  \Eprint {https://arxiv.org/abs/2108.02792} {arXiv:2108.02792 [quant-ph]}
  \BibitemShut {NoStop}%
\bibitem [{\citenamefont {Wei}\ \emph {et~al.}(2022)\citenamefont {Wei},
  \citenamefont {Malz},\ and\ \citenamefont {Cirac}}]{wei2022sequential}%
  \BibitemOpen
  \bibfield  {author} {\bibinfo {author} {\bibfnamefont {Z.-Y.}\ \bibnamefont
  {Wei}}, \bibinfo {author} {\bibfnamefont {D.}~\bibnamefont {Malz}},\ and\
  \bibinfo {author} {\bibfnamefont {J.~I.}\ \bibnamefont {Cirac}},\ }\bibfield
  {title} {\bibinfo {title} {Sequential generation of projected entangled-pair
  states},\ }\href {https://doi.org/10.1103/PhysRevLett.128.010607} {\bibfield
  {journal} {\bibinfo  {journal} {Phys. Rev. Lett.}\ }\textbf {\bibinfo
  {volume} {128}},\ \bibinfo {pages} {010607} (\bibinfo {year}
  {2022})}\BibitemShut {NoStop}%
\bibitem [{\citenamefont {Barratt}\ \emph {et~al.}(2021)\citenamefont
  {Barratt}, \citenamefont {Dborin}, \citenamefont {Bal}, \citenamefont
  {Stojevic}, \citenamefont {Pollmann},\ and\ \citenamefont
  {Green}}]{Barratt2021}%
  \BibitemOpen
  \bibfield  {author} {\bibinfo {author} {\bibfnamefont {F.}~\bibnamefont
  {Barratt}}, \bibinfo {author} {\bibfnamefont {J.}~\bibnamefont {Dborin}},
  \bibinfo {author} {\bibfnamefont {M.}~\bibnamefont {Bal}}, \bibinfo {author}
  {\bibfnamefont {V.}~\bibnamefont {Stojevic}}, \bibinfo {author}
  {\bibfnamefont {F.}~\bibnamefont {Pollmann}},\ and\ \bibinfo {author}
  {\bibfnamefont {A.~G.}\ \bibnamefont {Green}},\ }\bibfield  {title} {\bibinfo
  {title} {Parallel quantum simulation of large systems on small nisq
  computers},\ }\href {https://doi.org/10.1038/s41534-021-00420-3} {\bibfield
  {journal} {\bibinfo  {journal} {npj Quantum Information}\ }\textbf {\bibinfo
  {volume} {7}},\ \bibinfo {pages} {79} (\bibinfo {year} {2021})}\BibitemShut
  {NoStop}%
\bibitem [{\citenamefont {Dborin}\ \emph {et~al.}(2022)\citenamefont {Dborin},
  \citenamefont {Wimalaweera}, \citenamefont {Barratt}, \citenamefont {Ostby},
  \citenamefont {O'Brien},\ and\ \citenamefont {Green}}]{Dborin2022}%
  \BibitemOpen
  \bibfield  {author} {\bibinfo {author} {\bibfnamefont {J.}~\bibnamefont
  {Dborin}}, \bibinfo {author} {\bibfnamefont {V.}~\bibnamefont {Wimalaweera}},
  \bibinfo {author} {\bibfnamefont {F.}~\bibnamefont {Barratt}}, \bibinfo
  {author} {\bibfnamefont {E.}~\bibnamefont {Ostby}}, \bibinfo {author}
  {\bibfnamefont {T.~E.}\ \bibnamefont {O'Brien}},\ and\ \bibinfo {author}
  {\bibfnamefont {A.~G.}\ \bibnamefont {Green}},\ }\bibfield  {title} {\bibinfo
  {title} {Simulating groundstate and dynamical quantum phase transitions on a
  superconducting quantum computer},\ }\href
  {https://doi.org/10.1038/s41467-022-33737-4} {\bibfield  {journal} {\bibinfo
  {journal} {Nature Communications}\ }\textbf {\bibinfo {volume} {13}},\
  \bibinfo {pages} {5977} (\bibinfo {year} {2022})}\BibitemShut {NoStop}%
\bibitem [{\citenamefont {Astrakhantsev}\ \emph {et~al.}(2022)\citenamefont
  {Astrakhantsev}, \citenamefont {Lin}, \citenamefont {Pollmann},\ and\
  \citenamefont {Smith}}]{astrakhantsev2022time}%
  \BibitemOpen
  \bibfield  {author} {\bibinfo {author} {\bibfnamefont {N.}~\bibnamefont
  {Astrakhantsev}}, \bibinfo {author} {\bibfnamefont {S.-H.}\ \bibnamefont
  {Lin}}, \bibinfo {author} {\bibfnamefont {F.}~\bibnamefont {Pollmann}},\ and\
  \bibinfo {author} {\bibfnamefont {A.}~\bibnamefont {Smith}},\ }\bibfield
  {title} {\bibinfo {title} {Time evolution of uniform sequential circuits},\
  }\href@noop {} {\bibfield  {journal} {\bibinfo  {journal} {arXiv preprint
  arXiv:2210.03751}\ } (\bibinfo {year} {2022})}\BibitemShut {NoStop}%
\bibitem [{\citenamefont {Anand}\ \emph {et~al.}(2022)\citenamefont {Anand},
  \citenamefont {Hauschild}, \citenamefont {Zhang}, \citenamefont {Potter},\
  and\ \citenamefont {Zaletel}}]{anand2022holographic}%
  \BibitemOpen
  \bibfield  {author} {\bibinfo {author} {\bibfnamefont {S.}~\bibnamefont
  {Anand}}, \bibinfo {author} {\bibfnamefont {J.}~\bibnamefont {Hauschild}},
  \bibinfo {author} {\bibfnamefont {Y.}~\bibnamefont {Zhang}}, \bibinfo
  {author} {\bibfnamefont {A.~C.}\ \bibnamefont {Potter}},\ and\ \bibinfo
  {author} {\bibfnamefont {M.~P.}\ \bibnamefont {Zaletel}},\ }\bibfield
  {title} {\bibinfo {title} {Holographic quantum simulation of entanglement
  renormalization circuits},\ }\href@noop {} {\bibfield  {journal} {\bibinfo
  {journal} {arXiv preprint arXiv:2203.00886}\ } (\bibinfo {year}
  {2022})}\BibitemShut {NoStop}%
\bibitem [{Note9()}]{Note9}%
  \BibitemOpen
  \bibinfo {note} {Note that it is wrong to use $T_{A^\protect \dag \Psi
  :\Lambda }$ in Eq.~\protect \textup {\hbox {\mathsurround \z@ \protect
  \normalfont (\ignorespaces \ref {eq:errors2}\unskip \@@italiccorr )}} as the
  states making up the transfer matrix are not properly normalized}\BibitemShut
  {NoStop}%
\bibitem [{Note10()}]{Note10}%
  \BibitemOpen
  \bibinfo {note} {Here by physical quantities, we mean any quantity obtainable
  in principle from the initial wavefunction $\mathinner {|{\Phi _0}\rangle
  }$}\BibitemShut {NoStop}%
\end{thebibliography}%

\end{document}